\documentclass[11pt]{amsart}
\usepackage{latexsym,amsmath,amssymb,amscd,eucal}
\usepackage{xypic}

\newcommand{\newsection}{\setcounter{equation}{0}\section}

\newcommand{\mbf}[1]{{\boldsymbol {#1} }}

\def\appendix#1{\addtocounter{section}{1}\setcounter{equation}{0}
\renewcommand{\thesection}{\Alph{section}}
\section*{Appendix \thesection\protect\indent \parbox[t]{11.715cm} {#1}}
\addcontentsline{toc}{section}{Appendix \thesection\ \ \ #1} }
\newcommand{\eq}{\begin{equation}}
\newcommand{\eqend}{\end{equation}}

\newbox\ncintdbox \newbox\ncinttbox
\setbox0=\hbox{$-$} \setbox2=\hbox{$\displaystyle\int$}
\setbox\ncintdbox=\hbox{\rlap{\hbox to
\wd2{\hskip-.125em\box2\relax \hfil}}\box0\kern.1em}
\setbox0=\hbox{$\vcenter{\hrule width 4pt}$}
\setbox2=\hbox{$\textstyle\int$}
\setbox\ncinttbox=\hbox{\rlap{\hbox to
\wd2{\hskip-.175em\box2\relax \hfil}}\box0\kern.1em}

\def\Dirac{{D\!\!\!\!/\,}} % Dirac operator
\def\={\ =\ }

\def\spinc{spin$^c$~}

\newcommand{\complex}{{\mathbb C}} %% complex numbers
\newcommand{\zed}{{\mathbb Z}} %% integers
\newcommand{\nat}{{\mathbb N}} %% naturals
\newcommand{\real}{{\mathbb R}} %% real numbers
\newcommand{\rat}{{\mathbb Q}} %% rational numbers
\newcommand{\id}{{1\!\!1}} %% identity operator
\def\alg{{\mathcal A}}
\def\balg{{\mathcal B}}
\def\hil{{\mathcal H}}

\def\bun{{\mathcal E}}
\def\lin{{\mathcal L}}
\def\comp{{\mathcal K}}
\def\spinor{{S}}

\def\pt{{\rm pt}}
\def\ev{{\rm ev}}
\def\CS{{\rm WZ}}
\def\cpt{{\rm cpt}}
\def\K{{\rm K}}
\def\H{{\rm H}}
\def\C{{\rm C}}
\def\E{{\rm E}}

\def\B{{\mathbb{B}}}

\def\Ltwo{{\rm L}^2}
\def\S{{\mathbb{S}}}

\def\F{{\rm filt}}
\def\U{{\rm U}}
\def\im{{\rm im}}

\def\Tor{{\rm Tor}}

\def\Hom{{\rm Hom}}

\def\End{{\rm End}}
\def\MSpin{{\rm MSpin}}

\def\Cl{{{\rm C}\ell}}
\def\Fr{{\rm Fr}}
\def\ch{{\rm ch}}

\def\Id{{\rm id}}

\def\Index{{\rm Index}}

\def\Todd{{\rm Todd}}
\def\ind{{\rm ind}}
\def\Fred{{\rm Fred}}
\def\Cliff{{\rm Cliff}}

\def\Thom{{\rm Thom}}

\def\Euler{{\Lambda}}

\newcommand{\Tr}[1]{\:{\rm Tr}\,#1}
\newcommand{\STr}[1]{\:{\rm STr}\,#1}

\def\be{\begin{equation}}
\def\ee{\end{equation}}
\def\bea{\begin{eqnarray}}
\def\eea{\end{eqnarray}}
\def\bd{\begin{displaymath}}
\def\ed{\end{displaymath}}

\def\dd{{\rm d}}

\def\ii{{\,{\rm i}\,}}

\def\Vect{{\rm Vect}}

\makeatletter
\newdimen\normalarrayskip              % skip between lines
\newdimen\minarrayskip                 % minimal skip between lines
\normalarrayskip\baselineskip
\minarrayskip\jot
\newif\ifold             \oldtrue            
\def\arraymode{\ifold\relax\else\displaystyle\fi} % mode of array entries
\def\@arrayskip{\ifold\baselineskip\z@\lineskip\z@
     \else
     \baselineskip\minarrayskip\lineskip2\minarrayskip\fi}
\def\@arrayclassz{\ifcase \@lastchclass \@acolampacol \or
\@ampacol \or \or \or \@addamp \or
   \@acolampacol \or \@firstampfalse \@acol \fi
\edef\@preamble{\@preamble
  \ifcase \@chnum
     \hfil$\relax\arraymode\@sharp$\hfil
     \or $\relax\arraymode\@sharp$\hfil
     \or \hfil$\relax\arraymode\@sharp$\fi}}
\def\@array[#1]#2{\setbox\@arstrutbox=\hbox{\vrule
     height\arraystretch \ht\strutbox
     depth\arraystretch \dp\strutbox
     width\z@}\@mkpream{#2}\edef\@preamble{\halign \noexpand\@halignto
\bgroup \tabskip\z@ \@arstrut \@preamble \tabskip\z@ \cr}%
\let\@startpbox\@@startpbox \let\@endpbox\@@endpbox
  \if #1t\vtop \else \if#1b\vbox \else \vcenter \fi\fi
  \bgroup \let\par\relax
  \let\@sharp##\let\protect\relax
  \@arrayskip\@preamble}
\makeatother

\newcommand{\beq}{\begin{eqnarray}}
\newcommand{\eeq}{\end{eqnarray}}

\newcommand{\KO}{{\rm KO}}

\newcommand{\KK}{{\rm KK}}

\newcommand{\Spin}{{\rm Spin}}
\newcommand{\SO}{{\rm SO}}

\newcommand{\SL}{{\rm SL}}

\newcommand{\ass}{{\rm ass}}

\newcommand{\ocat}[1]{\textsf{Or}(#1)}
\newcommand{\subc}[1]{\textsf{Sub}(#1)}
\newcommand{\cat}[1]{\textsf{#1}}

\newcommand{\lmod}[1]{{#1}\!-\!\text{Mod}}
\newcommand{\rmod}[1]{\text{Mod}\!-\!{#1}}

\def\appendix#1{\addtocounter{section}{1}\setcounter{equation}{0}
\renewcommand{\thesection}{\Alph{section}}
\section*{Appendix \thesection. #1}
\addcontentsline{toc}{section}{Appendix \thesection\ \ \ #1} }

\newcommand{\dslash}{\not{\hbox{\kern-2pt $\partial$}}}
\newcommand{\pslash}{\not{\hbox{\kern-2.3pt $p$}}}

 \newtoks\nslashfraction
 \nslashfraction={.13}
 \newcommand{\nslash}[1]{\setbox0\hbox{$ #1 $}
   \setbox0\hbox to \the\nslashfraction\wd0{\hss \box0}/\box0 }

\def\ii{{\,{\rm i}\,}}

\newtheorem{theorem}{Theorem}[section]
\newtheorem{lemma}[theorem]{Lemma}
\newtheorem{cor}[theorem]{Corollary}
\newtheorem{proposition}[theorem]{Proposition}
\theoremstyle{definition}
\newtheorem{definition}[theorem]{Definition}
\theoremstyle{remark}
\newtheorem{example}[theorem]{Example}

\numberwithin{equation}{section}

\begin{document}

\hfill{\small HWM--07--41}

\hfill{\small EMPG--07--19}

\vskip 1cm

\title[Ramond-Ramond Fields, Fractional Branes and Orbifold
Differential K-theory]
{Ramond-Ramond Fields, Fractional Branes\\ and Orbifold
Differential K-theory}

\author{Richard J. Szabo and Alessandro Valentino}

\address{Department of Mathematics and Maxwell Institute for
  Mathematical Sciences, Heriot-Watt University, Colin Maclaurin
  Building, Riccarton, Edinburgh EH14 4AS, U.K.}

\email{R.J.Szabo@ma.hw.ac.uk}

\email{A.Valentino@ma.hw.ac.uk}

\begin{abstract}
We study D-branes and Ramond-Ramond fields on global orbifolds of
Type~II string theory with vanishing $H$-flux using methods of
equivariant K-theory and K-homology. We illustrate how Bredon
equivariant cohomology naturally realizes stringy orbifold
cohomology. We emphasize its role as the correct cohomological tool
which captures known features of the low-energy effective field
theory, and which provides new consistency conditions for fractional
D-branes and Ramond-Ramond fields on orbifolds. We use an equivariant
Chern character from equivariant K-theory to Bredon cohomology to
define new Ramond-Ramond couplings of D-branes which generalize
previous examples. We propose a definition for groups of
differential characters associated to equivariant K-theory. We
derive a Dirac quantization rule for Ramond-Ramond fluxes, and
study flat Ramond-Ramond potentials on orbifolds.
\end{abstract}

\maketitle

\renewcommand{\thefootnote}{\arabic{footnote}}
\setcounter{footnote}{0}

\newsection*{Introduction}

The study of fluxes and D-branes has been of fundamental importance in
understanding the nonperturbative structures of string theory and
M-theory. It has also established a common ground on which a fruitful
interaction between physics and mathematics takes place. For example,
the seminal papers~\cite{Minasian1997,Witten1998} demonstrated that
D-brane charges in Type~II superstring theory are classified by the
K-theory of the spacetime manifold, and that ordinary cohomology alone
cannot account for certain physical features induced by the dynamics
of D-branes. As emphasized by
refs.~\cite{Asakawa2001,Periwal:2000,Harvey:2000,Szabo2002}, and
analyzed in great detail in refs.~\cite{Reis2005,Reis2006}, another
description of D-branes is provided by K-homology which sheds light on
their geometrical nature and suggests that the standard picture
of a D-brane as a submanifold of spacetime equiped with a vector
bundle (and connection) should be modified.

Ramond-Ramond fields are dual objects to D-branes and have also been
extensively investigated, but until recently their geometric nature
has remained somewhat obscure. In ref.~\cite{Moore2000} it was
proposed that Ramond-Ramond fields are also classified by K-theory,
and that their total field strengths lie in the image of the Chern
character homomorphism from K-theory to ordinary cohomology. This
result led to the understanding that the Ramond-Ramond field is
correctly understood as a self-dual field quantized by K-theory, and
it explains various subtle issues surrounding the partition functions
of these fields. In refs.~\cite{Freed2000,Freed2000a} it was proposed
that these properties are most naturally formulated by regarding
Ramond-Ramond fields as cocycles for the differential K-theory of
spacetime, an elegant description that allows one to study the gauge
theory of Ramond-Ramond fields in topologically non-trivial
backgrounds which naturally incorporates consistency conditions such
as anomaly cancellation on branes in string theory and M-theory. These
issues were among the motivations that led to the
foundational paper~\cite{Hopkins2005}, in which a detailed, elaborate
construction for generalized differential cohomology theories is
given. The importance of these mathematical theories has been greatly
emphasized in refs.~\cite{Freed2006a,Freed:2006}, where they are used
to define and understand certain novel properties of quantum Hilbert
spaces of abelian gauge field fluxes. A twisted version of
differential K-theory has been proposed in
refs.~\cite{Freed:2006,Belov2006} and applied to the quantization of
Ramond-Ramond fields in an $H$-flux background, while a rigorous
geometrical definition of this theory has been developed recently
in~ref.~\cite{Carey2007}.

The goal of this paper is to extend these lines of developments to
study properties of Ramond-Ramond fields and D-branes in orbifolds of
Type~II superstring theory with vanishing $H$-flux. We limit our study
to the cases of good (or global) orbifolds $[X/G]$, where $X$ is a
manifold and $G$ is a finite group acting via diffeomorphisms of
$X$. It is possible to resolve singularities in the orbifold where it
fails to be a manifold, and replace the quotient space by a
non-compact manifold with appropriate asymptotic behaviour. However,
orbifold singularities do not pose a problem and one can still have
consistent superstrings propagating on
orbifolds~\cite{Dixon1985,Dixon1986}. It was proposed in
ref.~\cite{Witten1998} that D-branes on the orbifold spacetime $[X/G]$
are classified by the $G$-equivariant K-theory of the covering space
$X$, as defined in ref.~\cite{Segal1968}. A recent overview of related
developments in the case of abelian orbifolds can be found in
ref.~\cite{Kriz2007}.

One of the main new ingredients that we introduce into the description
of D-branes and fluxes on orbifolds is the use of Bredon
cohomology~\cite{Bredon1967,Dieck1987}. This is a powerful equivariant
cohomology theory that has both advantages and pitfalls. In contrast
to the more commonly used Borel equivariant cohomology, Bredon
cohomology is a good ``approximation'' to the classification of
D-brane charges. We will support this statement by showing that it
correctly captures the properties of Ramond-Ramond fields on an
orbifold, in particular it naturally takes into account the twisted
sectors of the string theory. It thereby gives a precise, rigorous
realization of stringy orbifold cohomology. We will also see that it
naturally arises in the Atiyah-Hirzebruch spectral sequence for
equivariant K-theory, a fact that we shall exploit to describe new
consistency conditions for D-branes and fluxes on orbifolds in terms
of classes in the Bredon cohomology of the covering space $X$. Related
to this feature is the fact that this equivariant cohomology theory is
the target for a Chern character homomorphism on equivariant K-theory,
defined in ref.~\cite{Luck1998}, which induces an isomorphism when
tensored over $\mathbb{R}$. By means of this technology, we present
new compact and elegant expressions for the Wess-Zumino couplings of
Ramond-Ramond fields to D-branes on $[X/G]$. This generalizes the
usual Ramond-Ramond couplings~\cite{Minasian1997} to orbifolds, and
yields appropriate correction terms to previous flat space
formulas. The major drawback of Bredon cohomology is that it is a
rather difficult, abstract theory to define, and is even more
difficult to explicitly calculate than other equivariant cohomology
theories.

Another main achievement of this paper is a proposed
definition of differential K-theory suitable for orbifolds. Though
extremely powerful and general, the machinery developed in
ref.~\cite{Hopkins2005} cannot be immediately applied to an
equivariant cohomology functor on the category of $G$-manifolds. By
using Bredon cohomology and the equivariant Chern character, we define
abelian groups that behave as natural generalizations of the ordinary
differential K-theory groups, in the sense that they agree in the case
of a trivial group and they satisfy analogous exact
sequences. Although far from having the generality of the work of
ref.~\cite{Hopkins2005}, our construction gives a systematic framework
in which to study Ramond-Ramond fields on orbifolds with a Dirac
quantization condition, including non-trivial contributions from flat
potentials, and it represents a first step in the development of
generalized differential cohomology theories in the equivariant
setting. It is here that the use of Bredon cohomology is particularly
important, both because of the equivariant Chern character isomorphism
and because the framework requires explicit use of differential forms,
neither of which can be accomodated directly by the Borel
construction.

The outline of the remainder of this paper is as follows. In
Section~\ref{Cohomologies} we summarize some basic notions about the
cohomology theories of spaces with group actions. In
Section~\ref{EqChern} we present a detailed definition of Bredon
cohomology and the construction of the equivariant Chern character of
ref.~\cite{Luck1998}, as these have not made appearences before in the
physics literature. These first two sections give the main
mathematical background for the rest of the paper. In
Section~\ref{DKhom} we make a brief excursion into the description of
D-branes using geometric equivariant K-homology, showing that the use
of K-cycles is very well-suited to the description of {\it fractional}
D-branes and their topological charges computed using equivariant
Dirac operator theory. In Section~\ref{RRCouplings} we use Bredon
cohomology and the equivariant Chern character to define Ramond-Ramond
couplings to D-branes on orbifolds and compare it with previous
examples in the literature. Our formulas include the appropriate
gravitational contributions which are derived from an equivariant
version of the Riemann-Roch theorem and equivariant index theory. In
Section~\ref{OrbdiffKG} we give a detailed mathematical construction
of the orbifold differential K-theory groups, and prove that
they fit into appropriate exact sequences which are useful in
applications. In Section~\ref{Fluxquant} we use the orbifold
differential K-theory to describe the flux quantization of
Ramond-Ramond fields on orbifolds by writing an equivariant version
of the Ramond-Ramond current in terms of the equivariant Chern
character. We also study the group of flat potentials in detail, and
illustrate how the spectral sequence for equivariant K-theory can be
used to determine obstruction classes in Bredon cohomology which yield
stability conditions for D-branes and fluxes on orbifolds. Appendix~A
contains some background material on functor categories used in the
main text, Appendix B records the definitions of equivariant
K-homology, while Appendix C demonstrates the use of geometric
equivariant K-cycles in the classification of D-brane charges on
orbifolds.
 
\subsection*{Acknowlegments}

We are grateful to J.~Figueroa-O'Farrill, D.~Freed, J.~Greenlees,
J.~Howie, A.~Konechny, W.~L\"uck, M.~Lawson, R.~Reis, P.~Turner and
S.~Willerton for helpful discussions and correspondence. This work was
supported in part by the Marie Curie Research Training Network Grant
{\sl ForcesUniverse} (contract no.~MRTN-CT-2004-005104) from the
European Community's Sixth Framework Programme.

\newsection{Cohomology of spaces with symmetries\label{Cohomologies}}

In this section we will recall some basic notions about (generalized)
equivariant cohomology theories that we will need throughout this
paper. In the following, $X$ denotes a topological space and $G$ a
finite group, unless otherwise stated. Throughout a (left) action $G\times
X\to X$ of $G$ on $X$ will be denoted $(g,x)\mapsto g\cdot x$. The
stabilizer or isotropy group of a point $x\in X$ is denoted
$G_x=\{g\in G~|~g\cdot x=x\}$. Recall that a continuous map $f:X\to Y$
of $G$-spaces is a \emph{$G$-map} if $f(g\cdot x)=g\cdot f(x)$ for all
$g\in G$ and $x\in X$.

\subsection{$\mbf G$-complexes\label{Gcomplex}}

A \emph{$G$-equivariant CW-decomposition} of a $G$-space $X$ consists
of a filtration $X_{n}$, $n\in{\mathbb{N}_0}$ such that
\begin{displaymath}
X=\bigcup_{n\in\nat_0}\,{X_{n}}
\end{displaymath}
and $X_{n}$ is obtained from $X_{n-1}$ by ``attaching'' equivariant
cells by the following procedure. Define
\begin{displaymath}
X_{0}=\coprod_{j\in{J_0}}\,G/K_{j} \ ,
\end{displaymath}
with $K_{j}$ a collection of subgroups of ${G}$ and the standard
(left) $G$-action on any coset space $G/K_j$. For $n\geq1$ set
\beq
X_{n}=\Big(X_{n-1}\amalg\coprod_{j\in J_n}\,\big(\B^{n}_{j}
\times{G/{K_{j}}}\big)\Big)\,\Big/\,\sim
\label{Xnattach}\eeq
where the equivalence relation $\sim$ is generated by $G$-equivariant
``attaching maps''%%
\beq
\phi^{n}_{j}\,:\,\S^{n-1}_{j}\times{G/K_{j}}
~\longrightarrow~ X_{n-1} \ .
\label{Gattach}\eeq
One requires that $X$ carries the colimit topology with
respect to ($X_{n}$), \emph{i.e.}, $B\subset{X}$ is closed if and only
if $B\cap{X_{n}}$ is closed in $X_{n}$ for all $n\in\nat_0$. We call
the image of $\B^{n}_{j}\times{G/{K_{j}}}$
(resp.~$\mathring{\B}^{n}_{j}\times{G/{K_{j}}}$) a
\emph{closed} (resp.~\emph{open}) $n$-cell of orbit type $G/K_{j}$. As
usual, we call the subspace $X_{n}$ the {$n$-skeleton} of
$X$. If $X=X_{n}$ and $X\neq{X_{n-1}}$, then $n$ is called
the (\emph{cellular}) \emph{dimension} of $X$ and $X$ is said to be of
\emph{finite type}. A $G$-space with a $G$-equivariant
CW-decomposition is called a \emph{$G$-complex}.

When $G=e$ is the trivial group, a $G$-complex is just an ordinary
CW-complex. In general, if $X$ is a $G$-complex then the orbit space
$X/G$ is an ordinary CW-complex. Conversely, there is an intimate
relation between $G$-complexes and ordinary CW-complexes whenever $G$
is a discrete group. Let $X$ be a $G$-space which is an ordinary
CW-complex. We say that \emph{$G$ acts cellularly} on $X$ if %%
\begin{itemize}
\item[1)] For each $g\in{G}$ and each open cell $E$ of $X$, the left
  translation $g\cdot E$ is again an open cell of $X$; and
\item[2)] If $g\cdot E=E$, then the induced map $E \to E$, $x \mapsto
  g\cdot x$ is the identity.
\end{itemize}
Then we have the following
\begin{proposition} Let $X$ be a CW-complex with a cellular action of
  a discrete group $G$. Then $X$ is a $G$-complex with $n$-skeleton
  $X_{n}$.
\end{proposition}
In the case that $X$ is a smooth manifold, we require the $G$-action
on $X$ to be smooth and there is an analogous result. Recall that the
applicability of algebraic topology to manifolds relies on the fact
that any manifold comes equiped with a canonical CW-decomposition. In
the case in which a group acts on the manifold one has the following
result due to Illman~\cite{Illman1978,Illman1983}.
\begin{theorem}
If $G$ is a compact Lie group or a finite group acting on a smooth
compact manifold $X$, then $X$ is triangulable as a finite
$G$-complex.
\end{theorem}
The collection of $G$-complexes with $G$-maps as morphisms form a
category. We are interested in equivariant cohomology theories defined
on this category (or on subcategories thereof).

\subsection{Equivariant cohomology theories\label{EqCohTh}}

We will now briefly spell out the main ingredients involved in
building an equivariant cohomology theory on the category of finite
$G$-complexes, leaving the details to the comprehensive treatments
of~refs.~\cite{Dieck1987} and~\cite{Luck2006}, and focusing instead on some
explicit examples. Fix a group $G$ and a commutative
ring $R$. A \emph{$G$-cohomology theory $\E^\bullet_{G}$ with values in
  $R$-modules} is a collection of contravariant functors $\E^{n}_{G}$
from the category of $G$-CW~pairs to the category of $R$-modules
indexed by $n\in\mathbb{Z}$ together with natural transformations%%
\begin{displaymath}
\delta^{n}_{G}(X,A)\,:\,\E^{n}_{G}(X,A) ~\longrightarrow~
\E^{n+1}_{G}(X):=\E^{n+1}_{G}(X,\emptyset)
\end{displaymath}
for all $n\in{\mathbb{Z}}$ satisfying the axioms of $G$-homotopy
invariance, long exact sequence of a pair, excision, and disjoint
union. The theory is called \emph{ordinary} if for any orbit $G/H$ one
has $\E_{G}^{q}(G/H)=0$ for all $q\neq{0}$. These axioms are
formulated in an analogous way to that of ordinary cohomology. The new
ingredient in an equivariant cohomology theory (which we have not
yet defined) are the \emph{induction structures}, which we shall now
describe.

Let $\alpha:H\to{G}$ be a group homomorphism, and let $X$ be an
$H$-space. Define the \emph{induction of $X$ with respect to $\alpha$}
to be the $G$-space $\ind_{\alpha}X$ given by%%
\begin{displaymath}
\ind_{\alpha}X:=G\times_{\alpha}X \ .
\end{displaymath}
This is the quotient of the product $G\times X$ by the $H$-action
$h\cdot(g,x):=(g\,\alpha(h^{-1}),h\cdot x)$, with the $G$-action on
$\ind_{\alpha}X$ given by $g'\cdot[g,x]=[g'\,g,x]$. If $H<G$ and
$\alpha$ is the subgroup inclusion, the induced $G$-space is denoted
$G\times_HX$.

An \emph{equivariant cohomology theory $\E_{(-)}^{\bullet}$ with values in
  $R$-modules} consists of a collection of $G$-cohomology theories
$\E_{G}^{\bullet}$ with values in $R$-modules for each group $G$ such
that for any group homomorphism $\alpha:H \to G$ and any $H$-CW~pair
$(X,A)$ with $\ker(\alpha)$ acting freely on $X$, there are for each
$n\in\mathbb{Z}$ natural isomorphisms
\beq
\ind_{\alpha}\,:\,\E^{n}_{G}\big(\ind_{\alpha}(X,A)\big)
~\xrightarrow{\approx}~\E^{n}_{H}(X,A)
\label{indgen}\eeq
satisfying%%
\begin{itemize}
\item[(a)] Compatibility with the coboundary homomorphisms:
\begin{displaymath}
\delta^{n}_{H}\circ{\ind_{\alpha}}=\ind_{\alpha}\circ{\delta^{n}_{G}}
\ ;
\end{displaymath}
\item[(b)] Functoriality: If $\beta:G\to{K}$ is another group
  homomorphism such that $\ker(\beta\circ\alpha)$ acts freely on $X$,
  then for every $n\in\mathbb{Z}$ one has%%
\begin{displaymath}
\ind_{\beta\circ\alpha}=\ind_{\alpha}\circ{\ind_{\beta}}\circ{}\E^{n}_{K}(f_{1})
\end{displaymath}
where
\begin{eqnarray*}
f_{1}\,:\,\ind_{\beta}\big(\ind_{\alpha}(X,A)\big)&\xrightarrow{\approx}&
\ind_{\beta\circ{\alpha}}(X,A) \\
(k,g,x)&\longmapsto&\big(k\,\beta(g)\,,\,x\big)
\end{eqnarray*}
is a $K$-homeomorphism and $\E^{n}_{K}(f_{1})$ is the morphism on
$K$-cohomology induced by $f_1$; and
\item[(c)] Compatibility with conjugation: For $g,g'\in{G}$ define
  ${\rm Ad}_g(g'\,)=g\,g'\,g^{-1}$. Then the homomorphism $\ind_{{\rm Ad}_g}$
  coincides with $\E^{n}_{G}(f_{2})$, where%%
\begin{eqnarray*}
f_{2}\,:\,(X,A)&\xrightarrow{\approx}& \ind_{{\rm Ad}_g}(X,A)\\
x &\longmapsto& \big(e\,,\,g^{-1}\cdot x\big)
\end{eqnarray*}
is a $G$-homeomorphism, where throughout $e$ denotes the identity
element in the group~$G$.
\end{itemize}
Thus the induction structures connect the various $G$-cohomologies and
keep track of the equivariance. They will be very important in the
construction of the equivariant Chern character for equivariant
K-theory in the next section, even if we are only interested in a
fixed group $G$.

\begin{example}[\emph{Borel cohomology}]\label{Borelex}
Let $\H^{\bullet}$ be a cohomology theory for CW-pairs (for example,
singular cohomology). Define%%
\begin{displaymath}
\H^{n}_{G}(X,A):=\H^{n}\big(EG\times_{G}(X,A)\big)
\end{displaymath}
where $EG$ is the total space of the classifying principal
$G$-bundle $EG\to BG$ which is contractible and carries a free
$G$-action. This is called (\emph{equivariant}) \emph{Borel
  cohomology}, and it is the most commonly used form of equivariant
cohomology in the physics literature. Note that $\H_G^\bullet$ is
well-defined because the quotient $EG\times_GX$ is unique up to the
homotopy type of $X/G$. The ordinary $G$-cohomology structures on
$\H^{\bullet}_{G}$ are inherited from the cohomology structures on
$\H^{\bullet}$. The induction structures for $\H^{\bullet}_{G}$ are
constructed as follows. Let $\alpha:H \to G$ be a group
homomorphism and $X$ an $H$-space. Define%%
\begin{eqnarray*}
b\,:EH\times_{H}X &\longrightarrow& EG\times_{G}G\times_{\alpha}X\\
(\varepsilon,x) &\longmapsto& \big(E\alpha(\varepsilon)\,,\,e\,,\,x
\big)
\end{eqnarray*}
where $\varepsilon\in{EH}$, $x\in{X}$ and $E\alpha:EH \to EG$ is the
$\alpha$-equivariant map induced by $\alpha$. The induction map
$\ind_{\alpha}$ is then given by pullback%%
\begin{displaymath}
\ind_{\alpha}:=b^{*}\,:\,\H^{n}_{G}(\ind_{\alpha}X)=\H^{n}(EG\times_{G}G
\times_{\alpha}X)
~\longrightarrow~\H^{n}(EH\times_{H}X)={\H^{n}_{H}(X)} \ .
\end{displaymath}
If $\ker(\alpha)$ acts freely on $X$, then the map $b$ is a homotopy
equivalence and hence the map $\ind_{\alpha}$ is an isomorphism.
\end{example}

\begin{example}[\emph{Equivariant K-theory}]\label{EqKex}
In~ref.~\cite{Segal1968}, equivariant topological K-theory is defined for
any $G$-complex $X$ as the abelian group completion of the semigroup
$\Vect_G^\complex(X)$ of complex $G$-vector bundles over $X$,
\emph{i.e.}, bundles $E\to X$ together with a lift of the $G$-action
on $X$ to the fibres. The higher groups are defined via iterated
suspension. To define the induction structures, recall that if $X$ is
an $H$-space and $\alpha:H \to G$ is a group homomorphism, then the
map%%
\begin{eqnarray*}
\varphi\,:\,X &\longrightarrow& G\times_{\alpha}X\\
x &\longmapsto& (e,x)
\end{eqnarray*}
is an $\alpha$-equivariant map which embeds $X$ as the subspace
$H\times_{\alpha}X$ of $G\times_{\alpha}X$, and which induces via
pullback of vector bundles the homomorphism%%
\begin{displaymath}
\varphi^{*}\,:\,\K^{\bullet}_{G}(G\times_{\alpha}X) ~\longrightarrow~
\K^{\bullet}_{H}(X) \ .
\end{displaymath}
This map defines the induction structure. It is invertible when
$\ker(\alpha)$ acts freely on $X$, with inverse the ``extension'' map
$E\mapsto G\times_{H}E$ for any $H$-vector bundle $E$ over $X$. The
induction structure can be used to prove the well-known
\emph{equivariant excision theorem}%%
\beq
\K^{\bullet}_{G/N}(X/N)\cong\K^{\bullet}_{G}(X)
\label{eqexcision}\eeq
where $N$ is a normal subgroup of $G$ acting freely on $X$. Indeed,
one has%%
\begin{displaymath}
X/N\cong(G/N)\times_{G}X
\end{displaymath}
and if we define $\alpha:G\to{G/N}$ to be the quotient map, then%%
\begin{displaymath}
\K^{\bullet}_{G/N}\big((G/N)\times_{\alpha}X\big)\cong{\K^{\bullet}_{G}(X)}
\end{displaymath}
since $\ker(\alpha)=N$ acts freely on $X$.
\end{example}

\newsection{The equivariant Chern character\label{EqChern}}
In this section we will describe the equivariant Chern character for
the $\K^{\bullet}_{G}$ functor and its target cohomology theory, Bredon
cohomology. To this end, we will introduce some technology related to
modules over functor categories, giving the necessary definitions and
directing the reader to the relevant literature for further
details. Some pertinent aspects of functor categories are summarized
in Appendix~A.

\subsection{Chern character in topological K-theory\label{CherntopK}}
Let us begin by recalling some basic notions about the ordinary Chern
character. Define $\pi_{-\bullet}\K$ to be the complex K-theory ring
of the point. It is the $\zed$-graded ring $\mathbb{Z}[[u,u^{-1}]]$ of
Laurent polynomials freely generated by an element $u$ of degree ${\rm
  deg}(u)=2$, where $u^{-1}\in{\K^{-2}(\pt)}$ is called the Bott
element and is represented by the Hopf bundle over $\S^2$. One then
has a homomorphism%%
\begin{displaymath}
\ch\,:\,\K^{\bullet}(X) ~\longrightarrow~\H(X;\mathbb{R}\otimes
\pi_{-\bullet}\K)^{\bullet}
\end{displaymath}
which induces the natural $\zed$-graded ring isomorphism%%
\begin{displaymath}
\K^{\bullet}(X)\otimes{\mathbb{R}}~ \xrightarrow{\approx}~ 
\H(X;\mathbb{R}\otimes\pi_{-\bullet} \K)^{\bullet}
\end{displaymath}
for any finite CW-complex $X$. This statement is true
even if we tensor over $\mathbb{Q}$. The use here of the
K-theory of the point as the coefficient ring serves just as a
re-grading of the cohomology ring $\H^{\bullet}(X;\mathbb{R})$. For
example, it is easy to check that%%
\begin{displaymath}
\H(X;\mathbb{R}\otimes\pi_{-\bullet} \K)^{0}\cong{\H^{\rm
    even}(X;\mathbb{R})} \ .
\end{displaymath}
In particular, the Chern character tells us that K-theory and
cohomology are the same thing up to torsion.

It is natural now to ask if there exists such a morphism for
equivariant K-theory. One might naively think that the correct target
theory for the equivariant Chern character would naturally be
Borel cohomology. But the problem is much more subtle than
it first may seem. The crucial point is that while in the ordinary
cohomology of (finite) CW-complexes the building blocks are the
cohomology groups of a point, in the equivariant case they are
the cohomology groups of the orbits $G/H$ for all subgroups $H$ of
$G$, as we saw in Section~\ref{Gcomplex}. Any equivariant
cohomology theory $\E_G^\bullet$ on the category of finite
$G$-complexes is completely specified by its value on the
orbit spaces $G/H$. A \emph{localization
  theorem} due to Atiyah and Segal~\cite{Atiyah1969} tells us that the
Borel cohomology of a $G$-space $X$ is isomorphic to its equivariant
K-theory localized at the augmentation ideal in the representation
ring $R(G)$ consisting of all elements whose characters vanish at the
identity $e$ in $G$ (regarding $\K^\bullet_G(X)$ as a module over
$R(G)$). Localizing at a prime ideal of $R(G)$ corresponds to
restricting $X$ to the set of fixed points of an associated conjugacy
class of cyclic subgroups of $G$. In this sense, Borel cohomology does
not take into account the ``contributions'' of the non-trivial
elements in $G$, and hence of the fixed points of the $G$-action.

There are several approaches to the equivariant Chern character (see
refs.~\cite{Atiyah1989,Slominska1976,Block1994,Freed2002a,Adem2003},
for example) which strongly depend on the types of groups involved
(discrete, continuous, \emph{etc.}) and on the ring one tensors with
($\real$, $\complex$, \emph{etc.}). As we are interested in finite
groups and real coefficients for our physical applications later on,
we will use the Chern character constructed in refs.~\cite{Luck2006}
and~\cite{Luck1998}. Thus we proceed to the more abstract, but
powerful and compact, definition of Bredon cohomology, which will turn
out to be the best suited equivariant cohomology theory for all of our
purposes.

\subsection{Bredon cohomology}\label{bredon}

Let $G$ be a discrete group. The \emph{orbit category}
$\ocat{G}$ of $G$ is defined as the category whose objects are
homogeneous spaces $G/H$, with $H<{G}$, and whose morphisms are
$G$-maps between them. From general considerations~\cite{Dieck1987} it
follows that a $G$-map between two homogeneous spaces $G/H$ and $G/K$
exists if and only if $H$ is conjugate to a subgroup of $K$, and hence
any such map is of the form%%
\beq
\big(g\,H ~\longmapsto~ g\,a\,K\big)
\label{MapGH}\eeq
for some $a\in{G}$ such that $a^{-1}\,H\,a<{K}$. If
$\mathfrak{F}$ is any family of subgroups of $G$ then there is a
subcategory $\ocat{G,\mathfrak{F}}$ with objects $G/H$ for
$H\in\mathfrak{F}$. A simple example is provided by the cyclic groups
$G=\mathbb{Z}_{p}$ with $p$ prime, for which the orbit category has
just two objects, $G/e=G$ and $G/G=\pt$.

If $\cat{Ab}$ denotes the category of abelian groups, then a
\emph{coefficient system} is a functor
\begin{displaymath}
\underline{F}\,:\,\ocat{G}^{\text{op}}~\longrightarrow~\cat{Ab}
\end{displaymath}
where $\ocat{G}^{\rm op}$ denotes the dual category to
$\ocat{G}$. With such a functor and any $G$-complex $X$,\footnote{When
  $G$ is an infinite discrete group, one should restrict to
  \emph{proper} $G$-complexes, \emph{i.e.}, with finite stabilizer for
  any point of $X$. Some further minor assumptions are needed when $G$
  is a Lie group.} one can define for each $n\in\mathbb{Z}$ the
group%%
\beq
C^{n}_{G}(X,\,\underline{F}\,):=\text{Hom}_{\ocat{G}}\big(\,
\underline{C}\,_{n}(X)\,,\,\underline{F}\,\big)
\label{CGnXF}\eeq
where $\underline{C}\,_{n}(X):\ocat{G}^{\rm op}\to\cat{Ab}$ is the
projective functor defined by
$$\underline{C}\,_{n}(X)(G/H):=C_{n}\big(X^{H}\big) \ , $$ the cellular
homology of the fixed point complex 
\beq
X^{H}:=\big\{x\in{X}~\big|~h\cdot{x}=x \quad 
\forall{h\in{H}}\big\} \ .
\label{XHdef}\eeq
In eq.~(\ref{CGnXF}), $\Hom_{\ocat{G}}(-,-)$ denotes the group of
natural transformations between two contravariant functors, with the
group structure inherited by the images of the functors in
$\cat{Ab}$. The functoriality property of
$\underline{C}\,_{n}(X)(G/H)$ is the natural one induced by the
identification $X^{H}\cong\text{Map}_{G}(G/H,X)$. Indeed, the two maps
\begin{eqnarray*}
X^{H}~\longrightarrow~\text{Map}_{G}(G/H,X) & , & \qquad
x~\longmapsto~f_{x}\big([g\,H]\big)=g\cdot{x} \ , \\[4pt]
\text{Map}_{G}(G/H,X)~\longrightarrow~X^{H} & , & \qquad
f~\longmapsto~{f(H)}
\end{eqnarray*}
are easily seen to be inverse to each other, and the desired
homeomorphism is obtained by giving the space
$\text{Map}_{G}(G/H,X)$ the \emph{compact-open} topology. In
particular, a $G$-map (\ref{MapGH}) induces a cellular map $X^K\to
X^H$, $x\mapsto a\cdot x$.

These groups can be expressed in terms of the $G$-complex structure of
$X$. If the $n$-skeleton $X_{n}$ is obtained by attaching equivariant
cells as in eq.~(\ref{Xnattach}) with $K_j$ the stabilizer of an
$n$-cell of $X$, then the cellular chain complex $C_\bullet(X)$
consists of $G$-modules $C_n(X)=\bigoplus_{j\in J_n}\,\zed[G/K_j]$ and
hence%%
\begin{displaymath}
\underline{C}^{}\,_{n}(X)(G/H)\cong\bigoplus_{j\in J_n}\,\mathbb{Z}\big[
\text{Mor}_{\ocat{G}}(G/H,G/K_j)\big] \ .
\end{displaymath}
For each $n\geq0$, the group $C_G^n(X\,,\,\underline{F}\,)$ is the
direct limit functor over all $n$-cells of orbit type $G/K_j$ in $X$
of the groups $\underline{F}\,(G/K_j)$. This follows by restricting
eq.~(\ref{CGnXF}) to the full subcategory $\ocat{G,\mathfrak{F}(X)}$,
with $\mathfrak{F}(X)$ the family of subgroups of $G$ which occur as
stabilizers of the $G$-action on $X$~\cite{Mislin2003}.

The $\zed$-graded group
$C^{\bullet}_{G}(X,\,\underline{F}\,)=\bigoplus_{n\in\mathbb{Z}}\,
C^{n}_{G}(X,\,\underline{F}\,)$ inherits a coboundary operator
$\delta$, and hence the structure of a cochain complex,
from the boundary operator on cellular chains. To a natural
transformation $f:\underline{C}\,_{n}(X)\to{\underline{F}\,}$, one
associates the natural transformation $\delta{f}$ defined by%%
\begin{eqnarray*}
\delta{f}(G/H)\,:\,C_{n}\big(X^{H}\big)&\longrightarrow&
{\underline{F}\,(G/H)}\\
\sigma&\longmapsto&f(G/H)(\partial\sigma)
\end{eqnarray*}
for $\sigma\in{C_{n-1}(X^{H})}$, with naturality induced from that of
the cellular boundary operator $\partial$. Then the \emph{Bredon
  cohomology} of $X$ with coefficient system $\underline{F}$ is
defined as%%
\begin{displaymath}
\H^{\bullet}_{G}(X;\,\underline{F}\,):=\text{H}\big(C^{\bullet}_{G}(X,\,
\underline{F}\,)\,,\,\delta\big) \ .
\end{displaymath}
This defines a $G$-cohomology theory. See~ref.~\cite{Luck2002} for the
proof that $\H^{\bullet}_{G}(X;\,\underline{F}\,)$ is an equivariant
cohomology theory, \emph{i.e.}, for the definition of the induction
structure. One can also define cohomology groups by restricting the
functors in eq.~(\ref{CGnXF}) to a subcategory
$\ocat{G,\mathfrak{F}}$. The definition of Bredon cohomology is
independent of $\mathfrak{F}$ as long as $\mathfrak{F}$ contains the
family $\mathfrak{F}(X)$ of stabilizers~\cite{Mislin2003}. This fact
is useful in explicit calculations. In particular, by taking
$\mathfrak{F}=H$ to consist of a single subgroup, one shows that the
Bredon cohomology of $G$-homogeneous spaces is given by
\beq
\H^\bullet_{G}(G/H;\,\underline{F}\,)\=
\H^0_{G}(G/H;\,\underline{F}\,)\=\underline{F}\,(G/H) \ .
\label{BredonGH}\eeq

\begin{example}[\emph{Trivial group}] 
When $G={e}$ is the trivial group, \emph{i.e.}, in the non-equivariant
case, the functors $\underline{C}\,_{n}(X)$ and $\underline{F}$ can be
identified with the abelian groups
$C_{n}(X)=\underline{C}\,_{n}(X)(e)$ and $F=\underline{F}\,(e)$. Then
$$C^{n}_{e}(X,F)=C^{n}(X,F)$$ and one has
$\H^{n}_{e}(X;\,\underline{F}\,)=\text{H}\left(C^{n}(X,F),\delta\right)$,
\emph{i.e.}, the ordinary $n$-th cohomology group of $X$ with
coefficients in $F$.
\label{Bredontrivgpex}\end{example}

\begin{example}[\emph{Free action}]
If the $G$-action on $X$ is \emph{free}, then all
stabilizers $K_j$ are trivial and $X^H=\emptyset$ for every $H\leq G$,
$H\neq e$. In this case one may take $\mathfrak{F}=e$ to compute the
cochain complex
$$
C_G^\bullet(X,\,\underline{F}\,)\cong\Hom_G\big(C_\bullet(X)\,,\,
\underline{F}\,(G/e)\big)
$$
and so the Bredon cohomology $\H^{\bullet}_{G}(X;\,\underline{F}\,)$
coincides with the equivariant cohomology
$$\H_G^\bullet\big(X\,;\,\underline{F}\,(G/e)\big)$$ of $X$ with
coefficients in the $G$-module
$\underline{F}\,(G/e)=\underline{F}\,(G)$. In the case of the constant
functor $\underline{F}=\underline{\zed}$, with
$\underline{\zed}\,(G/H)=\zed$ for every $H\leq G$ and the value on
morphisms in $\ocat{G}^{\rm op}$ given by the identity homomorphism of
$\zed$, this group reduces to the ordinary cohomology
$\H^\bullet(X/G;\zed)$.
\label{Bredonfreeactionex}\end{example}

\begin{example}[\emph{Trivial action}]
If the $G$-action on $X$ is \emph{trivial}, then the collection of
isotropy groups $K_j$ for the $G$-action is the set of all subgroups
of $G$ and $X^H=X$ for all $H\leq G$. In this case the functor
$\underline{C}\,_n(X)$ can be decomposed into a sum over $n$-cells of
projective functors $\underline{P}\,_{K_j}$ with
$K_{j}=G$~\cite{Mislin2003}, and so one has
$$
\Hom_{\ocat{G}}\big(\,\underline{C}\,_n(X)\,,\,\underline{F}\,\big)\cong
\Hom\big(C_n(X)\,,\,\lim_{\longleftarrow}{}_{\ocat{G}^{\rm op}}\,
\underline{F}\,(G/H)\big)
$$
where the inverse limit functor is taken over the opposite category
$\ocat{G}^{\rm op}$. It follows that the Bredon cohomology
$$
\H^\bullet_G(X;\,\underline{F}\,)=\H^\bullet\big(X\,;\,
\underline{F}\,(G/G)\big)
$$
is the ordinary cohomology of $X$ with coefficients in the abelian
group $\underline{F}\,(G/G)=\underline{F}\,(\pt)$.
\label{Bredontrivactionex}\end{example}

\subsection{Representation ring functors}

In what follows we will specialize the coefficient system for Bredon
cohomology to the
\emph{representation ring functor} $\underline{F}=\underline{R}(-)$
defined on the orbit category $\ocat{G}$ by sending the left coset
$G/H$ to $R(H)$, the complex representation ring of the group
$H$. A morphism (\ref{MapGH}) is sent to the
homomorphism $R(K) \to R(H)$ given by first restricting the
representation from $K$ to the subgroup conjugate to $H$, and then
conjugating by $a$. Since $\underline{R}(-)$ is a functor to rings,
the Bredon cohomology $\H^{\bullet}_{G}(X;\underline{R}(-))$ naturally
has a ring structure. Note that
\beq
R(H)~\cong~{\K^{0}_{G}(G/H)}\=\K^{\bullet}_{G}(G/H) \ ,
\label{RHGHKiso}\eeq
which follows from the induction structure of Example~\ref{EqKex} with
$X=\pt$ and $\alpha$ the subgroup inclusion $H\hookrightarrow G$. 
By eq.~(\ref{BredonGH}) the group (\ref{RHGHKiso}) also coincides with
the Bredon cohomology group $\H^{\bullet}_{G}(G/H;\underline{R}(-))$,
which is already an indication that Bredon cohomology is a better
relative of equivariant K-theory than Borel cohomology. Indeed, using
the induction structure of Example~\ref{Borelex} one shows that the
Borel cohomology
$$
\H_G^\bullet(G/H)=\H^\bullet(BH)
$$
coincides with the cohomology of the classifying space $BH=EH/H$,
which computes the group cohomology of $H$ and is typically
infinite-dimensional (even for finite groups $H$). In this paper we
will show that the Bredon cohomology
$\H^{\bullet}_{G}(X;\underline{R}(-))$ gives a more precise
realization of the stringy orbifold cohomology of $X$ in the context
of \emph{open} string theory.

In the construction of the equivariant Chern character in
Section~\ref{section2} below, it will be important to represent
the rational Bredon cohomology
$\H^{\bullet}_{G}(X;\mathbb{Q}\otimes{\underline{R}(-)})$ as a certain
group of homomorphisms of functors, similarly to the cochain groups
(\ref{CGnXF}). For this, we introduce another
category $\subc{G}$. The objects of $\subc{G}$ are the
subgroups of $G$,\footnote{If $G$ is infinite then one should restrict
  to finite subgroups of $G$.} and the morphisms are given by%%
\begin{displaymath}
\text{Mor}_{\subc{H,K}}:=\left\{f:H\to{K}~\big|~\exists\:{g\in{G} \ ,
    \ g\,H\,g^{-1}\leq{K}}\ , \ f={\rm
    Ad}_g\right\}\,\big/\,\text{Inn}(K) \ .
\end{displaymath}
In particular, there is a functor $\ocat{G}\to{\subc{G}}$ which sends
the object $G/H$ to $H$ and the morphism (\ref{MapGH}) in $\ocat{G}$
to the homomorphism $(g\mapsto{a^{-1}\,g\,a})$ in $\subc{G}$. If
$a$ lies in the centralizer
\beq
{Z_{G}(H)}:=\big\{g\in G~\big|~g^{-1}\,H\,g=H\big\}
\label{ZGH}\eeq
 of $H$ in $G$, then the morphism (\ref{MapGH}) is sent to the
identity map. Any functor
$\underline{F}:\subc{G}^{\text{op}}\to\cat{Ab}$ can be naturally
regarded as a functor on $\ocat{G}^{\text{op}}$.

Define the quotient functors
$\underline{C}\,_{\bullet}^{\text{qt}}(X)\,,\,
\underline{\H}\,_{\bullet}^{\text{qt}}(X):\subc{G}^{\rm
  op}\to\cat{Ab}$ by%%
\begin{displaymath}
\underline{C}\,_{\bullet}^{\text{qt}}(X)(H)~:=~
C_{\bullet}\big(X^{H}/Z_{G}(H)\big) \qquad\text{and}\qquad
\underline{\H}\,_{\bullet}^{\text{qt}}(X)(H)~:=~\H_{\bullet}
\big(X^{H}/Z_{G}(H)\big) \ .
\end{displaymath}
For any functor $\underline{F}:\subc{G}^{\text{op}}\to\cat{Ab}$
one has%%
\begin{displaymath}
\text{Hom}\big(C_{\bullet}(X^{H}/Z_{G}(H))\,,\,\underline{F}\,(H)\big)
\cong\text{Hom}_{Z_{G}(H)}\big(C_{\bullet}(X^{H})\,,\,\underline{F}\,(H)
\big) \ .
\end{displaymath}
By observing that the centralizer (\ref{ZGH}) is precisely the group of
automorphisms of $G/H$ in the orbit category $\ocat{G}$ sent to the
identity map in the subgroup category $\subc{G}$, we finally
have%%
\beq
C_G^\bullet(X,\,\underline{F}\,)\=
\text{Hom}_{\ocat{G}}\big(\,\underline{C}\,_{\bullet}(X)\,,\,
\underline{F}\,\big)~\cong~\text{Hom}_{\subc{G}}\big(\,
\underline{C}\,_{\bullet}^{\text{qt}}(X)\,,\,\underline{F}\,\big) \ .
\label{HomOrSubG}\eeq
At this point one can apply eq.~(\ref{HomOrSubG}) to the rational
representation ring functor
$\underline{F}=\rat\otimes\underline{R}(-)$, which by construction can
be regarded as an injective functor $\subc{G}^{\rm op}\to\cat{Ab}$, to
prove the%%
\begin{lemma}[\cite{Luck1998}]\label{subcat}
For any finite group $G$ and any $G$-complex $X$, there exists an
isomorphism of rings
\begin{displaymath}
\Phi_{X}\,:\,\H^{\bullet}_{G}\big(X\,;\,\mathbb{Q}
\otimes{\underline{R}(-)}\big)~\xrightarrow{\approx}~{\Hom}_{{\sf
    Sub}(G)}\big(\,\underline{\H}\,_{\bullet}^{\rm qt}(X)\,,\,\mathbb{Q}
\otimes\underline{R}(-)\big) \ .
\end{displaymath}
\end{lemma}

\subsection{Chern character in equivariant K-theory}\label{section2}

Before spelling out the definition of the equivariant Chern character,
we recall some basic properties of the equivariant K-theory of a
$G$-complex $X$. Let $H$ be a subgroup of $G$, and consider the fixed
point subspace of $X$ defined by eq.~(\ref{XHdef}). The action of $G$
does not preserve $X^{H}$, but the action of the normalizer $N_G(H)$ of
$H$ in $G$ does. If we denote with $i:X^{H}\hookrightarrow{X}$ the
inclusion of $X^{H}$ as a subspace of $X$, and with
$\alpha:N_G(H)\hookrightarrow{G}$ the inclusion of $N_G(H)$ as a
subgroup of $G$, then we naturally have the equality%%
\begin{displaymath}
i(n\cdot{x})=\alpha(n)\cdot{i(x)}
\end{displaymath}
for all ${n\in{N_G(H)}}$ and $x\in X^H$. It follows that the induced
homomorphism on equivariant K-theory is a map~\cite{Segal1968}%%
\begin{displaymath}
i^{*}\,:\,\K^\bullet_{G}(X)~\longrightarrow~{\K^\bullet_{N_G(H)}\big(X^{H}
\big)}
\end{displaymath}
which is called a \emph{restriction morphism}.

We also  need a somewhat less known property~\cite{Luck1998}. Let
$N\lhd\,{G}$ be a finite normal subgroup, and let $\text{Rep}(N)$ be the
category of (isomorphism classes of) irreducible complex representations of
$N$. Let $X$ be a (proper) $G/N$-complex, and let $G$ act on $X$ via
the projection map $G\to{G/N}$. Then for any complex $G$-vector bundle
$E\to{X}$ and any representation $V\in{\text{Rep}(N)}$, define
$\text{Hom}_{N}(V,E)$ as the vector bundle over $X$ with total space%%
\begin{displaymath}
\text{Hom}_{N}(V,E):=\bigcup_{x\in{X}}\,\text{Hom}_{N}(V,E_{x})
\end{displaymath}
where $N$ acts on the fibres of $E$ because of the action of $G$ via
the projection map. Now if $H\leq{G}$ is a subgroup which commutes
with $N$, $[H,N]=e$, then one can induce an $H$-vector bundle from
$\text{Hom}_{N}(V,E)$ by defining $(h\cdot f)(v)=h\cdot{f(v)}$, $v\in
V$ for any $h\in{H}$ and any $f\in\text{Hom}_{N}(V,E)$ (remembering
that $G$ acts on $E$). Hence there is a homomorphism of rings%%
\begin{displaymath}
\Psi\,:\,\K^\bullet_{G}(X)~\longrightarrow~{\K^\bullet_{H}(X)\otimes{R(N)}}
\end{displaymath}
defined on $G$-vector bundles by%%
\beq
\Psi\big([E]\big):=\sum_{V\in{{\rm Rep}(N)}}\,\big[\text{Hom}_{N}(V,E)
\big]\otimes[V] \ .
\label{PsiEdef}\eeq
This homomorphism satisfies some naturality properties which are
described in detail in~ref.~\cite{Luck1998}. Note that the sum
(\ref{PsiEdef}) is \emph{finite}, since $N$ is a finite subgroup.

We are now ready to construct the equivariant Chern character as a
homomorphism%%
\begin{displaymath}
\ch_{X}\,:\,\K^{0,1}_{G}(X)~\longrightarrow~{\H^{\rm even,odd}_{G}
\big(X\,;\,\mathbb{Q}\otimes\underline{R}(-)\big)}
\end{displaymath}
for any finite proper $G$-complex $X$. The strategy used
in~ref.~\cite{Luck1998} is to construct $\zed_2$-graded
homomorphisms%%
\beq
\ch^{H}_{X}\,:\,\K^{\bullet}_{G}(X)~\longrightarrow~
{\text{Hom}\big(\H_{\bullet}(X^{H}/Z_{G}(H))\,,\,\mathbb{Q}\otimes
{R}(H)\big)}
\label{chXH}\eeq
for any finite subgroup $H$, and then \emph{glue} them together as $H$
varies through the finite subgroups of $G$. To define the homomorphism
(\ref{chXH}), we first compose the ring homomorphisms%%
\begin{displaymath}
\K^{\bullet}_{G}(X)~\xrightarrow{i^{*}}~\K_{N_G(H)}^\bullet\big(X^{H}
\big)~\xrightarrow{\Psi}~\K^{\bullet}_{Z_{G}(H)}\big(X^{H}\big)
\otimes{R(H)}~\xrightarrow{\pi^{*}_{2}\otimes\Id}~\K^{\bullet}_{Z_{G}(H)}
\big(EG\times{X^{H}}\big)\otimes{R(H)}
\end{displaymath}
where $\pi_2:EG\times X^H\to X^H$ is the projection onto the second
factor. By using the induction structure of Example~\ref{EqKex}, one
then has%%
\bea
\K^{\bullet}_{Z_{G}(H)}\big(EG\times{X^{H}}\big)\otimes{R(H)}&
\xrightarrow{\approx}&\K^{\bullet}\big(EG\times_{Z_{G}(H)}X^{H}
\big)\otimes R(H) \nonumber\\ && \qquad ~
\xrightarrow{\ch\otimes\text{id}}~
\H\big(EG\times_{Z_{G}(H)}X^{H}\,;\,\mathbb{Q}\otimes\pi_{-\bullet}\K
\big)^{\bullet}
\otimes{R(H)} \nonumber
\eea
where $\ch$ is the ordinary Chern character. One finally has%%
\bea
\H^{\bullet}\big(EG\times_{Z_{G}(H)}X^{H}\,;\,\mathbb{Q}\big)
\otimes{R(H)}&\xrightarrow{\approx}&\H^{\bullet}\big(X^{H}/Z_{G}(H)\,;\,
\mathbb{Q}\big)\otimes{R(H)} \nonumber\\ && \qquad ~\cong~
\text{Hom}\big(\H_{\bullet}(X^{H}/Z_{G}(H))\,,\,
\mathbb{Q}\otimes{R(H)}\big) \ , \nonumber
\eea
where the first isomorphism follows from the Leray spectral sequence
by observing that the fibres of the projection
\begin{displaymath}
EG\times_{Z_{G}(H)}X^{H}~\longrightarrow~{X^{H}\,\big/\,Z_{G}(H)}
\end{displaymath}
are all classifying spaces of finite groups, having trival reduced
cohomology with $\rat$-coefficients and are therefore $\rat$-acyclic.

The equivariant Chern character is now defined as\footnote{If $G$ is
  infinite then the direct sum in eq.~(\ref{chXprod}) is understood as
  the inverse limit functor over the dual subgroup category
  $\subc{G}^{\rm op}$.}
\beq
\ch_{X}=\bigoplus_{H\leq G}\,\ch_X^{H} \ .
\label{chXprod}\eeq
By using the various naturality properties of the
homomorphism~(\ref{PsiEdef})~\cite{Luck1998}, one sees that $\ch_{X}$
takes values in
$\text{Hom}_{\subc{G}}\big(\underline{\H}\,_{\bullet}^{\text{qt}}(X)\,,\,
\mathbb{Q}\otimes\underline{R}(-)\big)$, and by Lemma~\ref{subcat} it
is thus a $\zed_2$-graded map%%
\begin{displaymath}
\ch_{X}\,:\,\K^{\bullet}_{G}(X)~\longrightarrow~\textrm{Hom}_{\subc{G}}
\big(\,\underline{\H}\,_{\bullet}^{\textrm{qt}}(X)\,,\,\mathbb{Q}\otimes
\underline{R}(-)\big)\cong{\H^{\bullet}_{G}\big(X\,;\,
\mathbb{Q}\otimes{\underline{R}(-)}\big)} \ .
\end{displaymath}
This map is well-defined as a ring homomorphism because all maps
involved above are homomorphisms of rings. As with the definition of
Bredon cohomology, the sum (\ref{chXprod}) may be restricted to any
family of subgroups of $G$ containing the set of stabilizers
$\mathfrak{F}(X)$.

To conclude, we have to prove that this map becomes an isomorphism
upon tensoring over $\rat$. For this, one proves that the morphism
$\ch_X$ in eq.~(\ref{chXprod}) is an isomorphism on homogeneous spaces
$G/H$, with $H$ a finite subgroup of $G$, and then uses induction on
the number of orbit types of cells in $X$ along with the
Mayer-Vietoris sequences for the pushout squares induced by the
attaching $G$-maps (\ref{Gattach}). The isomorphism on $G/H$ is a
consequence of the isomorphisms (\ref{BredonGH}) and
(\ref{RHGHKiso}). The details may be found
in~ref.~\cite{Luck1998}. Let $\underline{\pi_{-\bullet}\K}\,_G(-)$ be
the functor on $\ocat{G}$ defined by
$G/H\mapsto\K^\bullet_G(G/H)$. Then one has the following%%
\begin{theorem}
For any finite proper $G$-complex $X$, the Chern character $\ch_X$
extends to a natural $\zed$-graded isomorphism of rings
\begin{displaymath}
{\ch_{X}}\otimes\mathbb{Q}\,:\,{\K^{\bullet}_{G}(X)}\otimes
\mathbb{Q}~\xrightarrow{\approx}~\H_{G}\big(X\,;\,\mathbb{Q}
\otimes{\underline{\pi_{-\bullet}\K}\,_G(-)}\big)^{\bullet} \ .
\end{displaymath}
\label{eqChernthm}\end{theorem}

\newsection{D-branes and equivariant K-cycles\label{DKhom}}

In this section we will make some remarks concerning the topological
classification of D-branes and their charges on global orbifolds of
Type~II superstring theory with vanishing $H$-flux. Let $X$ be a
smooth manifold and $G$ a (finite) group acting by diffeomorphisms of
$X$. Ramond-Ramond charges on the global orbifold $[X/G]$ are
classified by the equivariant K-theory $\K_G^\bullet(X)$ as defined in
Example~\ref{EqKex}~\cite{GarciaCompean1998,Olsen2000,Witten1998}.
Dually, the equivariant K-homology $\K_\bullet^G(X)$ leads to an
elegant description of \emph{fractional D-branes} pinned at the
orbifold singularities in terms of equivariant K-cycles. In the
following we will frequently refer to Appendix~B for
detailed definitions and technical aspects of equivariant K-homology,
focusing instead here on some of the more qualitative aspects of
D-branes on orbifolds in this language. In the remainder of this paper
we will assume for definiteness that $G$ is a finite group.

\subsection{Fractional D-branes\label{FracDbrane}}

As in the non-equivariant case
$G=e$~\cite{Reis2005,Reis2006,Szabo2002}, a very natural description
of D-branes in the orbifold space, which captures the inherent
geometrical picture of D-brane states involving wrapped cycles in
spacetime, is provided by the topological realization of the equivariant
K-homology groups $\K_\bullet^G(X)$. The cycles for this homology
theory, called $G$-equivariant K-cycles, live in an additive category
$\cat{D}^G(X)$ whose objects are triples $(W,E,f)$ where $W$ is a
$G$-spin$^c$ manifold without boundary, $E$ is a $G$-vector bundle
over $W$, and
\beq\label{DGmap}
f\,:\,W~\longrightarrow~ X
\eeq
is a $G$-map. The group $\K_\bullet^G(X)$ is the quotient of this
category by the equivalence relation generated by bordism, direct sum,
and vector bundle modification, as detailed in Appendix~B. Note that
$W$ need not be a submanifold of spacetime. However, since $X$ is a
manifold, we can restrict the bordism equivalence relation to
\emph{differential bordism}~\cite{Reis2005} and assume that the map
(\ref{DGmap}) is a differentiable $G$-map in equivariant K-cycles
$(W,E,f)\in\cat{D}^G(X)$. In this way the category $\cat{D}^G(X)$
extends the standard K-theory classification to include branes
supported on non-representable cycles in spacetime. This definition of
equivariant K-homology thus gives a concrete geometric model for the
topological classification of D-branes $(W,E,f)$ in a global orbifold
$[X/G]$ which captures the physical constructions of orbifold D-branes
as $G$-invariant states of branes on the covering space $X$. In the
subsequent sections we will study the pairing of Ramond-Ramond fields
with these D-branes.

Consider a D-brane localized at a generic point in the orbifold
$[X/G]$ with the action of the \emph{regular} representation of $G$ on
the fibres of its Chan-Paton gauge bundle, \emph{i.e.}, the natural
action of the group algebra $\complex[G]$ as bounded linear operators
$\ell^2(G)\to\ell^2(G)$. It corresponds to a $G$-orbit of such branes
on the leaves $X^g=\{x\in X~|~g\cdot x=x\}$, $g\in G$ of the covering
space $X$. At a $G$-fixed point, this brane splits up into a set of
\emph{fractional branes} according to the decomposition of the
representation of $G$ on the fibres of the Chan-Paton bundle into
irreducible $G$-modules. Stable fractional D-branes correspond to
bound states of branes wrapping various collapsed cycles at the fixed
points. They are thus stuck at the orbifold points and provide the
open string analogs of ``twisted sectors''.

To formulate this physical construction in the language of equivariant
K-cycles $(W,E, f)$, let $G^\vee$ denote the set of conjugacy classes
$[g]$ of elements $g\in G$. It can be regarded as a set of
representatives for the isomorphism classes $\pi_0{\rm Rep}(G)$, where
${\rm Rep}(G)$ is the additive category of irreducible complex
representations of $G$ which coincides with the category of D-brane
boundary conditions at the orbifold points. There is a natural
subcategory $\cat{D}_{\rm frac}^G(X)$ of $\cat{D}^G(X)$ consisting of
triples $(W,E,f)$ for which $W$ is a $G$-fixed space, \emph{i.e.}, for
which
\beq
W^{g}= W
\label{WgW}\eeq
for all $g\in G$. By $G$-equivariance this implies $f(W)^g=f(W)$ for
all $g\in G$, and so the image of the brane worldvolume lies in the
subspace
$$
f(W)~\subset~\bigcap_{g\in G}\,X^g \ .
$$
This is precisely the set of $G$-fixed points of $X$, and so the
objects $(W,E,f)$ of the category $\cat{D}_{\rm frac}^G(X)$ are
naturally pinned to the orbifold points. We call $\cat{D}_{\rm
  frac}^G(X)$ the category of ``maximally fractional D-branes''.

In this case, an application of Schur's lemma shows that the
Chan-Paton bundle admits an isotopical decomposition and there is a
canonical isomorphism of $G$-bundles
\beq
E~\cong~\bigoplus_{[g]\in G^\vee}\,E_{[g]}\otimes\id_{[g]} \qquad
\mbox{with} \quad E_{[g]}\=\Hom_G\big(\id_{[g]}\,,\,E\big) \ ,
\label{Eisodecomp}\eeq
where $E_{[g]}$ is a complex vector bundle with trivial $G$-action and
$\id_{[g]}$ is the $G$-bundle $W\times V_{[g]}$ with
$\gamma:G\to\End(V_{[g]})$ the irreducible representation
corresponding to the conjugacy class $[g]\in G^\vee$. This isotopical
decomposition defines the action of $G$ on the Chan-Paton factors of
the D-brane, and it implies the well-known isomorphism
\beq
\K_G^\bullet(W)\cong R(G)\otimes\K^\bullet(W)
\label{KGWtrivial}\eeq
for $G$-fixed spaces $W$~\cite{Segal1968}. This is a special case of
the homomorphism $\Psi$ defined in eq.~(\ref{PsiEdef}). From the
direct sum relation in equivariant K-homology it follows that a
D-brane, represented by a K-cycle $(W,E,f)$, splits at an orbifold
point into a sum over irreducible fractional branes represented by the
K-cycles $(W,E_{[g]}\otimes\id_{[g]},f)$, $[g]\in G^\vee$.

It is important to realize that the full category $\cat{D}^G(X)$
contains much more information, and in particular the fractional
D-branes will not generically form a spanning set of K-cycles for the
group $\K_\bullet^G(X)$ (except in some specific examples). However,
it follows from the bordism relation in equivariant K-homology that
any two $G$-equivariant K-cycles $(W_i,E_i,f_i)$, $i=0,1$ which are
bordant along the same $G$-orbit determine the same element in
$\K_\bullet^G(X)$. This is expected since a purely topological classification
such as equivariant K-homology cannot capture the positional moduli
associated with the regular D-branes in $X/G$.

A related way to understand the role of fractional branes is through
the connection between geometric K-homology and bordism
theory~\cite{Reis2005}. Let $\MSpin_\bullet^c(X)$ be the
\emph{ordinary} spin$^c$ bordism group of $X$, which forgets about the
$G$-action and consists of \spinc bordism classes of pairs
$(W,f)$. Then there is a map
\bea
\MSpin^c_\bullet(X)\otimes_{\MSpin^c_\bullet(\pt)}\,{\rm Rep}(G)
&\longrightarrow&\cat{D}^G_{\rm frac}(X) \ , \nonumber \\
(W,f)\otimes V&\longmapsto& (W,W\times V,f)
\label{MSpincmap}\eea
which descends to give a homomorphism
$$
\MSpin^c_\bullet(X)\otimes_{\MSpin^c_\bullet(\pt)}\,\K_\bullet^G(\pt)~
\longrightarrow~\K_\bullet^G(X) \ .
$$
When $G=e$ this is the isomorphism of $\K_\bullet(\pt)$-modules
induced by the Atiyah-Bott-Shapiro orientation, and the map
(\ref{MSpincmap}) determines K-cycle generators in terms of
\spinc bordism generators~\cite{Reis2005}. The equivariant extension
of the Atiyah-Bott-Shapiro construction is given in
ref.~\cite{Landweber2005} in terms of finite-dimensional
$\zed_2$-graded $G$-Clifford modules. Since any $G$-Clifford module
can be built as a direct sum of tensor products of $G$-modules and
ordinary Clifford modules (see Appendix~B), there is an isomorphism of
$R(G)$-modules $\K_G^\bullet(\pt)\cong R(G)\otimes\K^\bullet(\pt)$ and
so these representation modules contain no new information about the
orbifold group. This seems to suggest that, at least in certain cases,
spanning sets of equivariant K-cycles can be taken to lie in the
subcategory~$\cat{D}^G_{\rm frac}(X)$.

\subsection{Topological charges\label{Topcharge}}

The topological charge of a fractional D-brane, in a given closed
string twisted sector of the orbifold string theory on a
$G$-spin$^c$ manifold $X$, can be
computed by using the equivariant Dirac operator theory developed in
Appendix~B. The equivariant index of the $G$-invariant \spinc Dirac
operator $\Dirac_E^X$ coupled to a $G$-vector bundle $E\to X$ takes
values in $\K_G^\bullet(\pt)\cong R(G)$. We can turn this into a
homomorphism on $\K_\bullet^G(X)$ with values in $\zed$ by composing
with the projection $R(G)\to\zed$ defined by taking the multiplicity
of a given representation
\beq
\gamma\,:\,G~\longrightarrow~\End(V_\gamma)
\label{gammaunitary}\eeq
of $G$ on a finite-dimensional complex vector space
$V_\gamma$. There is a corresponding class in the KK-theory group
$$[\gamma]~\in~\KK_\bullet\big(\complex[G]\,,\,\End(V_\gamma)\big)$$
which is represented by the Kasparov module
$(V_\gamma,\gamma,0)$ associated with the extension of the
representation~(\ref{gammaunitary}) to a complex representation of
group ring $\complex[G]$. By Morita equivalence, the Kasparov product
with $[\gamma]$ is the homomorphism on K-theory
$$\K_0\big(\complex[G]\big)~\longrightarrow~
\K_0\big(\End(V_\gamma)\big)\cong\K_0(\complex)\cong\zed$$ induced by
$\gamma:\complex[G]\to\End(V_\gamma)$. We may then define a
homomorphism $$\mu_\gamma\,:\,\K_0^G(X)~\longrightarrow~\zed$$ of
abelian groups by
\beq
\mu_\gamma\big([W,E, f]\big)\=\Index_\gamma
\big( f_*[\Dirac_E^W]\big)~:=~
\ass\big( f_*[\Dirac_E^W]\big)\otimes_{\complex[G]}[\gamma]
\label{mugammadef}\eeq
on equivariant K-cycles $(W,E, f)\in\cat{D}^G(X)$ (and extended
linearly), where
$$\ass\,:\,\K_\bullet^G(X)~\longrightarrow~
\K_\bullet\big(\complex[G]\big)$$ is the analytic assembly map
constructed in Appendix~B.

\subsection{Linear orbifolds}\label{Linorb1}

Let us now consider a simple class of examples wherein everything can
be made very explicit. Let $V$ be a complex vector space
of dimension $\dim_\complex(V)=d\geq1$, and let $G$ be a finite
subgroup of $\SL(V)$. Our spacetime $X$ is the $G$-space identified
with the product
$$
X=\real^{p,1}\times V \ ,
$$
where $G$ acts trivially on the Minkowski space $\real^{p,1}$. Fractional
D-branes carrying themselves a complex linear representation of $G$,
which is a submodule of $\real^{p,1}\times V$, have worldvolumes $W$
linearly embedded in the subspace $\real^{p,1}$ and have transverse
space given by the orthogonal complement $f(W)^\perp\cong V$ with
respect to a chosen inner product. Since the space of hermitean
metrics is contractible, all topological quantities below are
independent of this choice.

As a $G$-space, $V$ is equivariantly contractible to
a point and hence its compactly supported equivariant K-theory is
given by~\cite{Atiyah1969}
$$
\K^\bullet_{G,\cpt}(V)~\cong~\K^\bullet_{G}(\pt)~\cong~R(G)\=
\zed^{|G^\vee|} \ . 
$$
This group coincides with the Bredon cohomology
$\H^\bullet_{G,\cpt}(V;\underline{R}(-))$, owing to the fact that the
equivariant Chern character $\ch_{G/H}$ of Section~\ref{section2} is
the identity map (since the non-equivariant Chern character
$\ch=c_0:\K^0(\pt)\to\H^0(\pt;\zed)$ is the identity map). It follows
that the fractional D-branes, as
defined by elements of equivariant K-theory, can be identified with
representations of the orbifold group\footnote{If the transverse space
  $V$ is instead a \emph{real} linear $G$-module, then throughout one
  should restrict to the subring of $R(G)$ consisting of
  representations associated to conjugacy classes $[g]\in G^\vee$
  for which the centralizer $Z_G(g)$ acts on the fixed point subspace
  $V^g$ by oriented automorphisms~\cite{Karoubi2002}. This will follow
  immediately from the isomorphism~(\ref{chCiso}) below with $X=V$.}
$$
\gamma=\bigoplus_{a=1}^{|G^\vee|}\,N_a\,\gamma_a
$$
consisting of $N_a\geq0$ copies of the $a$-th irreducible
representation $$\gamma_a\,:\,G~\longrightarrow~\End(V_a) \ , \quad
a\=1,\dots,\big|G^\vee\big| \ , $$
which defines the action of $G$ on the fibres of the Chan-Paton
bundle. More precisely, each irreducible fractional brane is
associated to the $G$-bundle $V\times V_a$ over $V$.

By Poincar\'e duality, it follows from Proposition~2.1 of
ref.~\cite{Reis2005} that a basis for the equivariant K-homology group
$\K_\bullet^{G}(V)$ is provided by the geometric equivariant K-cycles
$(V,V\times V_a,\Id_{V})$, $a=1,\dots,|G^\vee|$. By $G$-homotopy
invariance~\cite[Lemma~1.4]{Reis2005} these
cycles can be contracted to $[\pt,V_a,i]$, where $i$ is the inclusion
of a point $\pt\subset V$ whose induced homomorphism
$$i_*\,:\,\K_\bullet^G(\pt)~\longrightarrow~\K_\bullet^G(V)$$ can be
taken to be the identity map $R(G)\to R(G)$. This is simply the
physical statement that the stable fractional branes in this case are
D0-branes in Type~IIA string theory (the Type~IIB theory containing no
such states). The $G$-invariant Dirac operator $\Dirac_{V_a}^\pt$ is
just Clifford multiplication twisted by the $G$-module $V_a$, and thus
the topological charges (\ref{mugammadef}) of the corresponding
fractional branes in the twisted sector labelled by $b$ are given by
$$
\mu_b\big([\pt,V_a,i]\big)\=\Index_{\gamma_b}
\big([\Dirac_{V_a}^\pt]\big)\=
\big[V_a\otimes(\Delta^+-\Delta^-)\big]\otimes_{\complex[G]}
[\gamma_b] \ ,
$$
where $\Delta^\pm$ are the half-spin representations of $\SO(2d)$ on
$V\cong\complex^d$ regarded as $\complex[G]$-modules. Acting on the
character ring the projection gives
$[W]\otimes_{\complex[G]}[\gamma_b]=\chi_W(g_b)$, where
$\chi_W:G\to\complex$ is the character of the $G$-module $W$ and
$[g_b]\in G^\vee$ is the conjugacy class corresponding to the
irreducible representation $\gamma_b$.

\newsection{Delocalization and Ramond-Ramond
  couplings\label{RRCouplings}}

The purpose of this section is to describe the delocalization of
Bredon cohomology and the equivariant Chern character, introduced in
Section~\ref{EqChern}, and to apply it to the study of the coupling
between Ramond-Ramond potentials and the D-branes of the previous
section. In the following we will require
some standard facts concerning string theory on global orbifolds,
particularly its low-energy field theory content. For more details
see~refs.~\cite{Dixon1985,Dixon1986}.

\subsection{Closed string spectrum}\label{clstring}

The boundary states corresponding to the fractional branes constructed
in Section~\ref{DKhom} have components in the twisted sectors of
the closed string Hilbert space $\hil$ of orbifold string theory on
$X$. The closed string is an embedding $x:\S^1\times\real\to X$ of the
worldsheet cylinder, with local coordinates
$(\sigma,\tau)\in\S^1\times\real$, into the $G$-\spinc spacetime
manifold $X$. The Hilbert space $\hil$ of physical string states
decomposes into a direct sum over twisted sectors, each characterized
by a conjugacy class, as%%
\beq
\mathcal{H}=\bigoplus_{[g]\in G^\vee}\,\mathcal{H}_{[g]}
\label{Hildecomp}\eeq
with only $G$-invariant states surviving in each superselection sector
$\hil_{[g]}$. Actually, the Hilbert space factorizes into one sector
for each element of the group $G$, but the action of $G$ mixes the
sectors within a given conjugacy class. The subspaces in
eq.~(\ref{Hildecomp}) are thus defined as
\begin{displaymath}
\mathcal{H}_{[g]}:=\bigoplus_{h\in[g]}\,\mathcal{H}_{h} \ ,
\end{displaymath}
where $\mathcal{H}_{h}$ is the subspace of states induced by the
twisted string field boundary condition%%
\beq
x(\sigma+2\pi,\tau)=h\cdot{x(\sigma,\tau)}
\label{twistbc}\eeq
with an analogous condition on the worldsheet fermion fields (using a
lift $\hat G$ of the orbifold group). Then $G$ acts on the subspace
$\mathcal{H}_{[g]}$, and projecting onto $G$-invariant states in
$\mathcal{H}_{[g]}$ is equivalent to projecting onto
$Z_G(h)$-invariant states in $\mathcal{H}_{h}$ for any $h$ in $[g]$.

The low-energy limit of Type~II orbifold superstring theory on $X$
contains Ramond-Ramond fields $C_{[g]}$ coming from the various twisted
sectors. The twisted boundary conditions (\ref{twistbc})
on the string embedding map imposes constraints on the low-energy
spectrum. For example, the untwisted sector given by $g=e$ contains
Ramond-Ramond fields defined on the entire spacetime manifold $X$,
while the twisted sector represented by $g\neq e$ gives rise to fields
defined only on the fixed point submanifold $X^{g}$. The GSO
projection then enforces the properties that the Ramond-Ramond form
potentials $C_{[g]}$ determine self-dual fields in each twisted
sector, and that they be of odd degree in Type~IIA theory and of even
degree in Type~IIB theory.

The Ramond-Ramond fields can thus be ``organised'' into the
differential complex%%
\beq
\Omega^\bullet_{G}(X;\real):=\bigoplus_{[g]\in G^\vee}\,\Omega^\bullet
\big(X^{g}\,;\,\real\big)^{Z_G(g)} \ .
\label{RRcomplex}\eeq
{Here we consider only fields coming from inequivalent
  twisted sectors and make a choice of submanifold
  $X^{g}$, since for any conjugate element $h\in[g]$ there is
  a diffeomorphism $X^{g}\cong{X^{h}}$. (No choice is needed in the
  case in which $G$ is an abelian group.)} As $\dd\circ
g^*=g^*\circ\dd$ for all $g\in G$, the derivation is given by
$$
\dd_G:=\bigoplus_{[g]\in G^\vee}\,\dd_g
$$
where
$\dd_g=\dd:\Omega^\bullet(X^{g};\real)\to\Omega^\bullet(X^{g};\real)$ is
the usual de~Rham exterior derivative. Note that only the centralizer
subgroup of $g$ in $G$ acts (properly) on $X^g$.

\subsection{Delocalization of Bredon cohomology\label{DelocBredon}}

We will now show how Bredon cohomology can be used to compute
the cohomology of the complex (\ref{RRcomplex}) of orbifold
Ramond-Ramond fields by giving a delocalized description of Bredon
cohomology with \emph{real} coefficients, following
refs.~\cite{Mislin2003} and~\cite{Luck1998} where further details can
be found. This is the stringy orbifold cohomology of $X$, defined as
the ordinary (real) cohomology of the \emph{orbifold resolution}
$\widetilde{X}=\coprod_{[g]\in G^\vee}\,X^g/Z_G(g)$. Note that there is a
natural surjective map $\pi:\widetilde{X}\to X$ defined by
$(x,[g])\mapsto x$, and a natural injection
$X\hookrightarrow\widetilde{X}$ into the connected component of
$\widetilde{X}$ corresponding to the untwisted sector~$[g]=[e]$.

Denote with $\underline{\real}(-)$ the real representation ring
functor $\mathbb{R}\otimes{\underline{R}(-)}$ on the orbit category
$\ocat{G}$. Let $\langle G\rangle$ denote the set of conjugacy classes
$[C]$ of cyclic subgroups $C$ of $G$. Let
$\underline{\real}\,_C(-)$ be the contravariant functor on
$\ocat{G}$ defined by $\underline{\real}\,_C(G/H)=0$ if $[C]$
contains no representative $g\,C\,g^{-1}<H$, and otherwise
$\underline{\real}\,_C(G/H)$ is isomorphic to the cyclotomic field
$\mathbb{R}(\zeta_{|C|})$ over $\mathbb{R}$ generated by the primitive
root of unity $\zeta_{|C|}$ of order $|C|$. A standard result from the
representation theory of finite groups then gives a natural
splitting%%
\begin{displaymath}
\underline{\real}(-)=\bigoplus_{[C]\in\langle G\rangle}\,
\underline{\real}\,_C(-) \ .
\end{displaymath}
By definition, for any module $\underline{M}\,(-)$ over the orbit
category one has
\bea\nonumber
\text{Hom}_{\ocat{G}}\big(\,\underline{M}\,(-)\,,\,
\underline{\real}\,_{C}(-)\big)&\cong&
\text{Hom}_{N_{G}(C)}\big(\,\underline{M}\,(G/C)\,,\,
\underline{\real}\,_{C}(G/C)\big) \\[4pt] \nonumber &\cong&
\underline{M}\,(G/C)\otimes_{N_{G}(C)}\,\underline{\real}\,_{C}(G/C)
\eea
where the normalizer subgroup $N_{G}(C)$ acts on
$\underline{\real}\,_{C}(G/C)\cong\mathbb{R}(\zeta_{|C|})$ via
identification of a generator of $C$ with $\zeta_{|C|}$.

These facts together imply that the cochain groups (\ref{CGnXF}) with
$\underline{F}=\underline{\real}(-)$ admit a splitting given by%%
\begin{displaymath}
C^\bullet_{G}\big(X\,,\,\underline{\real}(-)\big)\cong
\bigoplus_{[C]\in\langle
  G\rangle}\,C^\bullet\big(X^{C}\big)\otimes_{N_{G}(C)}\,
\underline{\real}\,_{C}(G/C) \ .
\end{displaymath}
As the centralizer $Z_{G}(C)$ acts properly on $X^{C}$, the natural
map%%
\begin{displaymath}
\bigoplus_{[C]\in\langle G\rangle}\,C^\bullet\big(X^{C}\big)\otimes_{N_{G}(C)}\,
\underline{\real}\,_{C}(G/C) ~\longrightarrow~ 
\bigoplus_{[C]\in\langle G\rangle}\,C^\bullet\big(X^{C}/Z_G(C)\big)
\otimes_{W_{G}(C)}\,\underline{\real}\,_{C}(G/C)
\end{displaymath}
is a cohomology isomorphism, where $W_{G}(C):=N_{G}(C)/Z_{G}(C)$ is
the Weyl group of $C<G$ which acts by translation on
$X^C/Z_G(C)$. Since $\underline{\real}\,_{C}(G/C)$ is a projective
$\mathbb{R}[W_{G}(C)]$-module, it follows that for any proper
$G$-complex $X$ the Bredon cohomology of $X$ with coefficient system
$\mathbb{R}\otimes{\underline{R}(-)}$ has a splitting%%
\beq\label{Bredoncyclic}
\H_{G}^{\bullet}\big(X\,;\,\mathbb{R}\otimes{\underline{R}(-)}\big)
\cong\bigoplus_{[C]\in\langle
  G\rangle}\,\H^{\bullet}\big(X^{C}/Z_{G}(C)\,;\,
\mathbb{R}\big)\otimes_{W_{G}(C)}\,\underline{\real}\,_{C}(G/C) \ .
\eeq

At this point, we note that the dimension of the $\mathbb{R}$-vector
space%%
\begin{displaymath}
\underline{\real}\,_C(G/C)^{W_G(C)}\cong
\mathbb{R}\otimes_{W_{G}(C)}\,\underline{\real}\,_{C}(G/C)
\end{displaymath}
is equal to the number of $G$-conjugacy classes of generators for
$C$. We also use the fact that for a finite group $G$ a sum over
conjugacy classes of cyclic subgroups is equivalent to a sum
over conjugacy classes of elements in $G$, and that $X^{\langle
  g\rangle}=X^g$ and $Z_G(\langle g\rangle)=Z_G(g)$. One finally
obtains a splitting of real Bredon cohomology groups\footnote{This
  splitting in fact holds over $\rat$~\cite{Mislin2003}.}
\beq\label{Bredonsplit}
\H_{G}^{\bullet}\big(X\,;\,\mathbb{R}\otimes{\underline{R}(-)}\big)
\cong\bigoplus_{[g]\in G^\vee}\,\H^{\bullet}\big(X^{g}\,;\,
\mathbb{R}\big)^{Z_{G}(g)}
\eeq
into ordinary cohomology groups
$$\H^\bullet\big(X^g\,;\,\real\big)^{Z_G(g)}~\cong~
\H^\bullet\big(X^g/Z_G(g)\,;\,\real\big)~\cong~
\H\big(\Omega^\bullet(X^g;\real)^{Z_G(g)}\,,\,\dd\big)$$ with
constant coefficients $\real$. The group on the right-hand side of
eq.~(\ref{Bredonsplit}) corresponds to the (real) ``delocalized
equivariant cohomology'' $\H^\bullet(\coprod_{g\in
  G}\,X^g)^G\otimes\real$ defined by Baum and
Connes~\cite{Baum1988,Baum1985}. Note that this group is
(non-canonically) isomorphic to $R(G)\otimes\H^\bullet(X;\real)$
when the $G$-action on $X$ is trivial. Furthermore, by using
Theorem~\ref{eqChernthm} one also has a decomposition for equivariant
K-theory with real coefficients given by
\begin{displaymath}
\K_{G}^{\bullet}(X)\otimes\mathbb{R}
\cong\bigoplus_{[g]\in G^\vee}\,\big(\K^{\bullet}(X^{g})\otimes
\mathbb{R}\big)^{Z_{G}(g)} \ .
\end{displaymath}
However, this decomposition captures only the torsion-free part of the
group $\K_{G}^{\bullet}(X)$.

\subsection{Delocalization of the equivariant Chern
  character\label{DelocCh}}

The complex (\ref{RRcomplex}) of orbifold Ramond-Ramond fields can
also be used to provide an explicit geometric description of the
(complex) equivariant Chern
character defined in Section~\ref{section2}. We will now explain this
construction, refering the reader to~ref.~\cite{Bunke2007} for the
technical details. Consider a complex $G$-bundle $E$ over $X$ equiped
with a $G$-invariant hermitean metric and a $G$-invariant metric
connection~$\nabla^{E}$. One can then define a closed $G$-invariant
differential form $$\ch(E)~\in~\Omega^\bullet(X;\complex)^{G}$$ in the
usual way by the Chern-Weil construction%%
\begin{displaymath}
\ch(E):=\Tr\big(\exp(-F^E/2\pi{\ii})\big)
\end{displaymath}
where $F^E$ is the curvature of the connection $\nabla^E$. It
represents a cohomology class
$$\big[\ch(E)\big]~\in~{\H^\bullet(X;\mathbb{C})^{G}}$$ in the fixed
point subring of the action of $G$ as automorphisms of
$\H^\bullet(X;\mathbb{C})$. By using the definition of the
homomorphisms (\ref{chXH}), with $\mathbb{Q}$ substituted by
$\mathbb{C}$ and $H=e$, one can establish the equality
\begin{displaymath}
\big[\ch(E)\big]=\ch^{e}_{X}\big([E]\big) \ .
\end{displaymath}

Let $C<{G}$ be a cyclic subgroup, and define the
cohomology class $$\big[\ch(g,E)\big]~\in~
{\H^\bullet\big(X^{C}\,;\,\mathbb{C}\big)^{Z_{G}(C)}}
\cong\H^\bullet\big(X^{C}/Z_{G}(C)\,;\,\mathbb{C}\big)\cong
\H\big(\Omega^\bullet(X^C;\complex)^{Z_G(C)}\,,\,\dd\big)$$
represented by%%
\begin{displaymath}
\ch(g,E):=\Tr\big(\gamma(g)\,
\exp(-F_C^E/2\pi{\ii})\big)
\end{displaymath}
where $g$ is a generator of $C$, $F_C^E$ is the restriction of the
invariant curvature two-form $F^E$ to the fixed point subspace
$X^{C}$, and $\gamma$ is a representation of $C$ on the fibres of the
restriction bundle $E|_{X^{C}}$ which is an $N_G(C)$-bundle over
$X^C$. The character $\chi_C$ naturally identifies
$R(C)\otimes\complex$ with the $\complex$-vector space of class
functions $C\to\complex$. By using the splitting (\ref{Bredoncyclic})
for complex Bredon cohomology, one can then show that%%
\begin{displaymath}
\ch^{C}_{X}\big([E]\big)(g)=\big[\ch(g,E)\big]
\end{displaymath}
up to the restriction homomorphism
$R(C)\otimes\complex\to\underline{\complex}\,_{C}(G/C)$ of rings with
kernel the ideal of elements whose characters vanish on all generators
of $C$.

Using eq.~(\ref{chXprod}) we can then define the map%%
\begin{displaymath}
\ch^{\mathbb{C}}\,:\,{\rm Vect}^{\mathbb{C}}_{G}(X)~
\longrightarrow~{\bigoplus_{[g]\in G^\vee}\,\Omega^{\rm even}\big(X^g\,;\,
\mathbb{C}\big)^{Z_G({g})}}
\end{displaymath}
from complex $G$-bundles $E\to X$ given by%%
\beq
\ch^{\mathbb{C}}(E)=\bigoplus_{[g]\in G^\vee}\,\Tr\big(\gamma(g)\,\exp(-F_g^E/2\pi{\ii})\big) \ .
\label{chCEdef}\eeq
At the level of equivariant K-theory, from Theorem~\ref{eqChernthm} it
follows that this map induces an isomorphism%%
\beq
\ch^{\mathbb{C}}\,:\,\K^\bullet_{G}(X)\otimes\mathbb{C}~
\xrightarrow{\approx}~\H_{G}\big(X\,;\,\mathbb{C}
\otimes{\underline{\pi_{-\bullet}\K}\,_G(-)}\big)^{\bullet}
\label{chCiso}\eeq
where we have used the splitting (\ref{Bredonsplit}). The map
(\ref{chCEdef}) coincides with the equivariant Chern character defined
in~ref.~\cite{Atiyah1989}.

\subsection{Wess-Zumino pairings\label{WZpairing}}

We now have all the necessary ingredients to define a coupling of the
Ramond-Ramond fields to a D-brane in the orbifold $[X/G]$. In this
section we will only consider Ramond-Ramond fields which are
topologically trivial, \emph{i.e.}, elements of the differential
complex (\ref{RRcomplex}), and use the delocalized cohomology theory
above by working throughout with complex coefficients. Under these
conditions we can straightforwardly make contact with existing
examples in the physics literature and write down their appropriate
generalizations.

To this aim, we introduce the bilinear product
\begin{displaymath}
\wedge_G\,:\,\Omega_{G}^\bullet(X;\real)\otimes\Omega_{G}^\bullet(X;
\real)~\longrightarrow{}~\Omega_{G}^\bullet(X;\real)
\end{displaymath}
defined on $\omega=\bigoplus_{[g]\in G^\vee}\,\omega_{[g]}$ and
$\eta=\bigoplus_{[g]\in G^\vee}\,\eta_{[g]}$ by
\beq
\omega\wedge_G{\eta}~:=~
\bigoplus_{[g]\in G^\vee}\,\omega_{[g]}\wedge_g\eta_{[g]}
\label{orbdiffprod}\eeq
where $\wedge_g=\wedge$ is the usual exterior product on
$\Omega^\bullet(X^g;\real)$. There is also an integration
$$\int_X^G\,:\,\Omega^\bullet_G(X;\real)~\longrightarrow~\real \
. $$ If $\omega\in\Omega^\bullet_{G}( X;\mathbb{R})$ then we set
\begin{displaymath}
\int_{X}^G\,\omega~:=~\frac1{|\,G^\vee\,|}~\sum_{[g]\in G^\vee}~
\int_{X^{g}/Z_G(g)}\,\omega_{[g]} \ .
\end{displaymath}
The normalization ensures that $\int_X^G\,\omega=\int_X\,\omega$ when
$G$ acts trivially on $X$ and $\omega\in\Omega^\bullet(X;\real)$ is
``diagonal'' in $R(G)\otimes\Omega^\bullet(X;\real)$.

Suppose now that $ f:W\rightarrow{X}$ is the smooth immersed
worldvolume cycle of a wrapped D-brane state $(W,E,f)\in\cat{D}^G(X)$
in the orbifold $[X/G]$, \emph{i.e.}, $ W$ is a $G$-\spinc manifold
equiped with a $G$-bundle $E\to W$ and an invariant connection
$\nabla^E$ on $E$. We define the \emph{Wess-Zumino pairing}
$$
\CS\,:\,\cat{D}^G(X)\times\Omega_G^\bullet(X;\complex)~\longrightarrow~
\complex
$$
between such D-branes and Ramond-Ramond fields as%%
\begin{equation}\label{coupling}
\CS\big((W,E,f)\,,\,C\big)=
\int_{ W}^G\,\tilde{C}\wedge_G{\,\ch^{\mathbb{C}}}(E)\wedge_G
\mathcal{R}(W,f) \ ,
\end{equation}
where $\tilde C=f^{*}C$ is the pullback along $f:W\to X$ of the
total Ramond-Ramond field $$C=\bigoplus_{[g]\in G^\vee}\,C_{[g]}$$ and
the equivariant Chern character is given by eq.~(\ref{chCEdef}) with
$\gamma$ giving the action of $G$ on the Chan-Paton factors of the
D-brane. The closed worldvolume form
$\mathcal{R}(W,f)\in\Omega^{\rm even}_{G,{\rm cl}}(W;\complex)$
represents a complex Bredon cohomology class which accounts for
gravitational corrections due to curvature in the spacetime $X$ and
depends only on the bordism class of $(W,f)$. It will be constructed
explicitly in Section~\ref{Gravcoupl} below in terms of the geometry
of the immersed cycle $f:W\rightarrow X$ and of the $G$-bundle~$\nu\to
W$ given by
\beq
\nu\=\nu(W;f)\=f^*(T_X)\oplus T_W \ .
\label{normbun}\eeq

It is easily seen that, modulo the curvature contribution
$\mathcal{R}(W,f)$, the very natural expression~(\ref{coupling})
reduces to the usual Wess-Zumino coupling of topologically trivial
Ramond-Ramond fields to D-branes in the case $G=e$. But even if a
group $G\neq e$ acts trivially on the brane worldvolume $ W$ (or on
the spacetime $X$), there can still be additional contributions to the
usual Ramond-Ramond coupling if $E$ is a \emph{non-trivial
  $G$-bundle}. This is the situation, for instance, for fractional
D-branes $$(W,E,f)~\in~\cat{D}_{\rm frac}^G(X)$$  placed at
orbifold singularities. In this case, we may use the isotopical
decomposition (\ref{Eisodecomp}) of the Chan-Paton bundle along with
eq.~(\ref{WgW}). Then the Wess-Zumino pairing (\ref{coupling})
descends to a pairing $$\CS_{\rm frac}\,:\,\cat{D}_{\rm
  frac}^G(X)\times\Omega_G^\bullet(X;\complex)~\longrightarrow~
\complex$$ with the additive subcategory of fractional branes at
orbifold singularities given by
\beq
\CS_{\rm frac}\big((W,E,f)\,,\,C\big)=
\int_W\,\Big(\,\frac1{|\,G^\vee\,|}~\sum_{[g]\in G^\vee}\,
\tilde C_{[g]}\wedge\ch\big(E_{[g]}\big)\,\chi_\gamma(g)\,
\Big)\wedge\mathcal{R}(W,f)
\label{CSfrac}\eeq
where $\chi_\gamma:G\to\complex$ is the character of the
representation $\gamma$ and $\mathcal{R}(W,f)\in\Omega^{\rm
  even}_{\rm cl}(W;\complex)$. One can immediately read off from the
Wess-Zumino action (\ref{CSfrac}) the Ramond-Ramond charges of
D0-branes, and the state corresponding to the representation $\gamma$
has (fractional) charge
$$
Q_\gamma^{[g]}=\frac{\chi_\gamma(g)}{|\,G^\vee\,|}
$$
with respect to the twisted Ramond-Ramond one-form field
$C^{(1)}_{[g]}$. These charges agree with both those of an open string
disk amplitude computation and a boundary state analysis for
fractional D0-branes~\cite{Diaconescu1999}. Our general formula
(\ref{coupling}) includes also the natural extension to the
Ramond-Ramond couplings of \emph{regular} D-branes which move freely
in the bulk of $X$ under the action of $G$, as well as to other
non-BPS D-brane states such as truncated branes.

\subsection{Linear orbifolds}\label{Linorb2}

We will now ``test'' our definition (\ref{coupling}) on the class of
examples considered in Section~\ref{Linorb1}. These are flat
orbifolds for which there are no non-trivial curvature contributions,
\emph{i.e.}, $\mathcal{R}(W,f)=1$. Let us specialize to the case of
cyclic orbifolds having twist group $G=\zed_n$ with $n\geq d$. In this
case, as $\mathbb{Z}_{n}$ is an abelian group, one has
$\zed_n^\vee=\zed_n$ (setwise) and we can label the non-trivial
twisted sectors of the orbifold string theory on $X$ by
$k=1,\dots,n-1$. The untwisted sector is labelled by $k=0$. We take a
generator $g$ of $\zed_n$ whose action on $V\cong\complex^d$ is given
by
$$
g\cdot\big(z^1\,,\,\dots\,,\,z^d\big):=\big(\zeta^{a_1}\,z^1\,,\,
\dots\,,\,\zeta^{a_d}\,z^d\big) \ ,
$$
where $\zeta=\exp(2\pi\ii/n)$ and $a_1,\dots,a_d$ are integers
satisfying $a_1+\cdots+a_d\equiv0~{\rm mod}~n$.\footnote{Both the
  requirement that the representation $V$ be complex and the form of
  the $G$-action are physical inputs ensuring that the closed string
  background $X$ preserves a sufficient amount of supersymmetry after
  orbifolding.} In this case the action of any element in $\zed_n$ has
only one fixed point, an orbifold singularity at the origin
$(0,\dots,0)$. Hence for any $g\neq e$ one has%%
\begin{displaymath}
X^{g}\cong\mathbb{R}^{p,1}
\end{displaymath}
and the differential complex (\ref{RRcomplex}) of orbifold
Ramond-Ramond fields is given by%%
\begin{displaymath}
\Omega_{\zed_n}^\bullet(X;\real)=\Omega^\bullet(X;\real)\oplus\Big(\,
\bigoplus_{k=1}^{n-1}\,\Omega^\bullet\big(\mathbb{R}^{p,1};\real\big)
\Big) \ .
\end{displaymath}

Consider now a D-brane $(W,E,f)\in\cat{D}_{\rm frac}^{\zed_n}(X)$ with
worldvolume cycle $f(W)\subset{\mathbb{R}^{p,1}}$ placed at the
orbifold singularity, \emph{i.e.},
$f:W\rightarrow{\mathbb{R}^{p,1}\times{(0,\dots,0)}}\subset{X}$. Let
the generator $g$ act on the fibres of the Chan-Paton bundle $E\to W$
in the $n$-dimensional regular representation
$\gamma(g)_{ij}=\zeta^i\,\delta_{ij}$. The action on worldvolume
fermion fields is determined by a lift $\hat\zed_n$ acting on the spinor
bundle $S\to W$. Then the pairing (\ref{coupling}) contains the
following terms. First of all, we have the coupling of the untwisted
Ramond-Ramond fields to $ W$ given by%%
\begin{displaymath}
\int_{ W}\,\tilde C\wedge\Tr\big(\exp(-F^E/2\pi{\ii})\big) \ ,
\end{displaymath}
which is just the usual Wess-Zumino coupling and hence the
Ramond-Ramond charge of a regular (bulk) brane is~$1$ as
expected. Then there are the contributions from the twisted sectors,
which by recalling eq.~(\ref{WgW}) are given by the expression%%
\begin{displaymath}
\int_{ W}\,\frac1n~\sum_{k=1}^{n-1}\,\tilde
C_{k}\wedge\Tr\big(\gamma(g^{k})\,
\exp(-F^E/2\pi{\ii})\big)
\end{displaymath}
where $g^{k}$ is an element of $\zed_n$ of order $k$. Since
$\gamma(g^k)_{ii}=\zeta^{ik}$, the coupling in this case is determined
by a discrete Fourier transform of the fields $\tilde C_k$ over the
group $\zed_n$. The brane associated with the $i$-th irreducible
representation of $\zed_n$ has charge $\zeta^{ik}/n$ with respect to the
Ramond-Ramond field in the $k$-th twisted sector. For $d=2$ and $d=3$
this pairing agrees with and uniformizes the gauge field couplings
computed in~refs.~\cite{douglas1996} and~\cite{Douglas1997},
respectively.

\subsection{An equivariant Riemann-Roch formula\label{EqRRformula}}

Let $X,W$ be smooth compact $G$-manifolds, and $f:
W\rightarrow{X}$ a smooth proper $G$-map. If we want to make sense of
the equations of motion for the Ramond-Ramond field $C$, which is a
quantity defined on the spacetime $X$, then we need to pushforward
classes defined on the brane worldvolume $W$ to classes defined on the
spacetime. This will enable the construction of Ramond-Ramond currents
in Section~\ref{Fluxquant} induced by the background and
D-branes which appear as source terms in the Ramond-Ramond field
equations. Some technical details of the constructions below are
provided in Appendix~C.

Consider first the non-equivariant case $G=e$. Let $\nu\to W$ be the
$\zed_2$-graded bundle (\ref{normbun}), \emph{i.e.}, the KO-theory
class of $\nu$ is the virtual bundle
$[\nu]=f^*[T_X]-[T_W]\in\KO^0(W)$. We assume that $\nu$ is
even-dimensional and endowed with a \spinc structure (this is
automatic if both $X$ and $W$ are spin$^c$). Then, as reviewed in
Appendix~C, one can define the Gysin homomorphism in ordinary K-theory
\begin{displaymath}
f^{\K}_{!}\,:\,\K^\bullet( W)~\longrightarrow~{\K^\bullet(X)} \ .
\end{displaymath}
Using the orientations on $X$ and $W$ one has Poincar\'e duality in
ordinary cohomology, inducing a Gysin homomorphism
$$f_!^{\H}\,:\,\H^\bullet(W;\rat)~\longrightarrow~\H^\bullet(X;\rat) \
, $$ where here we consider the $\zed_2$-grading given by even and odd
degree.

The pushforward homomorphisms in K-theory and in cohomology,
under the conditions stated above, are related by the Riemann-Roch
theorem which states that%%
\beq
\ch\big(f_!^{\K}(\xi)\big)=f^{\H}_{!}\big(\ch(\xi)\,\smile\,\Todd(\nu)^{-1}
\big)
\label{ordGRRthm}\eeq
for any class $\xi$ in $\K^\bullet( W)$. Here $\Todd(E)\in\Omega_{\rm
  cl}^{\rm even}(W;\complex)$ denotes the Todd genus characteristic
class of a hermitean vector bundle $E$ over $W$, whose Chern-Weil
representative is
$$
\Todd(E)=\sqrt{\det\Big(\,\frac{F^E/2\pi\ii}{\tanh\big(F^E/2\pi\ii
\big)}\,\Big)}
$$
where $F^E$ is the curvature of a hermitean connection $\nabla^E$ on
$E$. The Todd class of the $\zed_2$-graded bundle (\ref{normbun}) can
be computed by using multiplicativity, naturality and invertibility to
get $\Todd(\nu)=f^*\Todd(T_X)/\Todd(T_W)$. Thus the Chern character
does not commute with the Gysin pushforward maps, and the defect in
the commutation relation is precisely the Todd genus of the bundle
$\nu$. This ``twisting'' by the bundle $\nu$ over the D-brane
contributes in a crucial way to the Ramond-Ramond current in the
non-equivariant case~\cite{Cheung1997,Minasian1997,Olsen2000}.

Let us now attempt to find an equivariant version of the Riemann-Roch
theorem. As the morphism $f:W\rightarrow X$ is $G$-equivariant, the
$\zed_2$-graded bundle $\nu$ is itself a $G$-bundle with even
$G$-action. We assume that $\nu$ is $\K_G$-oriented. This requirement
is just the Freed-Witten anomaly cancellation
condition~\cite{Freed1999} in this case, generalized to global
worldsheet anomalies for D-branes represented by generic
$G$-equivariant K-cycles. It enables, analogously to the
non-equivariant case, the construction of an equivariant Gysin
homomorphism
\beq\label{eqGysin}
f_!^{\K_G}\,:\,\K^\bullet_G(W)~\longrightarrow~\K^\bullet_G(X) \ .
\eeq

We will demonstrate that, under some very
special conditions, one can construct a complex Bredon cohomology
class which is analogous to the Todd genus and which plays the role of
the equivariant commutativity defect as above. Let us suppose that
the $G$-action has the property that for any element $g\in G$, the
$N_G(g)$-bundle $\nu^g=\nu( W;f)^{g}\to{W^{g}}$ is the $\zed_2$-graded
bundle $\nu^g=f^*|_{W^g}(T_{X^g})\oplus T_{W^g}$ over the immersion
$f|_{W^g}:W^{g}\rightarrow{X^{g}}$ with a $Z_G({g})$-invariant
spin$^c$ structure. Note that these are highly non-trivial conditions,
because for an arbitrary $G$-bundle $E\to{X}$
one is not even guaranteed in general that $E^{g}\to{X^{g}}$ is a vector
bundle, as the dimension of the fibre may jump from point to point. As
a simple example of what can happen,\footnote{We are grateful to
  J.~Figueroa-O'Farrill for suggesting this example to us.}
let $X=\mathbb{R}$ and $G=\mathbb{R}^{+}$ the group of positive reals
under multiplication. Consider the $G$-bundle $X\times{V}\to{X}$ given
by projection onto the first factor, where $V$ is a finite-dimensional
real vector space and the $G$-action is%%
\begin{displaymath}
g\cdot(x,v)=\big(x\,,\,g^{x}\,v\big)
\end{displaymath}
for all $g\in G$. For any $g\neq1$, $(X\times{V})^{g}$ is not a fibre
bundle over $X^g=X$, as the $G$-invariant fibre space over $x=0$ is
$V$ while it is the null vector over any other point.

When $G$ is a finite group, one can apply a construction due to
Atiyah and Segal~\cite{Atiyah1989}. If $E$ is a complex $G$-vector
bundle over $X$, its restriction to the fixed point subspace $X^g$
for any $g\in G$ carries a representation of the normalizer subgroup
$N_G(g)$ fibrewise. We can thus decompose $E|_{X^g}$ into a Whitney
sum of sub-bundles
$E_\alpha=\Hom_g(\id_\alpha,E|_{X^g})\otimes\id_\alpha$ over the
eigenvalues $\alpha\in {\rm spec}(g)\subset\complex$ for the action of
$g$ on the fibres of $E|_{X^g}$, where $\Hom_g(\id_\alpha,E|_{X^g})$ is a
$Z_G(g)$-bundle over $X^g$ and $\id_\alpha$ is the $N_G(g)$-bundle
$X^g\times V_\alpha$ with $V_\alpha$ the corresponding eigenspace. We
define the class
\beq
\phi_g(E)=\sum_{\alpha\in {\rm spec}(g)}\,\alpha\,[E_\alpha]
\label{phigEdef}\eeq
in the ordinary K-theory of $X^g$ with complex
coefficients. By Schur's lemma, every element $h\in Z_G(g)$ commuting
with $g$ acts as a multiple of the identity on the total space of the
bundle $E_\alpha$, and so the class obtained in this way is
$Z_G(g)$-invariant. It follows that the map (\ref{phigEdef}) on ${\rm
  Vect}^\complex_G(X)$ induces a homomorphism
$$
\phi_g\,:\,\K_G^\bullet(X)\otimes\complex~\longrightarrow~
\big(\K^\bullet(X^g)\otimes\complex\big)^{Z_G(g)} \ .
$$
By setting
$$
\phi=\bigoplus_{[g]\in G^\vee}\,\phi_g
$$
we obtain a natural isomorphism leading to the
splitting~\cite{Atiyah1989}
\beq
\K^\bullet_{G}(X)\otimes{\mathbb{C}}\cong{\bigoplus_{[g]\in G^\vee}\,
\big(\K^\bullet(X^{g})\otimes\mathbb{C}\big)^{Z_G({g})}} \ .
\label{eigendecomp}\eeq
The equivariant Chern character (\ref{chCiso}) provides an isomorphism
componentwise between the equivariant K-theory group
(\ref{eigendecomp}) and the complex Bredon cohomology of $X$.

Suppose now that the equivariant Thom class
$\Thom_G(\nu)\in\K^\bullet_{G,\cpt}(\nu)$ can be decomposed according
to the splitting (\ref{eigendecomp}) in such a way that the component
in any subgroup $$\Thom\big(\nu^g\big)~\in~
\big(\K_\cpt^\bullet(\nu^{g})\otimes\complex\big)^{Z_G({g})}$$
coincides with the (ordinary) Thom class of the $\zed_2$-graded bundle
$\nu^{g}\to{W^{g}}$. Under these conditions, the equivariant Gysin
homomorphism (\ref{eqGysin}) constructed in Appendix~C decomposes
according to the splitting%%
\begin{displaymath}
f^{\K_{G}}_{!}=\bigoplus_{[g]\in G^\vee}\,f^{\K}_{g}
\end{displaymath}
where $f^{\K}_{g}$ is the K-theoretic Gysin homomorphism associated
to the smooth proper map
$$f\big|_{W^g}\,:\,W^{g}~\longrightarrow~{X^{g}} \ . $$

Define the characteristic class $\Todd_{G}$ by %
\beq\label{ToddGdef}
\Todd_{G}(\nu):=\bigoplus_{[g]\in G^\vee}\,\Todd\big(\nu^{g}\big)
\eeq
where $\Todd$ is the ordinary Todd genus. This class defines an
element of the even degree complex Bredon cohomology of the brane
worldvolume $W$. Under the conditions spelled out above, we can now
use the equivariant Chern character (\ref{chCiso}) and the usual
Riemann-Roch theorem for each pair $(W^{g},X^{g})$ to prove the
identity%%
\beq
f^{\H_{G}}_{!}\big(\ch^{\mathbb{C}}(\xi)\,\smile_G\,
\Todd_{G}(\nu)^{-1}\big)=
\ch^{\mathbb{C}}\big(f^{\K_{G}}_{!}(\xi)\big)
\label{GRRthm}\eeq
for any class $\xi\in{\K_{G}^\bullet( W)}\otimes{\mathbb{C}}$, as all
quantities involved in the formula (\ref{GRRthm}) are compatible with
the $G$-equivariant decompositions given above.

When the geometric conditions assumed above are not met, the
equivariant Todd class in the formula (\ref{GRRthm}) should be
modified by multiplying it with another equivariant characteristic
class $\Lambda_G(W)$ which reflects non-trivial geometry of the normal
bundles $N_{W^g}$ to the embeddings $W^g\subset W$. This should come
from applying a suitable fixed point theorem to the ordinary
Riemann-Roch formula (\ref{ordGRRthm}), but we have not found a
version which is suitable to our particular equivariant Chern
character in the general case on the category of $G$-spaces. When $f$
is the collapsing map $X\to\pt$, this is the content of the index
theorem used in Section~\ref{Gravcoupl} below. The formula
(\ref{GRRthm}) is, however, directly applicable on the category
$\cat{D}_{\rm frac}^G(X)$ of fractional D-branes. When $G$ is the
cyclic group $\zed_n$ as in Section~\ref{Linorb2} above, one can apply
the Thomasson-Nori fixed point theorem~\cite{Nori2000,Thomasson1992}
to get
\beq
\Lambda_{\zed_n}(W)=\bigoplus_{k=0}^{n-1}\,\zeta^k~\ch\big(\,
\mbox{$\bigwedge_{-1}$}
N_{W^{g^k}}^\vee\big)~\in~\H_{\zed_n}^\bullet\big(W\,;\,
\complex\otimes\underline{R}(-)\big)
\label{TNfixed}\eeq
where
$$
\mbox{$\bigwedge_{-1}$}N_{W^g}=\sum_{l=0}^{{\rm codim}(W^g)}\,
(-1)^l\,\big[\,\mbox{$\bigwedge^l$}N_{W^g}\big]~\in~\K^0\big({W^g}\big)
\ .
$$

\subsection{Gravitational pairings\label{Gravcoupl}}

We will now explain how to compute the curvature contributions
$\mathcal{R}(W,f)\in\Omega^{\rm even}_{G,{\rm cl}}(W;\complex)$ to the
Wess-Zumino functional (\ref{coupling}) for the brane geometries
described in Section~\ref{EqRRformula} above and for vanishing
$B$-field. We derive the cancelling form for the Ramond-Ramond gauge
anomaly inflow due to chiral fermions on the intersection worldvolume
for families of branes using the usual descent
procedure~\cite{Green1996}, which is due to curvature of the spacetime
manifold $X$ itself. For this, we must explicitly use the $G$-spin$^c$
structure on $X$. The standard mathematical intuition behind this
correction is to modify the equivariant Chern character to an isometry
with respect to the natural bilinear pairings on equivariant K-theory
and Bredon cohomology with complex
coefficients~\cite{Minasian1997,Olsen2000}.

The natural sesquilinear pairing between two classes of complementary
degrees in complex Bredon cohomology, represented by closed differential
forms $\omega,\eta\in\Omega^\bullet_{G,{\rm cl}}(X;\complex)$, is given by
$([\omega],[\eta])_{\H_G}:=\int_X^G\,\overline{\omega}\wedge_G\eta$.
On the other hand, the natural quadratic form on fractional branes
defined by classes in equivariant K-theory, represented by complex
$G$-vector bundles $E,F\to X$, is the topological charge
$$\big([E]\,,\,[F]\big)_{\K_G}:= \mu_1\big([X,E^\vee\otimes F,\Id_X]\big)$$ of
eq.~(\ref{mugammadef}) in the untwisted sector corresponding to the
representation $\gamma=1:\complex[G]\to\complex$ induced from the
trivial representation of $G$. This quantity agrees with the natural
intersection form on boundary states computed as the $G$-invariant
Witten index over open string states suspended between
D-branes~\cite{Moore2004}, which counts the difference between the
number of positive and negative chirality Ramond ground states and
hence computes the required chiral fermion anomaly.

The two bilinear forms are related through the local index theorem
which provides a formula for $\Index_1\big([\Dirac_{E}^X]\big)$ in terms
of integrals of characteristic forms over the various singular strata
of the orbifold $[X/G]$. It reads~\cite{Bunke2007}
\bea \nonumber
\Index_1\big([\Dirac_{E}^X]\big)&=&
\sum_{[g]\in G^\vee}\,\frac{2^{d_g}}{|g|}~
\int_{X^g/Z_G(g)}\,\ch(g,E)\wedge\Todd\big(T_{X^g}\big) \\ \nonumber
&& \qquad \qquad \wedge~\int_{\bigwedge^\bullet
  N_{X^g}}~\frac{\STr_S\big(\Delta(g)\big)}
{\sqrt{\det\big(1-N(g)\big)\,\det\big(1-N(g)\exp(-F^{N_{X^g}}/2\pi\ii)
\big)}} \ ,
\eea
where $d_g=\dim(X)-\dim(X^g)$, $|g|$ is the order of the element $g\in
G$, and $N(g)$ denotes the action of $g\in G$ on the fibres of the
normal bundle $N_{X^g}$ to $X^g$ in $X$. The determinant is taken
over the normal bundle $N_{X^g}$ for each $g\in G$, implemented by a
Berezin-Grassmann integration over the exterior algebra bundle
$\bigwedge^\bullet N_{X^g}$. The form $\Delta(g)$, regarded as an
element of $\bigwedge^\bullet N_{X^g}$ under the symbol map, is the
action of $g$ on the fibres of the spinor bundle
$S|_{X^g}$, and $\STr_S$ is the supertrace over the endomorphism
bundle of $S$.

This formula can be thought of as an equivariant
localization of the usual Atiyah-Singer index density onto the
submanifolds $X^g\subset X$ of fixed points of the $G$-action on $X$,
with the determinants reflecting the ``Euler class'' contributions
from the non-trivial normal bundles to $X^g$. This is the anticipated
physical result arising from the closed string twisted Witten index,
computed as the partition function on the cylinder with the twisted
boundary conditions (\ref{twistbc}). Supersymmetry localizes the
computation onto zero modes of the string fields which are constant
maps to the submanifolds $X^h\subset X$, while modular invariance
requires a (weighted) sum over all twisted sectors. The index can be
rewritten using the equivariant characteristic classes defined as in
eqs.~(\ref{chCEdef}) and (\ref{ToddGdef}) to get
\beq
\Index_1\big([\Dirac_{E}^X]\big)=\int_X^G\,\ch^\complex(E)\wedge_G\,
{\Todd_G(T_X)}\wedge_G\,{\Euler_G(X)} \ ,
\label{Index1local}\eeq
where we have used $(T_X)^g=T_{X^g}$ and the element of
$\Omega_{G,{\rm cl}}^\bullet(X;\complex)$ given by
\bea
\Euler_G(X)&:=&\bigoplus_{[g]\in G^\vee}\,\frac{|g|}
{2^{d_g}\,|\,G^\vee\,|}~
\int_{\bigwedge^\bullet N_{X^g}}~\frac{\STr_S\big(\Delta(g)\big)}
{\sqrt{\det\big(1-N(g)\big)\,{\det}
\big(1-N(g)\exp(-F^{N_{X^g}}/2\pi\ii)\big)}} \nonumber\\ &&
\label{EulerXdef}\eea
defines a characteristic class in the complex Bredon
cohomology of $X$. Note that the integrands of eq.~(\ref{EulerXdef})
are formally similar to the Chern characters of the virtual bundles
(\ref{TNfixed}) above.

By using multiplicativity of the equivariant Chern character
(\ref{chCEdef}) to write $$\ch^\complex\big(E^\vee\otimes
F\big)=\overline{\ch^\complex(E)}\,\wedge_G\,\ch^\complex(F)\ , $$ it
follows from eq.~(\ref{Index1local}) that the map (\ref{chCiso}) can
be turned into an isometry by ``twisting'' it with the closed
differential form $\sqrt{{\Todd_G(T_X)}\wedge_G\Euler_G(X)}$, which
when pulled back along $f:W\to X$ gives the required anomaly
cancelling form on the brane worldvolume. This should then be combined
with the correction $\Todd_G(\nu)^{-1}$ contributed by the
$\zed_2$-graded bundle (\ref{normbun}) to the Riemann-Roch formula
(\ref{GRRthm}). Then under the various conditions spelled out in
Section~\ref{EqRRformula} above, the required map $\mathcal{R}$ in
eq.~(\ref{coupling}) from $G$-\spinc bordism classes $[(W,f)]$ to
$\Omega_{G,{\rm cl}}^{\rm even}(W;\complex)$ is given by
\bea
&& \mathcal{R}(W,f)\=\frac{f^*\sqrt{\Todd_G(T_X)\wedge_G\,
{\Euler_G(X)}}}
{\Todd_G(\nu)}\=\Todd_G(T_W)\wedge_G\,f^*\sqrt{\frac{\Euler_G(X)}{
\Todd_G(T_X)}} \ .
\label{calRWf}\eea
The main new ingredient in this formula is the contribution from the
fixed point submanifolds $X^g\subset X$, particularly their normal
bundle characteristic classes (\ref{EulerXdef}). This corrects
previous topologically trivial, flat space formulas, even for
$G$-fixed worldvolumes $W$ (see ref.~\cite{GarciaCompean1998} for
example). Note that when the $G$-action on $X$ is trivial, one has
$\Todd_G(T_X)=\Todd(T_X)$ and $\Euler_G(X)$ is constant.

\newsection{Orbifold differential K-theory\label{OrbdiffKG}}%%

The main drawback of the delocalized theory of the previous section is
that it cannot incorporate the interesting effects of torsion, which
have been one of the driving forces behind the K-theory description of
D-branes and Ramond-Ramond fluxes, and as such it is desirable to have
a description which utilizes the full $R(G)$-module
$\K_G^\bullet(X)$. In this section we will develop an extension of
differential K-theory as defined in~ref.~\cite{Hopkins2005} to
incorporate the case of a $G$-manifold. These are the groups needed to
extend the analysis of the previous section to topologically
non-trivial, real-valued Ramond-Ramond fields. While we do not have a
formal proof that this is a proper definition of an equivariant
differential cohomology theory, we will see that it matches exactly
with expectations from string theory on orbifolds and also has the
correct limiting properties. For this reason we dub the theory that we
define `orbifold' differential K-theory, defering the terminology
`equivariant' to a more thorough treatment of our model (we discuss
this in more detail in Section~\ref{OrbdiffK} below). In the
following we will spell out the definition of differential K-theory
groups. The crux of the extensions of these definitions to the
equivariant setting will be explicit constructions of the exact
sequences they are described by, which are important for physical
considerations. We will determine concrete realizations of the various
morphisms involved, which are given in a general but abstract
framework in~ref.~\cite{Hopkins2005}.\footnote{An explicit proof of
  these exact sequences has been given recently in
  ref.~\cite{Carey2007} using the geometric description of
  differential K-cocyles in terms of bundles with connection. Our
  proof is more general, but less geometric, as it exploits the
  realization in terms of maps to classifying spaces.} See
refs.~\cite{Freed2000,Freed2006a} for an introduction to differential
cohomology theories and their applications in physics.

\subsection{Differential cohomology theories}\label{gendiff}

Differential K-theory of a manifold is an enrichment of its K-theory,
which encodes global topological information, with local geometric
information contained in the de~Rham complex. Consider a
(generalized) cohomology theory $\E^\bullet$ defined on the
category of smooth manifolds $X$ along with a canonical map%%
\beq\label{canmapE}
\varphi\,:\,\E^{\bullet}(X)~\longrightarrow~
{\H(X;\mathbb{R}\otimes{\pi_{-\bullet}\E})}^{\bullet}
\eeq
which induces an isomorphism%%
\begin{displaymath}
\E^{\bullet}(X)\otimes\mathbb{R}\cong
{\H(X;\mathbb{R}\otimes{\pi_{-\bullet}\E})}^{\bullet} \ ,
\end{displaymath}
\emph{i.e.}, the image of $\varphi$ is a full lattice and its kernel
is the torsion subgroup of $\E^\bullet(X)$. Then one can define
\emph{differential $\E$-theory} as the cohomology
theory $\check{\E}{}^\bullet$ which lifts $\E^\bullet$ via the
pullback square%%
\begin{equation}\label{diffdiag}
\xymatrix{\check{\E}^{\bullet}(-) ~\ar[r]~ \ar[d] & ~
\Omega_{\rm cl}(-;\mathbb{R}\otimes{\pi_{-\bullet}}\E)^{\bullet}
\ar[d]\\
\E^{\bullet}(-)
~\ar[r]_{\!\!\!\!\!\!\!\!\!\!\!\!\!\!\!\!\!\!\!\varphi}~ & ~
{\H(-;\mathbb{R}\otimes{\pi_{-\bullet}\E})}^{\bullet}} \ ,
\end{equation}
where $\Omega_{\rm cl}(X;\mathbb{R}\otimes{\pi_{-\bullet}}\E)^q$
denotes the real vector space of closed $\E^\bullet(\pt;\real)$-valued
differential forms $\omega$ on $X$ of total degree $q$, and the right
vertical map in the commutative diagram (\ref{diffdiag}) is given by
sending $\omega$ to its de~Rham cohomology class $[\omega]_{\rm
  dR}$. A class in $\check{\E}^{q}(X)$ is given by a pair
$(\xi,\omega)$, with $\xi\in\E^{q}(X)$ such that
\beq
\varphi(\xi)=[\omega]_{\rm dR} \ ,
\label{varphiomega}\eeq
together with an \emph{isomorphism} that realizes the equality
(\ref{varphiomega}) explicitly in
${\H(X;\mathbb{R}\otimes{\pi_{-\bullet}\E})}^{q}$.

In their foundational paper~\cite{Hopkins2005} Hopkins and Singer
define the differential $\E$-theory associated to any generalized
cohomology theory $\E^\bullet$, and prove its naturality and homotopy
properties. This is done by generalizing the concept
of \emph{function space} in algebraic topology, which can be used to
define the cohomology of a space, to that of \emph{differential
  function space}, where here the term ``differential'' typically means
something different from differentiable or smooth. Because of this,
the differential $\E$-groups are defined in an abstract way and
are difficult to realize explicitly. An explicit construction for
differential K-theory is given in~ref.~\cite{Hopkins2005}. In the following
we will go through this construction in some detail. This will be our
starting point to give a definition of the differential cohomology
theory associated to equivariant K-theory $\K_G^\bullet$, which will
reduce to ordinary differential K-theory in the case where the group
$G$ is the trivial group. The validity of our definition will be
confirmed by explicit construction of the exact sequences, that will
also be important in our later physical applications.

\subsection{Differential K-theory}\label{diffK}

Throughout $X$ will denote a smooth manifold. Let $\Fred$ be the
algebra of Fredholm operators on a separable Hilbert space. Recall
that $\Fred$ is a classifying space for complex K-theory through the
isomorphism%%%%
\begin{displaymath}
[X,\Fred]~\xrightarrow{\text{Index}(-)}~\K^{0}(X)
\end{displaymath}
which associates to any map $f:X\to\Fred$ the index bundle of $f$ in
$\K^{0}(X)$. Let%%
\begin{displaymath}
u~\in~{Z^{\rm even}(\Fred;\mathbb{R})}
\end{displaymath}
be a cocycle of even degree which represents the Chern character of
the universal bundle. Then for any map $f:X\to{\Fred}$
representing a complex vector bundle $E\to X$, the pullback $f^{*}u$
is a representative of $\ch(E)$ in $\H^{\rm even}(X;\mathbb{R})$.

The differential K-theory group $\check{\K}^{0}(X)$ is defined to be
the set of triples $(c,h,\omega)$, where $c:X\to\Fred$,
$\omega$ is a closed differential form in $\Omega_{\rm cl}^{\rm
  even}(X;\mathbb{R})$, and $h$ is a cochain in $C^{{\rm
    even}-1}(X;\mathbb{R})$ satisfying%5
\beq
\delta{h}=\omega-c^{*}u \ .
\label{hcochain}\eeq
The cochain $h$ in eq.~(\ref{hcochain}) is precisely the isomorphism
refered to in Section~\ref{gendiff} above, which is invisible in the
cohomology groups, and in this equation the closed differential form
$\omega$ is regarded as a cochain under the de~Rham map
$\omega\mapsto\int_{(-)}\,\omega$. Two triples
$(c_{0},h_{0},\omega_{0})$ and $(c_{1},h_{1},\omega_{1})$ are
equivalent if there exists a triple $(c,h,\omega)$ on $X\times[0,1]$,
with $\omega=\omega(t)$ constant along $t\in[0,1]$, such that%%
\begin{equation}
(c,h,\omega)\big|_{t=0}\=(c_{0},h_{0},\omega_{0}) \qquad \mbox{and}
\qquad (c,h,\omega)\big|_{t=1}\=(c_{1},h_{1},\omega_{1}) \ .
\label{chomegaequiv}\end{equation}

The equivalence (\ref{chomegaequiv}) can be rephrased~\cite{Freed2000}
by requiring that there exists a map
$$F\,:\,{X}\times[0,1]~\longrightarrow~\Fred$$ and a differential form
$\sigma\in\Omega^{{\rm even}-2}(X;\mathbb{R})$ such that%%
\begin{eqnarray}
F\big|_{t=0}&=&c_{0} \ , \nonumber\\[4pt]
F\big|_{t=1}&=&c_{1} \ , \nonumber\\[4pt]
\omega_{1}&=&\omega_{0} \ , \nonumber\\[4pt]
h_{1}&=&h_{0}+\pi_{*}\,F^{*}u+\dd\sigma
\label{equivaltrels}\end{eqnarray}
where $\pi:{X}\times[0,1]\to{X}$ is the natural projection. The relations
(\ref{equivaltrels}) say that $c_{0}$ and $c_{1}$ are homotopic maps,
hence they represent the same class in $\K^{0}(X)$, and that the
cochains $h_{0}$ and $h_{1}$ are related by the homotopy that connects
the maps $c_0$ and $c_1$. We also see that the closed form $\omega$
completely characterizes the triple $(c,h,\omega)$.

Borrowing terminology used in representing classes in the differential
cohomology $\check{\H}^{2}(X)$ as principal ${\rm U}(1)$-bundles with
connection, the class $[c]\in{\K^{0}(X)}$ is called the
\emph{characteristic class}, the closed differential form
$\omega$ is called the \emph{curvature}, while the cochain $h$ is
called the \emph{holonomy} of the triple. From the
defining property of the universal cocycle $u$ and
eq.~(\ref{hcochain}) it follows that%%
\begin{displaymath}
\ch\big([c]\big)=[\omega]_{\rm dR} \ .
\end{displaymath}
Thus the cohomology class represented by the curvature $\omega$ lies
in the image of the (real) Chern character, which is a lattice of
maximal rank inside the cohomology group with real coefficients.

Let us now define the differential K-theory group
$\check{\K}^{-1}(X)$. Recall that the classifiying space for $\K^{-1}$
is the based loop space $\Omega\Fred$. Thus we need a cocycle%%
\begin{displaymath}
u^{-1}~\in~{Z^{\rm odd}(\Omega\Fred;\mathbb{R})}
\end{displaymath}
which represents the universal odd Chern character. Consider the
evaluation map%%
\begin{displaymath}
\ev\,:\,{\Omega\Fred}\times\S^1~\longrightarrow~\Fred \ .
\end{displaymath}
Then the cocycle $u^{-1}$ is defined by%%
\begin{displaymath}
u^{-1}=\Pi_{*}~\ev^{*}u
\end{displaymath}
where $\Pi:\Omega\Fred\times\S^1\to\Omega\Fred$ is the natural
projection. In fact, $u^{-1}$ can be defined as the \emph{slant product} of
$\ev^{*}u$ with the fundamental class of the circle $\S^1$,
\emph{i.e.}, by integrating the cocycle $\ev^{*}u$ along $\S^1$.
As above, a class in $\check{\K}^{-1}(X)$ is represented by a triple
$(c,h,\omega)$, where $c:X\to\Omega\Fred$, $\omega$ is a closed
differential form in $\Omega_{\rm cl}^{{\rm even}-1}(X;\mathbb{R})$,
and $h$ is a cocycle in $C^{{\rm even}-2}(X;\mathbb{R})$ satisfying
$$\delta{h}=\omega-c^{*}u^{-1} \ . $$ Two triples
$(c_{0},h_{0},\omega_{0})$ and $(c_{1},h_{1},\omega_{1})$ are
equivalent if a relation like that in eq.~(\ref{chomegaequiv}) holds.

In an analogous way one can define the higher differential K-theory
groups $\check{\K}^{-n}(X)$ for any positive integer $n$. One can
prove that Bott periodicity in complex K-theory induces a periodicity
in differential K-theory given by%%
\begin{displaymath}
\check{\K}^{-n}(X)\cong\check{\K}^{-n-2}(X) \ .
\end{displaymath}
This periodicity enables one to define the higher differential
K-theory groups in positive degrees. The group composition law on
$\check{\K}^{-n}(X)$ is given by%%
\begin{displaymath}
(c_{1},h_{1},\omega_{1})+(c_{2},h_{2},\omega_{2}):=
(c_{1}\cdot{c_{2}},h_{1}+h_{2},\omega_{1}+\omega_{2})
\end{displaymath}
where the dot denotes pointwise multiplication. The identity element
is given by the triple $(\,\underline{c}\,,0,0)$, where throughout
$\underline{c}$ denotes any map which is homotopic to the (constant)
identity map. To allow for the presence of the characteristic class
$\omega$ in the definition, the abelian groups $\check{\K}^{-n}(X)$ are
generally infinite-dimensional. The definition of
$\check{\K}^{-n}(X)$ depends, up to homotopy type and cohomology
class, on the choice of classifying space and of universal cocycle
$u$~\cite{Hopkins2005}.

A key property is the exact sequences which characterize the
differential K-theory groups $\check{\K}^{-n}(X)$ for any $n\in\zed$
as extensions of topological K-theory by certain groups of
differential forms. In each case the differential K-theory group
$\check{\K}^\bullet(X)$ is an extension of the setwise fibre product
$$A_\K^\bullet(X)=\big\{(\xi,\omega)\in
\K^\bullet(X)\times\Omega_{\rm
  cl}(X;\real\otimes\pi_{-\bullet}\K)^\bullet~\big|~\ch(\xi)=
[\omega]_{\rm dR}\big\}$$ by the torus of
topologically trivial flat fields given by
\beq
0~\longrightarrow~\K^{\bullet-1}(X)\otimes\real/\zed~\longrightarrow~
\check{\K}^\bullet(X)~\longrightarrow~A_\K^\bullet(X)~
\longrightarrow~0 \ .
\label{diffKexseqgen}\eeq
This will be useful below when we define
equivariant differential K-theory. 

As in the case of topological K-theory, there are geometrical
realizations of the groups $\check{\K}^{-n}(X)$~\cite{Freed2000}. In
particular, a class in $\check{\K}^{0}(X)$ can be represented by a
complex vector bundle $E\to X$ equiped with a connection
$\nabla^E$. To the pair $(E,\nabla^E)$ we can associate the triple
$(f,\eta,\omega)$, where $f:X\to{BU}$ is a map which classifies the
bundle $E$, $\omega=\ch(\nabla^E)$ is a Chern-Weil representative of
the Chern character of $[E]$, and $\eta$ is a Chern-Simons form such
that $\dd\eta=f^{*}\omega_{BU}-\omega$ with
$\omega_{BU}=\ch(\nabla_{BU})$ the Chern character form of the
universal bundle $E_{BU}\to BU$ with respect to the universal
connection $\nabla_{BU}$ on $E_{BU}$.

In the following we will define abelian groups that can be thought of as a
natural generalization of the differential K-theory of a manifold $X$
acted upon by a (finite) group $G$. In this case one cannot employ the
powerful machinery developed in~ref.~\cite{Hopkins2005}, as the equivariant
K-theory $\K^\bullet_{G}(X)$ is not a cohomology theory defined on the
category of manifolds. Instead, we will take as our starting point the
explicit definition of the groups $\check{\K}^{-n}(X)$ given above,
and naturally generalize it to groups $\check{\K}_{G}^{-n}(X)$ which
accomodate the action of the group $G$ in such a way that when $G={e}$
is trivial, one has $\check{\K}_{G}^{-n}(X)\cong{\check{\K}^{-n}(X)}$.
\subsection{Orbifold differential forms\label{Orbdiffs}}

We want to generalize the commutative diagram (\ref{diffdiag}) to the
case in which our underlying \emph{cohomology} theory $\E^\bullet(X)$
is the equivariant K-theory $\K^\bullet_{G}(X)$. We first need
a homomorphism $\varphi$ from equivariant K-theory to a target
cohomology theory which induces an isomorphism when tensored over the
reals. For this, we will use the Chern character constructed in
Section~\ref{section2} with the target cohomology theory given by
Bredon cohomology. Then we need a refinement of this cohomology which
reduces to the de~Rham complex when the group $G$ is trivial. This
complex may be thought of as the complex of \emph{differential forms
  on the orbifold $X/G$}. For this purpose, we will use the
differential complex $(\Omega^\bullet_{G}(X;\real),\dd_G)$ defined in
Section~\ref{clstring}. Using the delocalization formula
(\ref{Bredonsplit}) one shows that this complex is a refinement for
Bredon cohomology with real coefficients, in the case when $G$ is a
finite group. It comes equiped with a natural product defined in
eq.~(\ref{orbdiffprod}). As a refinement for Bredon cohomology, the
complex $\Omega_{G}^\bullet(X;\mathbb{R})$ gives a well-defined map%%
\begin{displaymath}
\omega~\longmapsto~{[\omega]_{G-{\rm
      dR}}}\in\H\big(\Omega_{G}^{\bullet}(X;
\mathbb{R})\,,\,\dd_G\big)\cong\H_G^\bullet\big(X\,;\,\real\otimes
\underline{R}(-)\big)
\end{displaymath}
and reduces to the usual de~Rham complex of differential forms in the
case $G={e}$.
 
There is an alternative complex one could construct which is
``manifestly'' equivariant, in the sense that its functoriality
property over the category of groups is transparent. It can also be
generalized to the case in which $G$ is an infinite discrete
group. However, it is not evident how to define a ring
structure on this complex, and its physical relation to Ramond-Ramond
fields is not clear. We include its definition here for
completeness.\footnote{We are grateful to W.~L\"uck for suggesting
  this construction to us.} See Appendix~A and~ref.~\cite{Dieck1987}
for the relevant definitions concerning modules over a functor
category and their tensor products.

Starting from the real representation ring functor
$\underline{\mathbb{R}}(-)=\mathbb{R}\otimes{\underline{R}}(-)$ over
the orbit category $\ocat{G}$, there is a natural map of real vector
spaces%%
\beq
\underline{\mathbb{R}}(-)
\otimes_{\mathbb{R}\ocat{G}}{\underline{C}^\bullet
(X;\mathbb{R})}~\longrightarrow~
{\Hom_{\real\ocat{G}}\big(\,\underline{C}\,_\bullet(X;
\mathbb{R})\,,\,\underline{\mathbb{R}}(-)\big)}
\label{realrepmap}\eeq
where $\underline{C}^\bullet(X;\mathbb{R})$ is the left
$\mathbb{R}\ocat{G}$-module obtained by dualizing the functor
$$\underline{C}\,_\bullet(X;\mathbb{R}):=
\real\otimes\underline{C}\,_\bullet(X)$$ defined in
Section~\ref{bredon}. Note that both
$\underline{C}\,_\bullet(X;\mathbb{R})$ and
$\underline{\mathbb{R}}(-)$, being contravariant functors, are right
$\mathbb{R}\ocat{G}$-modules. The map (\ref{realrepmap}) is
given on decomposable elements as%%
\begin{displaymath}
\lambda\otimes{f}~\longmapsto~\big(\sigma\mapsto{f(\sigma)_*
(\lambda)}\big)
\end{displaymath}
and it is an isomorphism of real vector spaces.\footnote{In general,
  to have an isomorphism one has to require the $G$-manifold $X$ to be
  cocompact and proper.}

Define the differential complex%%
\begin{displaymath}
\Omega_{G}^\bullet\big(X\,;\,\mathbb{R}\otimes\underline{R}(-)\big):=
\underline{\mathbb{R}}(-)
\otimes_{\mathbb{R}\ocat{G}}{\underline{\Omega}^\bullet(X;\mathbb{R})}
\end{displaymath}
where $\underline{\Omega}^\bullet(X;\mathbb{R})$ is the functor
$\ocat{G}\to\cat{Ab}$ given by $\underline{\Omega}^\bullet(X;\mathbb{R}):
G/H\mapsto\Omega^\bullet(X^{H};\mathbb{R})$, and with derivation
$\dd_{\rm orb}$ induced by the exterior derivative $\dd$. Since the
de~Rham map induces a chain homotopy equivalence of left
$\mathbb{R}\ocat{G}$-complexes $\underline{C}^\bullet(X;\mathbb{R})\to
\underline{\Omega}^\bullet(X;\mathbb{R})$, there is a $G$-equivariant
chain homotopy equivalence%%
\begin{displaymath}
\underline{\mathbb{R}}(-)\otimes_{\mathbb{R}\ocat{G}}
{\underline{C}^\bullet(X;\mathbb{R})}~\longrightarrow~
\underline{\mathbb{R}}(-)\otimes_{\mathbb{R}\ocat{G}}
{\underline{\Omega}^\bullet(X;\mathbb{R})} \ .
\end{displaymath}
Combined with the isomorphism (\ref{realrepmap}) we can thus conclude
\begin{displaymath}
\H^{\bullet}_{G}\big(X\,;\,\mathbb{R}\otimes\underline{R}(-)\big)
\cong{\H\big(\Omega_{G}^{\bullet}(X;\mathbb{R}\otimes\underline{R}(-))\,,\,
\dd_{\rm orb}\big)} \ .
\end{displaymath}
If one chooses to work with this complex, then the construction of
orbifold differential K-theory groups given in Section~\ref{OrbdiffK}
below can be carried through in exactly the same way. But since the
two complexes $\Omega^\bullet_{G}(X;\real)$ and
$\Omega^\bullet_{G}(X;\real\otimes\underline{R}(-))$ are in general
\emph{not} isomorphic, the two differential cohomology theories
obtained will be generically distinct.

\subsection{Orbifold differential K-groups\label{OrbdiffK}}

Having sorted out all the ingredients necessary to make sense of a
generalization of the diagram~(\ref{diffdiag}), we will now define the
differential equivariant K-theory groups
$\check{\K}_{G}^{-n}(X)$. First, let us recall some further basic
facts about equivariant K-theory. Similarly to ordinary K-theory, a
model for the classifying space of the functor $\K_{G}^{0}$ is given
by the $G$-algebra of Fredholm operators $\Fred_{G}$ acting on a
separable Hilbert space which is a representation space for $G$ in
which each irreducible representation occurs with infinite
multiplicity~\cite{Atiyah1967}. Then there is an isomorphism
\begin{displaymath}
\K^{0}_{G}(X)\cong\left[X,\Fred_{G}\right]_{G}
\end{displaymath}
where $[-,-]_{G}$ denotes the set of equivalence classes of
$G$-homotopic maps, and the isomorphism is given by taking the
index bundle.

There is also a universal space $\text{Vect}^{n}_{G}$, equiped with a
universal $G$-bundle $\widetilde{E}^{n}_{G}$, such that
$\left[X,\text{Vect}^{n}_{G}\right]_{G}$ corresponds to the set of
isomorphism classes of $n$-dimensional $G$-vector bundles over
$X$~\cite{Luck1998}. These spaces are constructed as follows. Let
$\cat{E}G$ be the category whose objects are the elements of $G$ and
with exactly one morphism between each pair of objects. The geometric
realization (or nerve) of the set of isomorphism classes in $\cat{E}G$
is, as a simplicial space, the total space of the classifying principal
$G$-bundle $EG$. With $\Vect^n(\pt)$ the category of $n$-dimensional
complex vector spaces $V\cong\complex^n$, the universal space
$\text{Vect}^{n}_{G}$ is defined to be the geometric realization of
the functor category $[\cat{E}G,\Vect^n(\pt)]$ (see Appendix~A). The
universal $n$-dimensional $G$-vector bundle $\widetilde{E}_G^n$ is
then defined as
\beq
\widetilde{E}_G^n=\widetilde{\Vect}{}_G^n
\times_{{\rm GL}(n,\complex)}
\,\complex^n~\longrightarrow~\Vect_G^n \ ,
\label{univGbun}\eeq
where $\widetilde{\Vect}{}_G^n$ is the geometric realization of the
functor category defined as above but with $\Vect^n(\pt)$ replaced
with the category consisting of objects $V$ in $\Vect^n(\pt)$ together
with an oriented basis of $V$.

We assume sufficient regularity conditions on the
infinite-dimensional spaces $\Fred_{G}$ and
$\widetilde{E}^{n}_{G}$. Since $\Fred_G$ and the group completion
$\Omega B\Vect_G$ are both classifying spaces for equivariant
K-theory, they are $G$-homotopic and we can thereby choose a cocycle
$$u_{G}~\in~{Z_{G}^{\rm even}(\Fred_{G};\real)}$$ representing the
equivariant Chern character of the universal $G$-bundle
(\ref{univGbun}). Generally, the group $Z_{G}^{\rm even}(X;\real)$ is
the subgroup of closed cocycles in the complex%%
\beq
C_{G}^{{\rm even}}(X;\mathbb{R}):=\bigoplus_{[g]\in G^\vee}\,
C^{\rm even}\big(X^{g}\,;\,
\mathbb{R}\big)^{Z_{G}(g)}
\label{CGevenXR}\eeq
which, by the results of Section~\ref{DelocBredon}, is a cochain model
for the Bredon cohomology group $\H_{G}^{{\rm even}-1}(X;\mathbb{R}
\otimes{\underline{R}}(-))$. The equivariant Chern character is
understood to be composed with the delocalizing isomorphism of
Section~\ref{DelocBredon}. Since it is a natural homomorphism, for any
$G$-bundle $E\to X$ classified by a $G$-map $f:X\to\Fred_G$ one has%%
\begin{displaymath}
\ch_{X}\big([E]\big)=\big[f^{*}u_{G}\big] \ .
\end{displaymath}

\begin{definition}
The \emph{orbifold differential K-theory}
$\check{\K}^{0}_{G}(X)$ of the (global) orbifold $[X/G]$ is the group
of triples $(c,h,\omega)$, where $c:X\to\Fred_{G}$ is a $G$-map,
$\omega$ is an element in $\Omega_{G,{\rm cl}}^{\rm
  even}(X;\mathbb{R})$, and $h$ is an element
in $C_{G}^{{\rm even}-1}(X;\mathbb{R})$
satisfying%%
\beq\label{Eqdeltah}
\delta{h}=\omega-c^{*}u_{G} \ .
\eeq
Two triples $(c_{0},h_{0},\omega_{0})$ and $(c_{1},h_{1},\omega_{1})$
are said to be \emph{equivalent} if there exists a triple
$(c,h,\omega)$ on $X\times[0,1]$, with the group $G$ acting trivially
on the interval $[0,1]$ and with $\omega$ constant along $[0,1]$, such
that
$$
(c,h,\omega)\big|_{t=0}\=(c_{0},h_{0},\omega_{0}) \qquad \mbox{and}
\qquad (c,h,\omega)\big|_{t=1}\=(c_{1},h_{1},\omega_{1}) \ .
$$
\label{orbdiffKdef}\end{definition}

In eq.~(\ref{Eqdeltah}) the closed orbifold differential form $\omega$
is regarded as an orbifold cochain in the complex (\ref{CGevenXR}) by
applying the de~Rham map componentwise on the fixed point submanifolds
$X^g$, $g\in G$. The higher orbifold differential K-theory groups
$\check{\K}^{-n}_G(X)$ are defined analogously to those of
Section~\ref{diffK} above. To confirm that this is an appropriate
extension of the ordinary differential K-theory of $X$, we should show
that the orbifold differential K-theory groups fit into exact
sequences which reduce to those given by
eq.~(\ref{diffKexseqgen}) when $G$ is taken to be the trivial
group. For this, we define the group%%
\begin{displaymath}
A_{\K_G}^{0}(X):=\left\{(\xi,\omega)\in \K_G^{0}(X)\times\Omega_{G,{\rm
    cl}}^{\rm even}(X;\mathbb{R}
)~\big|~ \ch_X(\xi)=[\omega]_{G-{\rm dR}}\right\} \ .
\end{displaymath}

\begin{theorem}
The orbifold differential K-theory group $\check{\K}^{0}_{G}(X)$
satisfies the exact sequence%%
\bea
&&0~\longrightarrow~{\frac{\H_{G}^{{\rm even}-1}\big(X\,;\,\mathbb{R}
\otimes{\underline{R}}(-)\big)}{\ch_X\big({\K_{G}^{-1}(X)}\big)}}~
\longrightarrow~{\check{\K}_{G}^{0}(X)}~\longrightarrow~
{A^{0}_{\K_{G}}(X)}~\longrightarrow~{0} \ .
\label{exseqeqgen}\eea
\label{exseqthm}\end{theorem}

\begin{proof}
Consider the subgroup of $\H_G^{{\rm
    even}-1}(X;\mathbb{R}\otimes\underline{R}(-))$ defined as the
image of the equivariant K-theory group $\K_G^{-1}(X)$ under the Chern
character $\ch_X$. It consists of Bredon cohomology classes of the
form $[\tilde{c}^{*}u_G^{-1}]$, where
$\tilde{c}:X\to{\Omega\Fred_G}$. There is a surjective map%%
\begin{eqnarray*}
f\,:\,\check{\K}_G^{0}(X) &\longrightarrow& A_{\K_G}^{0}(X) \\
\big[(c,h,\omega)\big] &\longmapsto& \big([c]\,,\,\omega\big)
\end{eqnarray*}
which is a well-defined homomorphism, \emph{i.e.}, it does not depend
on the chosen representative of the orbifold differential K-theory
class. By definition, the kernel of $f$ consists of triples of the form
$(\,\underline{c}\,,h,0)$. We also define the map
\begin{eqnarray*}
g\,:\,\H^{{\rm even}-1}_G\big(X\,;\,\mathbb{R}\otimes\underline{R}(-)
\big)&\longrightarrow& {\check{\K}_G^{0}(X)}\\
\left[h\right] &\longmapsto& \big[(\,\underline{c}\,,h,0)\big] \ ,
\end{eqnarray*}
which is a well-defined homomorphism because the class
$\left[(\,\underline{c}\,,h,0)\right]$ depends only on the Bredon
cohomology class $[h]\in{\H_G^{{\rm
      even}-1}(X;\mathbb{R}\otimes\underline{R}(-))}$.
Then by construction one has $\im({g})=\ker({f})$.

The homomorphism $g$ is not injective. To determine the kernel of $g$, we
use the fact that the zero element in $\check{\K}^{0}_G(X)$ can be
represented as%%
\begin{displaymath}
\big[(\,\underline{c}\,,0,0)\big]=
\big[(\,\underline{c}\,,\pi_{*}\,F^{*}u_G+\dd_{G}\sigma,0)\big]
\end{displaymath}
with $F:X\times{\S^1}\to\Fred_G$ and $\sigma\in\Omega^{{\rm
    even}-2}_G(X;\real)$ (see
eq.~(\ref{equivaltrels})). To the map $F$ we can associate a map
$\tilde{c}:X\to\Omega\Fred_G$ such that
$F=\ev\circ(\tilde{c}\times{\Id_{\S^1}})$. This follows from the
isomorphism
\begin{displaymath}
\K_G^{-1}(X)\cong\ker\left({i^{*}}:\K_G^{0}(X\times{\S^1})\to{\K_G^{0}(X)}
\right)
\end{displaymath}
where $i$ is the inclusion
$i:X\hookrightarrow{X\times{\pt}}\subset{X\times{\S^1}}$. Now use the
fact that at the level of (real) Bredon cohomology one has an
equality%%
\begin{displaymath}
\pi_{*}\,\big(\tilde{c}\times{\Id_{\S^1}}\big)^{*}=\tilde{c}{}^{*}\,\Pi_{*}
\end{displaymath}
since the projection homomorphisms $\pi_{*}$ and $\Pi_{*}$ both
correspond to integration (slant product) along the $\S^1$ fibre. Then
one has the identity%%
\begin{displaymath}
\big[\pi_{*}\,F^{*}u_G\big]\=
\big[\pi_{*}\,(\tilde{c}\times{\Id_{\S^1}})^{*}~
\ev^{*}u_G\big]\=\big[\tilde{c}^{*}\,\Pi_{*}~\ev^{*}u_G\big]\=
\big[\tilde{c}^{*}u_G^{-1}\big] \ .
\end{displaymath}
It follows that $\ker({g})$ is exactly the group
$\ch_X({\K_G^{-1}(X)})$, and putting everything together we arrive at
eq.~(\ref{exseqeqgen}).
\end{proof}

The torus $\H_G^{{\rm
    even}-1}\big(X;\mathbb{R}\otimes\underline{R}(-)\big)/
\ch_X\big(\K_G^{-1}(X)\big)\cong \K_G^{-1}(X)\otimes\real/\zed$
is called the group of \emph{topologically trivial flat fields} (or of
``orbifold Wilson lines''). We can rewrite the
sequence~(\ref{exseqeqgen}) in various illuminating ways. Consider the
\emph{characteristic class} map%%
\begin{eqnarray*}
f_{\rm cc}\,:\,\check{\K}_G^{0}(X)&\longrightarrow&{\K_G^{0}(X)}\\
\big[(c,h,\omega)\big]&\longmapsto&\left[{c}\right]
 \end{eqnarray*}
and the map%%
\begin{eqnarray*}
g_{\rm cc}\,:\,\Omega_G^{{\rm even}-1}(X;\mathbb{R})&\longrightarrow&
\check{\K}_G^{0}(X)\\
h&\longmapsto&\big[(\,\underline{c}\,,h,\dd_{G}h)\big] \ .
\end{eqnarray*}
Let $\Omega_{\K_G}^{{\rm
    even}-1}(X;\mathbb{R})$ be the
subgroup of elements in $\Omega_{G,{\rm cl}}^{{\rm
    even}-1}(X;\mathbb{R})$ whose Bredon
cohomology class lies in $\ch_X({\K_G^{-1}(X)})$. Then by using
arguments similar to those used in arriving at the sequence
(\ref{exseqeqgen}), one finds the%%
\begin{cor}[{\bf Characteristic class exact sequence}]
The orbifold differential K-theory group $\check{\K}^{0}_{G}(X)$
satisfies the exact sequence%%
\beq
0~\longrightarrow~\dfrac{\Omega_{G}^{{\rm even}-1}(X;
\mathbb{R})}{\Omega_{\K_{G}}^{{\rm
    even}-1}(X;\mathbb{R})}~
\longrightarrow~{\check{\K}_{G}^{0}(X)}~\longrightarrow~
{\K_{G}^{0}(X)}~\longrightarrow~{0} \ .
\label{charclasseq}\eeq
\label{charclasscorr}\end{cor}

The quotient space of orbifold differential forms in the exact
sequence (\ref{charclasseq}) is called the group of
\emph{topologically trivial fields}. An element of this group is a
globally defined (and hence topologically trivial) gauge potential on
the orbifold $X/G$ up to large (quantized) gauge transformations, with
$\omega$ the corresponding field strength. Finally, consider the
\emph{field strength} map%%
\bea \nonumber
f_{\rm fs}\,:\,\check{\K}_G^{0}(X)&\longrightarrow&
{\Omega_{G,{\rm cl}}^{\rm
    even}(X;\mathbb{R})}\\ 
\big[(c,h,\omega)\big]&\longmapsto&\omega \ .
\label{fieldstrmap}\eea
The kernel of the homomorphism $f_{\rm fs}$ is the group which
classifies the \emph{flat fields} (which are not necessarily
topologically trivial) and is denoted
$\K_{G}^{-1}(X;\mathbb{R}/\mathbb{Z})$. This group will be described
in more detail in the next section, where we shall also conjecture an
essentially purely algebraic definition of
$\K_{G}^{-1}(X;\mathbb{R}/\mathbb{Z})$ which explains the notation. In
any case, we have the

\begin{cor}[{\bf Field strength exact sequence}]
The orbifold differential K-theory group $\check{\K}^{0}_{G}(X)$
satisfies the exact sequence%%
\beq
0~\longrightarrow~{\K_{G}^{-1}(X;\mathbb{R}/\mathbb{Z})}~
\longrightarrow~{\check{\K}_{G}^{0}(X)}~\longrightarrow~
{\Omega_{\K_G}^{\rm even}(X;\mathbb{R})}~\longrightarrow~{0} \ .
\label{diffKexseqGfs}\eeq
\label{fieldstrengthcorr}\end{cor}

Higher orbifold differential K-theory groups satisfy analogous exact
sequences, with the appropriate degree shifts throughout. It is clear
from our definition that one recovers the ordinary differential
K-theory groups in the case of the trivial group $G={e}$, and in this
sense our orbifold differential K-theory may be regarded as its
equivariant generalization. At this point we hasten to add that,
although our groups are well-defined and satisfy desired properties
which are useful for physical applications such as functoriality and
the various exact sequences above, we have not proven that our
orbifold theory is a differential cohomology theory. We have also not
given a definition of what a generic orbifold (or equivariant)
differential cohomology theory is. For instance, it would be
interesting to define a ring structure and an integration on
$\check{\K}^\bullet_{G}(X)$. In particular, the integration requires
knowledge of a relative version of orbifold differential K-theory,
which we have not developed in this paper.

We have also investigated the possibility that the group
$\check{\K}^\bullet_{G}(X)$ reduces to the ordinary differential
K-theory $\check{\K}^\bullet(X/G)$ in the case of a free $G$-action on
$X$, and to $\check{\K}^\bullet(X)\otimes {R}(G)$ in the case of a
trivial group action, as one might naively expect from the analogous
results for equivariant topological K-theory (the equivariant excision
theorem (\ref{eqexcision}) with $N=G$ and eq.~(\ref{KGWtrivial}),
respectively) and for Bredon cohomology
(Examples~\ref{Bredonfreeactionex} and~\ref{Bredontrivactionex},
respectively). On the contrary, these decompositions do not occur,
because the corresponding isomorphisms in equivariant
K-theory are estabilished via the induction maps and these usually do
not lift at the ``cochain level'' as isomorphisms. Properties such as
induction structures reflect homotopy invariance of topological
cohomology groups, which is not possessed by differential cohomology
groups due to their ``local'' dependence on the complex of
differential forms. We will see an explicit example of this in the
next section. With this in mind, it would be interesting then to
define a suitable analog of the induction structures in an equivariant
cohomology theory. These and various other interesting mathematical
issues surrounding the orbifold differential K-theory groups that we
have defined will not be pursued in this paper. 

\newsection{Flux quantization of orbifold Ramond-Ramond
  fields\label{Fluxquant}}

In this final section we will argue that the orbifold
differential K-theory defined in the previous section can be used to
describe Ramond-Ramond fields and their flux quantization condition in
orbifolds of Type~II superstring theory with vanishing $H$-flux.
To formulate the self-duality property of orbifold Ramond-Ramond
fields in equivariant K-theory, one needs an appropriate equivariant
version of Pontrjagin duality~\cite{Freed2006a}. This appears to be a
very deep and complicated problem, and is beyond the scope of the
present paper. Furthermore, to generalize the pairing of
Section~\ref{WZpairing} to topologically non-trivial Ramond-Ramond
fields, one needs to define an integration on the orbifold
differential cohomology theory defined in Section~\ref{OrbdiffK}, and
regard the Ramond-Ramond fields properly as cocycles for it. In
addition, one needs a graded ring structure and an appropriate
groupoid representing the orbifold differential K-theory, whose
objects are the Ramond-Ramond form gauge potentials $C$ and whose
isomorphism classes are the gauge equivalence classes in
$\check{\K}_G^\bullet(X)$. Lacking these ingredients, most of our
analysis in this section will be essentially purely
``topological''. We shall study the somewhat simpler problem of the
proper K-theory quantization of orbifold Ramond-Ramond fields, in
particular due to their sourcing by fractional D-branes, in terms of
the formulation provided by orbifold differential K-theory.

\subsection{Ramond-Ramond currents}

We will begin by rephrasing the relation between the D-brane charge
group and the group of Ramond-Ramond fluxes ``measured at infinity''
in the equivariant case, which is a statement about the K-theoretic
classification of Ramond-Ramond fields on a global orbifold
$[X/G]$. For this, we invoke an argument due to Moore and
Witten~\cite{Moore2000} which will suggest that the equivariant Chern
character $\ch_X$ constructed in Section~\ref{EqChern} gives the
right quantization rule for orbifold Ramond-Ramond fields. Suppose
that our spacetime $X$ is a non-compact $G$-manifold. Suppose further
that there are D-branes present in Type~II superstring theory on
$X/G$. Their Ramond-Ramond charges are classified by the equivariant
K-theory $\K^i_{G,\cpt}(X)$ with compact support, where $i=0$ in
Type~IIB theory and $i=-1$ in Type~IIA theory.

We require that the brane be a source for the equation of motion for
the total Ramond-Ramond field strength $\omega$. This means that it
creates a Ramond-Ramond current $J$. If we require that the
worldvolume $W$ be compact in equivariant K-cycles
$(W,E,f)\in\cat{D}^G(X)$, then $J$ is supported in the interior
$\mathring{X}$ of $X$. Let $X_\infty$ be the ``boundary of $X$ at
infinity'', which we assume is preserved by the action of $G$. Then
$\K^\bullet_{G,\cpt}(X)\cong{\K^\bullet_{G}(X,X_\infty)}$. Since $J$ is
trivialized by $\omega$ in $\mathring{X}$, the D-brane charge lives in
the kernel of the natural forgetful homomorphism%%
\beq\label{forgetmap}
\mathfrak{f}^\bullet\,:\,\K^\bullet_{G,\cpt}(X)~\longrightarrow~
{\K^\bullet_{G}(X)}
\eeq
induced by the inclusion
$(X,\emptyset)\hookrightarrow(X,X_\infty)$. We denote by
$i:X_\infty\hookrightarrow X$ the canonical inclusion.

The long exact sequence for the pair $(X,X_\infty)$ in equivariant
K-theory truncates, by Bott periodicity, to the six-term exact
sequence%%
\begin{displaymath}
\xymatrix{\K_{G}^{-1}(X_\infty)~\ar[r]&~{\K_{G}^{0}(X,X_\infty)}~
\ar[r]^{~~~~~~ \ \ \mathfrak{f}^0}& ~\K_{G}^{0}(X)\ar[d]^{i^*}\\
\K_{G}^{-1}(X)~\ar[u]^{i^*}&\ar[l]^{\!\!\mathfrak{f}^{-1}}~\K_G^{-1}(X,X_\infty)~&
\ar[l]~\K_{G}^{0}(X_\infty) \ . }
\end{displaymath}
It follows that the charge groups are given by%%
\begin{displaymath}
\ker\big(\mathfrak{f}^0\big)~\cong~\frac{\K_{G}^{-1}(X_\infty)}
{i^*\big(\K_{G}^{-1}(X)\big)} \qquad \mbox{and} \qquad
\ker\big(\mathfrak{f}^{-1}\big)~\cong~\frac{\K_{G}^0(X_\infty)}
{i^*\big(\K_{G}^0(X)\big)} \ .
\end{displaymath}
This formula means that the group of Type~IIB (resp.~Type~IIA) brane
charges is measured by the group $\K_G^{-1}(X_\infty)$
(resp.~$\K_G^0(X_\infty)$) of ``orbifold Ramond-Ramond fluxes at
infinity'' which cannot be extended to all of spacetime $X$. We may
then interpret, for arbitrary spacetimes $X$, the
group $\K_{G}^{-1}(X)$ (resp.~$\K_{G}^{0}(X)$) as the group classifying
Ramond-Ramond fields in the orbifold $X/G$ which are not sourced by
branes in Type~IIB (resp.~Type~IIA) string theory.

The Ramond-Ramond current can be described explicitly in the
delocalized theory of Section~\ref{RRCouplings}. The Wess-Zumino
pairing (\ref{coupling}) between a topologically trivial, complex
Ramond-Ramond potential and a D-brane represented by an equivariant
K-cycle $(W,E,f)\in\cat{D}^G(X)$ contributes a source term to the
Ramond-Ramond equations of motion, which is the class
$$\big[Q(W,E,f)\big]~\in~\H_G^{\rm
  even}\big(X\,;\,\complex\otimes\underline{R}(-)\big)$$ represented
by the pushforward
$$
Q(W,E,f)=f_!^{\H_G}\big(\ch^\complex(E)\wedge_G\,\mathcal{R}(W,f)\big)
\ .
$$
We now use the Riemann-Roch formula (\ref{GRRthm}) and the fact that
$f^*$ is right adjoint to $f_!^{\H_G}$, \emph{i.e.}, $f_!^{\H_G}\circ
f^*=\Id_{\H_G^\bullet(X;\,\underline{\complex}(-))}$. Using the explicit
expression for the curvature form in eq.~(\ref{calRWf}), we can
then rewrite this class as
\beq
Q(W,E,f)=\ch^\complex\big(f_!^{\K_G}(E)\big)
\wedge_G\,\sqrt{\Todd_G(T_X)\wedge_G\,\Euler_G(X)} \ .
\label{QWEfclass}\eeq
This is the complex Bredon cohomology class of the Ramond-Ramond
current $J$ created by the D-brane $(W,E,f)$. In the case $G=e$, the
expression (\ref{QWEfclass}) reduces to the standard class of the
current for Ramond-Ramond fields in Type~II superstring theory on
$X$~\cite{Cheung1997,Minasian1997,Moore2000,Olsen2000}.

There is a natural extension of the current (\ref{QWEfclass}) which
allows us to formally conclude, in analogy with the non-equivariant
case, that the complex Bredon cohomology class associated to a class
$\xi\in{\K^\bullet_{G}(X)\otimes\complex}$ representing a
Ramond-Ramond field is assigned by the equivariant Chern character. If
the Ramond-Ramond field is determined by a differential form
$C/2\pi\,\sqrt{\Todd_G(T_X)\wedge_G\,\Euler_G(X)}$ with
$C\in\Omega^\bullet_G(X;\complex)$ and $\dd_GC=\omega$, then this is
the class $[\omega(\xi)]$ in
$\H_G^\bullet(X;\complex\otimes\underline{R}(-))$ represented by the
closed differential form
\beq\label{FCxich}
\frac{\omega(\xi)}{2\pi\,\sqrt{\Todd_G(T_X)\wedge_G\,\Euler_G(X)}}=
\ch^\complex(\xi) \ .
\eeq
This is just the anticipated flux quantization condition from orbifold
differential K-theory. The appearance of the additional
gravitational terms in eq.~(\ref{FCxich}) is inconsequential to this
identification. Given the canonical map (\ref{canmapE}) in a
generalized cohomology theory $\E^\bullet$, any other map
$\E^{\bullet}(X)\rightarrow
{\H(X;\mathbb{R}\otimes{\pi_{-\bullet}\E})}^{\bullet}$ with the same
properties described in Section~\ref{gendiff} is obtained by
multiplying $\varphi$ with an invertible element in
$\H(X;\mathbb{R}\otimes{\pi_{-\bullet}\E})^0$. In the case at hand,
the characteristic class $\sqrt{\Todd_G(T_X)\wedge_G\,\Euler_G(X)}$ is
an invertible closed differential form which represents this element
in $\H_G^{\rm even}(X;\complex\otimes\underline{R}(-))$. This class
reduces to the usual gravitational correction $\sqrt{\Todd(T_X)}$ when
$G$ acts trivially on $X$.

We should stress that this analysis of the delocalized theory assumes
the strong conditions spelled out in Section~\ref{Gravcoupl}, which
require a deep geometrical compatibility of the equivariant K-cycle
$(W,E,f)$ with the orbifold structure of $[X/G]$ (or else an explicit
determination of the unknown characteristic class $\Lambda_G(W)$
correcting the Riemann-Roch formula as explained in
Section~\ref{EqRRformula}). The example of the linear orbifolds
considered in  Sections~\ref{Linorb1} and~\ref{Linorb2}, and in
Section~\ref{Linorb3} below, is simple enough to satisfy these
conditions. It would be very interesting to find a geometrically
non-trivial explicit example to test these requirements on. In any
case, the results above suggest that the orbifold differential
K-theory (or more precisely a complex version of it) defined in the
previous section is the natural framework in which to describe
topologically non-trivial Ramond-Ramond fields on orbifolds. It would
be highly desirable to determine the correct generalization of
eq.~(\ref{QWEfclass}) to the orbifold differential K-theory group
$\check{\K}_G^\bullet(X)$ of the previous section, and thereby
extending the delocalized Ramond-Ramond currents to include effects
such as torsion.

\subsection{Linear orbifolds\label{Linorb3}}

To understand certain aspects of the orbifold differential K-theory
groups, it is instructive to study the K-theory
classification of Ramond-Ramond fields on the linear orbifolds
considered in Sections~\ref{Linorb1} and~\ref{Linorb2}. Since the
$\complex$-linear $G$-module $V$ is equivariantly contractible, one
has $\H_G^{\rm odd}(V;\real\otimes\underline{R}(-))=0$ and
$\K_G^0(V)=R(G)$. From Theorem~\ref{exseqthm} it then follows that
$$
\check{\K}_G^0(V)~\cong~A_{\K_G}^0(X)~
\cong~\big\{(\gamma,\omega)\in R(G)\times
\Omega_{G,{\rm cl}}^{\rm even}(V;\real)~
\big|~\ch_{G/G}(\gamma)=[\omega]_{G-{\rm dR}}\big\} \ .
$$
Since the equivariant Chern character $\ch_{G/H}:R(H)\to R(H)$ for
$H\leq G$ is the identity map, the setwise fibre product truncates to
the lattice of quantized orbifold differential forms and one has
\beq
\check{\K}_G^0(V)=\Omega_{\K_G}^{\rm even}(V;\real) \ .
\label{checkKG0V}\eeq
This is the group of Type~IIA Ramond-Ramond form potentials on $V$. It
naturally contains those fields which trivialize the Ramond-Ramond
currents sourced by the stable fractional D0-branes of the Type~IIA
theory, corresponding to characteristic classes $[c]$ in the
representation ring $R(G)$ as explained in Section~\ref{Linorb1}.

This can be explicitly described as an extension of the group of
topologically trivial Ramond-Ramond fields $C$ of odd degree by the
equivariant K-theory of $V$, as implied by
Corollary~\ref{charclasscorr}. Since $V$ is connected and
$G$-contractible, one has $\Omega_{G,{\rm cl}}^0(V;\real)=\real\otimes
R(G)$ and the group (\ref{checkKG0V}) has a natural splitting
\beq
\check{\K}_G^0(V)=R(G)\oplus\Big(\,\bigoplus_{k=1}^d\,
\Omega_{G,{\rm cl}}^{2k}(V;\real)\Big) \ .
\label{checkKG0Vsplit}\eeq
Any closed orbifold form $\omega$ on $V$ of positive degree
is exact, $\omega=\dd_{G}C$, with the gauge invariance $C\mapsto
C+\dd_{G}\xi$. It follows that there is a natural map
$$
\bigoplus_{k=1}^d\,\Omega_{G,{\rm cl}}^{2k}(V;\real)~\longrightarrow~
\frac{\Omega_G^{\rm odd}(V;\real)}{\Omega_{\K_G}^{\rm odd}(V;\real)}
$$
which associates to the field strength $\omega$ the corresponding
globally well-defined Ramond-Ramond potential $C$.

On the other hand, the orbifold differential K-theory group
$\check{\K}_G^{-1}(V)$ of Type~IIB Ramond-Ramond fields on $V$ can be
computed by using the characteristic class exact sequence
(\ref{charclasseq}) with degree shifted by $-1$. Using
$\K_G^{-1}(V)=0$, one finds
\beq
{\check{\K}_{G}^{-1}(V)}=\dfrac{\Omega_{G}^{\rm even}(V;
\mathbb{R})}{\Omega_{\K_{G}}^{{\rm even}}(V;\mathbb{R})} \ .
\label{checkKG1V}\eeq
This result reflects the fact that the Type~IIB theory has no stable
fractional D0-branes. Hence there is no extension and the
Ramond-Ramond fields are induced solely by the closed string
background. Their field strengths $\omega=\dd_{G}C$ are
determined entirely by the potentials $C$, which are globally defined
differential forms of even degree.

Note that for any $G$-homogeneous space $G/H$ one has
$\K_G^{-1}(G/H)=0$ and $$\Omega_G^{\rm
  odd}(G/H;\real)=0 \ . $$
>From the characteristic class exact sequence (\ref{charclasseq}) one
thus computes that the orbifold differential K-theory group
\beq
\check{\K}_G^0(G/H)~\cong~\K_G^0(G/H)~\cong~R(H)
\label{checkKG0GH}\eeq
is given by the characteristic classes (of fractional D0-branes),
while Theorem~\ref{exseqthm} (with degree shifted by $-1$) implies
that the orbifold differential K-theory group
\beq
\check{\K}_G^{-1}(G/H)~\cong~\frac{\H_G^{\rm even}\big(G/H\,;\,
\real\otimes\underline{R}(-)\big)}{\ch_{G/H}\big(\K_G^0(G/H)\big)}~
\cong~R(H)\otimes\real/\zed
\label{checkKG1GH}\eeq
is given by the topologically trivial flat fields. Setting $H=G$ in
eqs.~(\ref{checkKG0GH}) and~(\ref{checkKG1GH}) shows that the
differential $\K_G$-theory groups of a point generically differ from
the groups~(\ref{checkKG0V}) and~(\ref{checkKG1V}), even though $V$ is
$G$-contractible. This exemplifies the $G$-homotopy non-invariance of
the orbifold differential K-theory groups, required to capture the
non-vanishing (but topologically trivial) gauge potentials on $V$.

\subsection{Flat potentials}

In Section~\ref{Linorb3} above we encountered some examples of
topologically trivial Ramond-Ramond fields, corresponding to gauge
equivalence classes with trivial K-theory flux $[c]=0$. They are the
globally defined orbifold differential forms
$$C~\in~\Omega_G^\bullet(X;\real)$$
with the gauge symmetry $C\to C+\xi$, where $\dd_{G}\xi=0$ and
$\xi\in\Omega_{\K_G}^\bullet(X;\real)$, and
field strength $$\omega=\dd_{G}C \ . $$ The \emph{flat} Ramond-Ramond
fields are instead classified by the abelian group
$\K_G^i(X;\real/\zed)$, where $i=0$ for Type~IIB theory and $i=-1$ for
Type~IIA theory. In the previous section this group was defined to be
the subgroup of orbifold differential K-theory with vanishing
curvature. In the following we will conjecture a very natural
algebraic definition of these groups which ties them somewhat more
directly to equivariant K-theory groups.

To motivate this conjecture, we first compute the groups
$\K_G^\bullet(V;\real/\zed)$ for the linear orbifolds of
Section~\ref{Linorb3} above, wherein the associated differential
K-theory groups were determined explicitly. Using the field strength
exact sequence (\ref{diffKexseqGfs}), by definition one has
$$
\K_G^{-1}(V;\real/\zed)\cong\ker\big(f_{\rm fs}:\check{\K}_G^0(V)\to
\Omega_{\K_G}^{\rm even}(V;\real)\big)
$$
which from the natural isomorphism (\ref{checkKG0V}) trivially gives
\beq
\K_G^{-1}(V;\real/\zed)=0 \ .
\label{KGflat1V}\eeq
Similarly, using $\K_G^{-1}(V)=0$ one has
$$
\K_G^0(V;\real/\zed)\cong\ker\big(f_{\rm fs}:\check{\K}_G^{-1}(V)\to
\Omega_{G,{\rm cl}}^{\rm odd}(V;\real)\big) \ .
$$
Using the natural isomorphism (\ref{checkKG1V}), the field strength
map is $$f_{\rm fs}\big([C]\big)\=\dd_{G}C \qquad \mbox{for}
\quad C\in\Omega_{G}^{\rm
  even}(V;\real) \ , $$ giving
$$
\K_G^0(V;\real/\zed)\cong\frac{\Omega_{G,{\rm cl}}^{\rm even}(V;
\mathbb{R})}{\Omega_{\K_{G}}^{{\rm even}}(V;\mathbb{R})} \ .
$$
Similarly to eq.~(\ref{checkKG0Vsplit}), there is a natural splitting
of the vector space of closed orbifold differential forms given by
$$
\Omega_{G,{\rm cl}}^{\rm even}(V;\mathbb{R})=
\big(R(G)\otimes\real\big)~\oplus~\Big(\,\bigoplus_{k=1}^d\,
\Omega_{G,{\rm cl}}^{2k}(V;\real)\Big)
$$
and we arrive finally at
\beq
\K_G^0(V;\real/\zed)=R(G)\otimes\real/\zed \ .
\label{KGflat0V}\eeq

These results of course simply follow from the fact that $V$ is
$G$-contractible, so that every $\dd_{G}$-closed Ramond-Ramond
field is trivial, except in degree zero where the gauge equivalence
classes are naturally parametrized by the twisted sectors of the
string theory in eq.~(\ref{KGflat0V}). Note that both groups of flat
fields (\ref{KGflat1V}) and (\ref{KGflat0V}) are unchanged by
(equivariant) contraction of the $G$-module $V$ to a point, as an
analogous (but simpler) calculation shows. This suggests that the
groups $\K_G^\bullet(X;\real/\zed)$ have at least some $G$-homotopy
invariance properties, unlike the differential $\K_G$-theory
groups. This motivates the following conjectural algebraic framework
for describing these groups.

We will propose that the group $\K_G^\bullet(X;\real/\zed)$ is an
\emph{extension} of the torus of topologically trivial flat orbifold
Ramond-Ramond fields by the torsion elements in $\K_G^{\bullet+1}(X)$,
as they have vanishing image under the equivariant Chern character
$\ch_X$. The resulting group may be called the ``equivariant K-theory
with coefficients in $\real/\zed$''. The short exact sequence of
coefficient groups
$$
0~\longrightarrow~\zed~\longrightarrow~\real~\longrightarrow~
\real/\zed~\longrightarrow~0
$$
induces a long exact sequence of equivariant K-theory groups which, by
Bott periodicity, truncates to the six-term exact sequence
\beq\label{Bocksteinseq}
\xymatrix{\K_{G}^{0}(X)~\ar[r]&~{\K_{G}^{0}(X;\real)}~
\ar[r]& ~\K_{G}^{0}(X;\real/\zed)\ar[d]^{\beta}\\
\K_{G}^{-1}(X;\real/\zed)~\ar[u]^{\beta}&\ar[l]~\K_G^{-1}(X;\real)~&
\ar[l]~\K_{G}^{-1}(X) \ . }
\eeq
The connecting homomorphism $\beta$ is a suitable variant of the usual
Bockstein homomorphism. We assume that the equivariant K-theory with
real coefficients is defined simply by the $\zed_2$-graded ring
$$\K_G^\bullet(X;\real)\=\K_G^\bullet(X)\otimes\real~\cong~
\H_G^\bullet\big(X\,;\,\real\otimes\underline{R}(-)\big) \ , $$ where
we have used Theorem~\ref{eqChernthm}. The maps to real K-theory in
eq.~(\ref{Bocksteinseq}) may then be identified with the equivariant
Chern character $\ch_X$, whose image is a full lattice in the Bredon
cohomology group $\H_G^\bullet(X;\real\otimes\underline{R}(-))$. Then
the abelian group $\K_G^\bullet(X;\mathbb{R}/\mathbb{Z})$ sits in the
exact sequence
\beq
0~\longrightarrow~\K_G^\bullet(X)\otimes\real/\zed~\longrightarrow~
\K_G^\bullet(X;\real/\zed)~\xrightarrow{\beta}~\Tor
\big(\K_G^{\bullet+1}(X)\big)~\longrightarrow~0 \ .
\label{RZcoeffsdef}\eeq

When $G=e$, eq.~(\ref{Bocksteinseq}) is the usual Bockstein exact
sequence for K-theory. In this case, an explicit geometric
realization of the groups $\K^\bullet(X;\real/\zed)$ in terms of
bundles with connection has been given by Lott~\cite{Lott1994}. 
Moreover, in ref.~\cite{Hopkins2005} a geometric construction of
the map ${\K^{-1}(X;\mathbb{R}/\mathbb{Z})}\to{\check{\K}^{0}(X)}$ in
the field strength exact sequence is given. Unfortunately, no such
geometrical description is immediately available for our equivariant
differential K-theory, due to the lack of a Chern-Weil theory for the
homotopy theoretic equivariant Chern character of
Section~\ref{section2}. Our conjectural definition (\ref{RZcoeffsdef})
is satisfied by the linear orbifold groups (\ref{KGflat1V})
and~(\ref{KGflat0V}).

In ref.~\cite{deBoer2001} a very different definition of the groups
$\K_G^\bullet(X;\real/\zed)$ is given, by defining both equivariant
K-theory and cohomology using the Borel construction of
Example~\ref{Borelex}. Then the Bockstein exact sequence
(\ref{Bocksteinseq}) is written for the ordinary K-theory groups of
the homotopy quotient $X_G=EG\times_GX$. While these groups reduce,
like ours, to the usual K-theory groups of flat fields when $G=e$,
they do not obey the exact sequence (\ref{RZcoeffsdef}). The reason is
that the equivariant Chern character used is \emph{not} an isomorphism
over the reals, as explained in Section~\ref{CherntopK} (see also
ref.~\cite{Luck2001} for a description of $\K^\bullet(X_G)$ as the
completion of $\K^\bullet_G(X)$ with respect to a certain
ideal). Moreover, an associated differential K-theory construction
would directly involve differential forms on the infinite-dimensional
space $X_G$ which is only homotopic to the finite-dimensional
CW-complex $X/G$. The physical interpretation of such fields is not
clear. Even in the simple case of the linear orbifolds $V$ studied
above, this description predicts an infinite set of equivariant fluxes
of arbitrarily high dimension on the infinite-dimensional classifying
space $BG$, and one must perform some non-canonical quotients in order
to try to isolate the physical fluxes. The differences between the
equivariant K-theory and Borel cohomology groups of $V$ also require
postulating certain effects of fractional branes on the
orbifold, as in ref.~\cite{Bergman2001}. In contrast, with our
constructions the relation between orbifold flux groups and Bredon
cohomology is much more natural, and it involves only finitely-many
orbifold Ramond-Ramond fields.

\subsection{Consistency conditions\label{FluxBredon}}

As we have stressed throughout this paper, the usage of Borel
cohomology as a companion to equivariant K-theory in the topological
classification of D-branes and Ramond-Ramond fluxes on orbifolds has
various undesirable features, most notably the fact that it involves
\emph{torsion} classes substantially, especially when finite group
cohomology is involved. In our applications to string geometry, it is
more convenient to use an equivariant cohomology theory with
substantial \emph{torsion-free} information. Bredon cohomology
naturally accomplishes this, as instead of group cohomology the basic
object is a representation ring. In fact, as we now demonstrate, the
formulation of topological consistency conditions for orbifold
Ramond-Ramond fields and D-branes within the framework of equivariant
K-theory naturally necessitates the use of classes in Bredon
cohomology.

Given a Bredon cohomology class
$\lambda\in\H_G^\bullet(X;\underline{R}(-))$, let us ask if there
exists a Ramond-Ramond field for which $\omega=\lambda$ in the sense
of eq.~(\ref{FCxich}). For this, we must find an equivariant K-theory
lift $\xi\in\K_G^\bullet(X)$ of $\lambda$. As in the non-equivariant
case~\cite{Diaconescu2003,Maldacena2001}, the obstructions to such a
lift can be determined via a suitable spectral sequence. For
equivariant K-theory the appropriate spectral sequence is described in
refs.~\cite{Davis1998,Mislin2003} (see also ref.~\cite{Segal1968})
using the skeletal filtration $(X_n)$ of Section~\ref{Gcomplex}. We
will now briefly explain the construction of this spectral sequence
and its natural relationship with the obstruction theory for
Ramond-Ramond fluxes in Bredon cohomology.

The $\E_1$-term of the spectral sequence is the relative
$G$-equivariant K-theory group
$$
\E_1^{p,q}=\K_G^{p+q}(X_p,X_{p-1})
$$
with differential
$$
\dd_1^{p,q}\,:\,\E_1^{p,q}~\longrightarrow~\E_1^{p+1,q}
$$
induced by the long exact sequence of the triple
$(X_{p+1},X_p,X_{p-1})$ in equivariant K-theory, \emph{i.e.},
$\dd_1^{p,q}$ is the composition of the map $i^*$ induced by the
inclusion $i:(X_p,\emptyset)\hookrightarrow(X_p,X_{p-1})$ with the
cellular coboundary operator of the pair $(X_{p+1},X_p)$. From
eq.~(\ref{Xnattach}) it follows that there is a homeomorphism
$\coprod_{j\in J_p}\,\big(\mathring{\mathbb{B}}{}_j^p\times
G/K_j\big)\to X_p\setminus X_{p-1}$, and hence
$$
\E_1^{p,q}~\cong~\bigoplus_{j\in J_p}\,\K^{p+q}_G\big(
\mathring{\mathbb{B}}{}_j^p\times G/K_j\big)~\cong~
\bigoplus_{j\in J_p}\,\K_G^q\big(G/K_j\big) \ .
$$
Thus $\E_1^{p,q}=0$ for $q$ odd, while for $q$ even the group
$\E_1^{p,q}$ is a direct sum of representation rings $R(K_j)$ over all
isotropy subgroups of $p$-cells of orbit type $G/K_j$. It parametrizes
equivariant K-theory classes defined on the $p$-skeleton of $X$ which
are trivial on the $(p-1)$-skeleton, and gives the supports of
$p$-form fields and charges on the orbifold which carry no lower
or higher degree fluxes.

The $\E_2$-term of the spectral sequence is the cohomology of the
differential $\dd_1$. The cohomology of the cochain complex assembled
from such terms is the equivariant cohomology with coefficient system
$\underline{R}(-)$ on $\ocat{G,\mathfrak{F}(X_p)}$ for $q=0$, and thus
a necessary condition for a $p$-form Ramond-Ramond field to lift to
$\K_G^\bullet(X)$ is that it define a non-trivial cocycle in Bredon
cohomology. This is consistent with Definition~\ref{orbdiffKdef}. The
resulting Atiyah-Hirzebruch spectral sequence may then be written
$$
\E_2^{p,q}=\H_G^p\big(X\,;\,\underline{\pi_{-q}\K}\,_G(-)\big) \quad
\Longrightarrow \quad \K_G^{p+q}(X)
$$
and it lives in the first and fourth quadrants of the
$(p,q)$-plane. On the $r$-th terms $\E_r^{p,q}$, the differential
$\dd^{p,q}_r$ has bidegree $(r,-r+1)$, and $\E_{r+1}^{p,q}$ is the
corresponding cohomology group. Note that $\dd^{p,q}_r=0$ for
all $r$ even, since then either its source or its target vanishes (as
$\K^q(\complex[H])=0$ for all $q$ odd and $H\leq G$). The
$\E_\infty$-term is the inductive limit
$$
\E_\infty^{p,q}=\lim_{\stackrel{\scriptstyle\longrightarrow}{\scriptstyle
  r}}\,\E_r^{p,q} \ .
$$

For a finite-dimensional manifold $X$, one has
$\E_r^{p,q}=\E_\infty^{p,q}$ for all $r>\dim(X)$ and the spectral
sequence converges to $\K_G^{p+q}(X)$. This means that the
$\E_\infty$-term is the associated graded group of a decreasing finite
filtration
$\F_{p,q}\,\K_G^{p+q}(X)\subset\F_{p-1,q+1}\,\K_G^{p+q}(X)$, $0\leq
p\leq\dim(X)$ with $\K_G^q(X)=\F_{0,q}\,\K_G^q(X)$ and
\beq
\frac{\F_{p,q}\,\K_G^{p+q}(X)}{\F_{p+1,q-1}\,\K_G^{p+q}(X)}~\cong~
\E_\infty^{p,q} \ .
\label{Epqfilt}\eeq
Explicitly, if $\iota:X_{p-1}\hookrightarrow X$ denotes the inclusion
of the $(p-1)$-skeleton in $X$, then the filtration groups
$$
\F_{p,q}\,\K_G^{p+q}(X):=\ker\big(\iota^*:\K_G^{p+q}(X)\to
\K^{p+q}_G(X_{p-1})\big)
$$
consist of Ramond-Ramond fluxes where the field strength $\omega$ is a
form of degree $\geq p$, while the extension groups (\ref{Epqfilt})
consist of $p$-form fluxes with vanishing higher and lower degree
fluxes. By Theorem~\ref{eqChernthm}, the equivariant Chern character
$\ch_X$ determines an isomorphism from the limit of the
spectral sequence to its $\E_2$-term. Thus the spectral sequence
collapses rationally, and so the images of all higher differentials
$\dd^{p,q}_r$, $r>2$ in the spectral sequence consist of torsion
classes.

It follows that the next non-trivial obstruction to extending a
Ramond-Ramond field is given by a ``cohomology operation''
\beq
\dd_3^{p,0}\,:\,\H_G^p\big(X\,;\,\underline{R}(-)\big)~\longrightarrow~
\H_G^{p+3}\big(X\,;\,\underline{R}(-)\big) \ .
\label{d3p0def}\eeq
Thus a necessary condition for a Bredon cohomology class
$\lambda\in\H_G^p(X;\underline{R}(-))$ to survive to $\E_\infty^{p,0}$
is given by
\beq
\dd_3^{p,0}(\lambda)=0 \ .
\label{d3plambda0}\eeq
We interpret the condition (\ref{d3plambda0}) as a (partial)
requirement of global worldsheet anomaly cancellation for
Ramond-Ramond fluxes and, dually, the worldvolume homology cycles that
they pair with. This is the orbifold generalization of the
Freed-Witten condition~\cite{Diaconescu2003,Freed1999,Maldacena2001}
formulated in terms of obstruction classes in Bredon cohomology. It
is a necessary condition for the existence of a fractional D-brane
whose lowest non-vanishing Ramond-Ramond charge is
$\lambda\in\H_G^p(X;\underline{R}(-))$. On the other hand, in
computing the $\E_3$-term as the cohomology of the differential
(\ref{d3p0def}), we must also take the quotient by the image of
$\dd_3^{p-3,0}$. This means that a class $\lambda$ satisfying
eq.~(\ref{d3plambda0}) must be further subjected to the
identifications
\beq
\lambda\sim\lambda+\dd_3^{p-3,0}(\lambda'\,)
\label{lambdad3id}\eeq
in $\E_3^{p,0}$, for any class
$\lambda'\in\H_G^{p-3}(X;\underline{R}(-))$. We interpret the
condition (\ref{lambdad3id}) as accounting for Ramond-Ramond charge
violation due to D-instanton effects in the orbifold background, as
explained in ref.~\cite{Maldacena2001} for the non-equivariant
case. It means that while there exists a fractional brane whose lowest
Ramond-Ramond charge is $\dd_3^{p-3,0}(\lambda'\,)$, this D-brane is
unstable.

The passage from the limit (\ref{Epqfilt}) with $q=0$ to the actual
equivariant K-theory group $\K_G^p(X)$ requires solving a typically
non-trivial extension problem. Even when the spectral sequence
collapses at the $\E_2$-term, the extension can lead to important
torsion corrections which distinguish the classifications of
Ramond-Ramond fields based on Bredon cohomology and on equivariant
K-theory. The extension problem changes the additive structure on the
K-theory group of fluxes from that of the equivariant cohomology
classes. This corresponds physically to non-trivial correlations
between Ramond-Ramond fields of different degrees, when represented by
orbifold differential forms. This torsion enhancement in equivariant
K-theory compared to Bredon cohomology can shift the Dirac charge
quantization condition on the Ramond-Ramond fields by fractional units
and can play an important role near the orbifold
points~\cite{Bergman2001,Bergman2001rp}.

In the non-equivariant case $G=e$, the differential $\dd_3^{p,0}$ is
known to be given by the Steenrod square cohomology operation ${\rm
  Sq}^3$. The vanishing condition ${\rm Sq}^3(\lambda)=0$ implies the
vanishing of the third integer Stieffel-Whitney class of the
Poincar\'e dual cycle to $\lambda\in\H^p(X;\zed)$, which is just the
condition guaranteeing that the corresponding brane worldvolume is a
spin$^c$ submanifold of $X$. Unfortunately, for $G\neq e$ the
differential $\dd_3^{p,0}$ is not known and the geometrical meaning of
the condition (\ref{d3plambda0}) is unclear. It would be interesting to
understand this requirement in terms of an obstruction theory for
Bredon cohomology, analogously to the non-equivariant case, as this
would open up interesting new consistency conditions for D-branes and
Ramond-Ramond fields on global orbifolds $[X/G]$. However, we are not
aware of any characteristic class theory underlying the Bredon
cohomology groups $\H_G^p(X;\underline{R}(-))$.

\setcounter{section}{0}

\setcounter{subsection}{0}

\appendix{Linear algebra in functor categories}

In this appendix we will summarize some notions about algebra in
functor categories that were used in the main text of the paper. They
generalize the more commonly used concepts for modules over a
ring. For further details see ref.~\cite{Dieck1987}.

Let $R$ be a commutative ring, and denote the category of (left)
$R$-modules by $\lmod{R}$. Let $\Gamma$ be a \emph{small} category,
\emph{i.e.}, its class of objects ${\rm Obj}(\Gamma)$ is a set. If
$\mathcal{C}$ is another category, then one denotes by%%
\begin{displaymath}
\left[\Gamma,\mathcal{C}\right]
\end{displaymath}
the \emph{functor category} of (covariant) functors
$\Gamma\to{\mathcal{C}}$. The objects of
$\left[\Gamma,\mathcal{C}\right]$ are (covariant) functors
$\phi:\Gamma\to\mathcal{C}$ and a morphism from $\phi_{1}$ to
$\phi_{2}$ is a natural transformation $\alpha:\phi_{1}\to\phi_{2}$
between functors.

In particular, in the main text we used the functor category%%
\begin{displaymath}
\lmod{R\Gamma}:=\left[\Gamma,\lmod{R}\right]
\end{displaymath}
whose objects are called \emph{left $R\Gamma$-modules}. If one denotes
with $\Gamma^{\text{op}}$ the dual category to $\Gamma$, then there is
also the functor category%%
\begin{displaymath}
\rmod{R\Gamma}:=\left[\Gamma^{\text{op}},\lmod{R}\right]
\end{displaymath}
of contravariant functors $\Gamma\to\lmod{R}$, whose objects are
called \emph{right $R\Gamma$-modules}. As an example, let $G$ be a
discrete group regarded as a category with a single object and a
morphism for each element of $G$. A covariant functor $G\to\lmod{R}$
is then the same thing as a left module over the group ring $R[G]$ of
$G$ over $R$.

As the name itself suggests, all standard definitions from the linear
algebra of modules have extensions to this more general setting. For
instance, the notions of \emph{submodule, kernel, cokernel, direct
  sum, coproduct, etc.} can be naturally defined objectwise. If $M$
and $N$ are $R\Gamma$-modules, then $\text{Hom}_{R\Gamma}(M,N)$ is the
$R$-module of all natural transformations $M\to{N}$. This notation
should not be confused with the one used for the set of all morphisms
between two objects in $\Gamma$, and usually it is clear from the
context.

If $M$ is a right ${R\Gamma}$-module and $N$ is a left
${R\Gamma}$-module, then one can define their categorical \emph{tensor
  product}%%
\begin{displaymath}
M\otimes_{R\Gamma}N
\end{displaymath}
in the following way. It is the $R$-module given by first forming the
direct sum%%
\begin{displaymath}
F=\bigoplus_{\lambda\in\text{Obj}(\Gamma)}\,
M(\lambda)\otimes_{R}N(\lambda)
\end{displaymath}
and then quotienting $F$ by the $R$-submodule generated by all
relations of the form%%
\begin{displaymath}
f^*(m)\otimes{n}-m\otimes{f_*(n)}=0 \ ,
\end{displaymath}
where $(f:\lambda\to{\rho})\in\text{Mor}(\Gamma)$,
$m\in{M(\rho)},\:n\in{N(\lambda)}$ and $f^*(m)=M(f)(m),\:f_*
(n)=N(f)(n)$. This tensor product commutes with coproducts.
If $M$ and $N$ are functors from $\Gamma$ to the category of vector
spaces over a field $\mathbb{K}$, then their tensor product is
naturally equiped with the structure of a vector space over
$\mathbb{K}$. When $\Gamma$ is the orbit category $\ocat{G}$ and
$R=\zed$, the tensor product has precise limiting cases. For an
arbitrary contravariant module $M$ and the constant covariant module
$N$, the categorical product $M\otimes_{\zed\ocat{G}}N$ is the tensor
product of the right $\zed[G]$-module $M(G/e)$ with the constant left
$\zed[G]$-module $N(G/e)$, $M(G/e)\otimes_{\zed[G]}N(G/e)$. On the
other hand, if the contravariant module $M$ is constant and the
covariant module $N$ is arbitrary, then $M\otimes_{\zed\ocat{G}}N$ is
just $N(G/G)$.

\setcounter{subsection}{0}

\appendix{Equivariant K-homology\label{App:EqKhom}}

This appendix is devoted to explaining in more detail some of the
definitions and technical constructions in equivariant K-homology
theories that were used in the main text to describe states of
D-branes in orbifolds.

\subsection{Spectral definition\label{EqKhom}}

A natural way to define the equivariant homology theory $\K_\bullet^G$
is by means of a \emph{spectrum} for equivariant topological K-theory
$\K^\bullet_G$, which within the context of Section~\ref{EqChern} is a
particular covariant functor $\underline{\Vect}^G(-)$ from the orbit
category $\ocat{G}$ to the tensor category $\cat{Spec}$ of
spectra~\cite{Davis1998}. Given any $G$-complex $X$, the corresponding
pointed $G$-space is $X_+=X\amalg\pt$ and one defines the loop
spectrum $X_+\otimes_G\underline{\Vect}^G(-)$ by
\beq
X_+\otimes_G\underline{\Vect}^G(-)=\coprod_{G/H\in\ocat{G}}\,
\big(X_+^H\wedge\underline{\Vect}^G(G/H)\big)\,\big/\,\sim \ ,
\label{loopspec}\eeq
where the equivalence relation $\sim$ is generated by the
identifications $f^*(x)\wedge s\sim x\wedge f_*(s)$ with $(f:G/K\to
G/H)\in{\rm Mor}(\ocat{G})$, $x\in X_+^H$,
and $s\in\underline{\Vect}^G(G/K)_\bullet$. One then puts
\beq
\K_\bullet^G(X):=\pi_\bullet\big(X_+\otimes_G\underline{\Vect}^G(-)
\big) \ .
\label{KiGXspec}\eeq

By using various $G$-homotopy equivalences of the loop spectra
(\ref{loopspec}), one shows that this definition of equivariant
K-homology comes with a natural induction structure in the sense of
Section~\ref{EqCohTh}. For the trivial group it reduces to the
ordinary K-homology $\K_\bullet^e=\K_\bullet$ given by the Bott
spectrum ${BU}$. If $G$ is a finite group, any finite-dimensional
representation of $G$ naturally extends to a complex representation of
the group ring $\complex[G]$. Then there is an analytic assembly map
$$
\ass\,:\,\K_\bullet^G(X)~\longrightarrow~\K_\bullet\big(\complex[G]\big)
$$
to the K-theory of the ring $\complex[G]$, induced by the collapsing
map $X\to\pt$ and the isomorphisms
$$\K_\bullet\big(\complex[H]\big)~\cong~
\pi_\bullet\big(\,\underline{\Vect}^G(G/H)\big)~\cong~
\K_\bullet^G(G/H)~\cong~ R(H)$$ for any subgroup
$H\leq G$. In the following we will give two concrete realizations of
the homotopy groups~(\ref{KiGXspec}).

\subsection{Analytic definition\label{Andef}}

The simplest realization of the equivariant K-homology
group $\K_\bullet^G(X)$ is within the framework of an equivariant
version of Kasparov's KK-theory $\KK^G_\bullet$. Let $\alg$ be a
$G$-algebra, \emph{i.e.}, a $C^*$-algebra $\alg$ together
with a group homomorphism $$\lambda\,:\,G~\longrightarrow~{\rm
  Aut}(\alg) \ . $$ By a Hilbert
$(G,\alg)$-module we mean a Hilbert $\alg$-module $\bun$ together with
a $G$-action given by a homomorphism $\Lambda:G\to{\rm GL}(\bun)$ such
that
\beq
\Lambda_g(\varepsilon\cdot a)=\Lambda_g\big(\varepsilon\cdot
\lambda_g(a)\big)
\label{covrep}\eeq
for all $g\in G$, $\varepsilon\in\bun$ and $a\in\alg$. Let
$\lin(\bun)$ denote the $*$-algebra of $\alg$-linear maps
$T:\bun\to\bun$ admitting an adjoint with respect to the $\alg$-valued
inner product on $\bun$. The induced $G$-action on $\lin(\bun)$
is given by $g\cdot T:=\Lambda_g\circ T\circ\Lambda_{g^{-1}}$. Let
$\comp(\bun)$ be the subalgebra of $\lin(\bun)$ consisting of
generalized compact operators.

Given a pair $(\alg,\balg)$ of $G$-algebras, let
$\cat{D}^G(\alg,\balg)$ be the set of triples $(\bun,\phi,T)$ where
$\bun$ is a countably generated Hilbert $(G,\balg)$-module,
$\phi:\alg\to\lin(\bun)$ is a $*$-homomorphism which commutes with the
$G$-action,
\beq
\phi\big(\lambda_g(a)\big)=\Lambda_g\circ\phi(a)\circ\Lambda_{g^{-1}}
\label{covrep2}\eeq
for all $g\in G$ and $a\in\alg$, and $T\in\lin(\bun)$ such that
\begin{itemize}
\item[1)] $[T,\phi(a)]\in\comp(\bun)$ for all $a\in\alg$; and
\item[2)] $\phi(a)\,(T-T^*)$, $\phi(a)\,(T^2-1)$, $\phi(a)\,(g\cdot
  T-T)\in\comp(\bun)$ for all $a\in\alg$ and $g\in G$.
\end{itemize}
The standard equivalence relations of KK-theory are now analogously
defined. The set of equivalence classes in $\cat{D}^G(\alg,\balg)$
defines the equivariant KK-theory groups $\KK_\bullet^G(\alg,\balg)$.

If $X$ is a smooth proper $G$-manifold without boundary, and $G$ acts
on $X$ by diffeomorphisms, then the algebra
$\alg=\C_0(X)$ of continuous functions on $X$ vanishing at infinity is
a $G$-algebra with automorphism $\lambda_g$ on $\alg$ given by
$$\lambda_g(f)(x)~:=~\big(g^*f\big)(x)\=f\big(g^{-1}\cdot x\big) \ , $$
where $g^*$ denotes the pullback of the $G$-action on $X$ by left
translation by $g^{-1}\in G$. We define
\beq\label{KKGandef}
\K_\bullet^G(X):=\KK_\bullet^G\big(\C_0(X)\,,\,\complex\big)
\eeq
with $G$ acting trivially on $\complex$. The conditions (\ref{covrep})
and (\ref{covrep2}) naturally capture the physical requirements that
physical orbifold string states are $G$-invariant and also that the
worldvolume fields on a fractional D-brane carry a ``covariant
representation'' of the orbifold group~\cite{douglas1996}.

\subsection{The equivariant Dirac class\label{Diracclass}}

We can determine a canonical class in the abelian group
(\ref{KKGandef}) as follows. Let $\dim(X)=2n$, and let $G$ be a finite
subgroup of the rotation group $\SO(2n)$.\footnote{Throughout the
  extension to $\K^G_1$ or $\K^{-1}_G$ and $\dim(X)$ odd can be
  described in the same way as in degree zero by replacing $X$ with
  $X\times\S^1$.} Let $$\Cliff(2n)=\Cliff^+(2n)\oplus\Cliff^-(2n)$$
denote the complex $\zed_2$-graded euclidean Clifford algebra
on $n$ generators $e_1,\dots,e_n$ with the relations
$$
e_i\,e_j+e_j\,e_i=-2\,\delta_{ij} \ .
$$
A choice of a complete $G$-invariant riemannian
metric on $X$ defines a $G$-bundle of Clifford algebras
$$\Cliff\=\Cliff\big(T^*_X\big)~:=~\Fr^*\times_{\SO(2n)}\,\Cliff(2n)$$
which is an associated bundle to the metric coframe bundle over $X$,
the principal $\SO(2n)$-bundle $\Fr^*=\Fr(T^*_X)$ of oriented
orthonormal frames on the cotangent bundle
$T^*_X=\Fr^*\times_{\SO(2n)}\,\real^{2n}$. The action of $\SO(2n)$ on
the Clifford algebra is through the
spin group $\Spin(2n)\subset\Cliff(2n)$. The Lie group
$\Spin^c(2n)\subset\Cliff(2n)$ is a central extension of $\SO(2n)$ by
the circle group $\U(1)$,
\beq
1~\longrightarrow~\U(1)~\longrightarrow~\Spin^c(2n)~\longrightarrow~
\SO(2n)~\longrightarrow~1 \ ,
\label{Spinc2ndef}\eeq
where the quotient map in eq.~(\ref{Spinc2ndef}) is consistent with
the double covering of $\SO(2n)$ by $\Spin(2n)$ so that
$$\Spin^c(2n)=\Spin(2n)\times_{\zed_2}\U(1) \ . $$

The $G$-manifold $X$ is said to have a \emph{$G$-\spinc structure} or
to be \emph{$\K_G$-oriented} if there is an extension of the coframe
bundle to a principal $\Spin^c(2n)$-bundle $\Fr^*_L$ over $X$ which is
compatible with the $G$-action. The extension $\Fr^*_L$ may be
regarded as a principal circle bundle over $\Fr^*$,
$$
\xymatrix{ & & \U(1)\ar[ld]\ar[d] & \\ 
\hat G~\ar[r]\ar[d] & ~\Spin^c(2n)~\ar[r]\ar[d] & ~
\Fr^*_L~\ar[r]\ar[d]&~ X \ , \\
G~\ar[r] & ~\SO(2n)~\ar[r] & ~\Fr^*~\ar[ru] & }
$$
where the pullback square on the bottom left defines the required
covering of the orbifold group $G<\SO(2n)$ by a subgroup of the \spinc
group $\hat G<\Spin^c(2n)$. This lift is also necessary in order to
account for the spacetime fermions present in string theory. The
kernel of the homomorphism $\hat G\to G$ is identified with the circle
group $\U(1)<\Spin^c(2n)$ in the Clifford algebra $\Cliff(2n)$. We fix
a choice of lift and hence assume that $G$ is a discrete subgroup of
the \spinc group. $\zed_2$-graded Clifford modules are likewise
extended to representations of $\complex[G]\otimes\Cliff(2n)$, with
$\complex[G]$ the group ring of $G$, called $G$-Clifford modules. The
topological obstruction to the existence of a $G$-\spinc structure on
$X$ is the equivariant third integral Stiefel-Whitney class
$(W_3)_G(T_X^*)\in\H_G^3(X;\zed)$ of the cotangent bundle $T_X^*$ in
Borel cohomology.

The associated bundles of half-spinors on $X$ are defined as
\beq
\spinor^\pm\=S\big(T^*_X\big)^\pm~:=~
\Fr^*_L\times_{\Spin^c(2n)}\,\Delta^\pm \ ,
\label{halfspinbun}\eeq
where $\Delta^\pm$ are the irreducible half-spin representations of
$\SO(2n)$. Since $G$ lifts to $\hat G$ in the \spinc group, the
half-spin representations $\Delta^\pm$ restrict to representations of $G$
and the half-spinor bundles (\ref{halfspinbun}) are $G$-bundles. The
$G$-invariant Levi-Civita connection determines a connection one-form
on $\Fr^*$, and together with a choice of $G$-invariant connection
one-form on the principal $\U(1)$-bundle $\Fr^*_L\to\Fr^*$,
they determine a connection one-form on the principal
$\Spin^c(2n)$-bundle $\Fr^*_L\to X$ which is $G$-invariant. This
determines an invariant connection
$$
\nabla^{\spinor\otimes
  E}:=\nabla^\spinor\otimes1+1\otimes\nabla^E\,:\,
\C^\infty\big(X\,,\,\spinor^+\otimes E\big)~\longrightarrow~
\C^\infty\big(X\,,\,T^*_X\otimes\spinor^+\otimes E\big)
$$
where $\nabla^E$ is a $G$-invariant connection on a $G$-bundle $E\to
X$. The contraction given by Clifford multiplication defines a map
$$
\Cl\,:\,\C^\infty\big(X\,,\,T^*_X\otimes\spinor^+\otimes E\big)~
\longrightarrow~\C^\infty\big(X\,,\,\spinor^-\otimes E\big)
$$
which graded commutes with the $G$-action, and the $G$-invariant
\spinc Dirac operator on $X$ with coefficients in $E$ is defined as
the composition
\beq
\Dirac_E^{X}=\Cl\circ\nabla^{\spinor\otimes E} \ .
\label{DiracEX}\eeq

We will view the operator (\ref{DiracEX}) as an operator on
$\Ltwo$-spaces $$\Dirac_E^{X}\,:\,\Ltwo\big(X\,,\,\spinor^+\otimes
E\big)~\longrightarrow~\Ltwo\big(X\,,\,\spinor^-\otimes E\big) \ . $$
It induces a class
$\big[\Dirac_E^X\big]\in\K_0^G(X)$ as follows. The $G$-algebra
$\C_0(X)$ acts on the $\zed_2$-graded $G$-Hilbert space
$\bun:=\Ltwo(X,\spinor\otimes E)$ by multiplication. Define the
bounded $G$-invariant operator
$T:=\Dirac_E^X\,\big((\Dirac_E^X)^2+1\big)^{-1/2}\in\Fred_G$. Then
$\big[\Dirac_E^X\big]$ is represented by the $G$-equivariant Fredholm
module $(\bun,T)$.

\subsection{Geometric definition\label{Topdef}}

The natural geometric description of D-branes in an orbifold
space is provided by the topological version of the groups
$\K_\bullet^G(X)$ due to Baum, Connes and
Douglas~\cite{Baum1982,Baum2000}. This can be defined for an arbitrary
discrete, countable group $G$ on the category of proper, finite
$G$-complexes $X$ and proven to be isomorphic to analytic equivariant
K-homology~\cite{Baum2007}. Recall that the topological equivariant
K-theory $\K_G^\bullet(X)$ is defined by applying the Grothendieck
functor $\K^\bullet$ to the additive category $\Vect^\complex_G(X)$
whose objects are complex $G$-vector bundles over $X$, \emph{i.e.},
$\K_G^\bullet(X):=\K^\bullet\big(\Vect^\complex_G(X)\big)$.
In the homological setting, the relevant category
is instead the additive category of \emph{$G$-equivariant K-cycles}
$\cat{D}^G(X)$, whose objects are triples $(W,E,f)$ where
\begin{itemize}
\item[(a)] $W$ is a manifold without boundary with a smooth proper
  cocompact $G$-action and $G$-\spinc structure;
\item[(b)] $E$ is an object in $\Vect^\complex_G(W)$; and
\item[(c)] $ f:W\to X$ is a $G$-map.
\end{itemize}
Two $G$-equivariant K-cycles $(W,E, f)$ and $(W',E', f'\,)$ are
said to be \emph{isomorphic} if there is a $G$-equivariant
diffeomorphism $h:W\to W'$ preserving the $G$-\spinc structures on
$W,W'$ such that $h^*(E'\,)\cong E$ and $ f'\circ h= f$.

Define an equivalence relation $\sim$ on the category $\cat{D}^G(X)$
generated by the operations of
\begin{itemize}
\item Bordism: $(W_i,E_i, f_i)\in\cat{D}^G(X)$,
  $i=0,1$ are \emph{bordant} if there is a triple $(M,E, f)$ where
  $M$ is a manifold with boundary $\partial M$, with a smooth proper
  cocompact $G$-action and $G$-\spinc structure, $E\to M$ is a complex
  $G$-vector bundle, and $ f:M\to X$ is a $G$-map such that
  $(\partial M,E|_{\partial M}, f|_{\partial
    M})\cong(W_0,E_0, f_0)\amalg(-W_1,E_1, f_1)$. Here $-W_1$
  denotes $W_1$ with the reversed $G$-\spinc structure;
\item Direct sum: If $(W,E, f)\in\cat{D}^G(X)$ and
  $E=E_0\oplus E_1$, then
  $$(W,E, f)\cong(W,E_0, f)\amalg(W,E_1, f) \ ; $$ and
\item Vector bundle modification: Let
  $(W,E, f)\in\cat{D}^G(X)$ and $H$ an even-dimensional $G$-\spinc
  vector bundle over $W$. Let $\widehat{W}=\S(H\oplus\id)$ denote the
  sphere bundle of $H\oplus\id$, which is canonically a $G$-\spinc
  manifold, with $G$-bundle projection $\pi:\widehat{W}\to W$. Let
  $$\spinor(H)=\spinor(H)^+\oplus\spinor(H)^-$$ denote the
  $\zed_2$-graded $G$-bundle over $W$ of spinors on $H$. Set
  $\widehat{E}=\pi^*\big((\spinor(H)^+)^\vee\otimes E\big)$ and
  $\widehat{ f}= f\circ\pi$. Then
  $\big(\,\widehat{W}\,,\,\widehat{E}\,,\,\widehat{
    f}~\big)\in\cat{D}^G(X)$ is the \emph{vector bundle modification}
  of $(W,E, f)$ by $H$.
\end{itemize}
We set
$$
\K^G_{0,1}(X)=\cat{D}_{{\rm even},{\rm odd}}^G(X)\,\big/\,\sim
$$
where the parity refers to the dimension of the K-cycle, which is
preserved by $\sim$.

Using the equivariant Dirac class, one can construct a
homomorphism from the geometric to the analytic K-homology group. On
K-cycles we define $(W,E,f)\mapsto f_*\big[\Dirac_E^W\big]$ and extend
linearly. This map can be used to express $G$-index theorems within
this homological framework and it extends to give an isomorphism
between the two equivariant K-homology groups~\cite{Baum2007}. (See
also ref.~\cite{Reis2006} for a related construction in the
non-equivariant case.)

\appendix{D-brane charges of equivariant
  K-cycles\label{Dbranecharge}}

In this appendix we will review the construction of the equivariant
Gysin homomorphism and how it shows that D-brane charges on the
orbifold $[X/G]$ take values in the equivariant K-theory
$\K_G^\bullet(X)$. Let $X$ and $W$ be smooth compact $G$-manifolds,
and $f:W\rightarrow{X}$ a smooth proper $G$-map. We begin by dealing
with the non-equivariant setting $G=e$. Assume that the
$\zed_2$-graded bundle $\nu$ of eq.~(\ref{normbun}) is of even
rank~$r=2n$ and endowed with a \spinc structure. We will generalize
the construction~\cite{Minasian1997,Olsen2000,Witten1998},
establishing that the charge of a D-brane supported on $W$ with
Chan-Paton gauge bundle $E\to W$ in Type~II superstring theory without
$H$-flux takes values in the complex K-theory of spacetime
$X$, to D-branes represented by generic topological K-cycles $(W,E,f)$,
\emph{i.e.}, including those D-branes which are not representable as
wrapping embedded cycles in $X$. It is based on the diagram
$$
\xymatrix{
\nu\cong U~\ar[d]_\pi \ar[rd]^\jmath & \\
W~\ar[d]_\kappa \ar[r]^f & ~X \\
X\times\real^{2q}~\ar[ru]_{\pi_1} & }
$$
over the brane immersion $f$, which we explain momentarily. The
spin$^c$ condition on the bundle $\nu$ is the appropriate
generalization of the Freed-Witten anomaly cancellation
condition~\cite{Freed1999} to this situation. It amounts to a choice
of line bundle $L\to W$ whose first Chern class
$$c_1(L)~\in~\H^2(W;\zed)$$ obeys $c_1(L)\equiv
f^*w_2(T_X)-w_2(T_W)~{\rm mod}~2$, where $w_2(T_X)$ and
$w_2(T_W)$ are the second Stiefel-Whitney classes of the tangent
bundles of $X$ and $W$. The set of all such K-orientations is an
affine space modelled on $2\,\H^2(W;\zed)$.

Consider first the usual case where $f:W\hookrightarrow X$ is a
smoothly embedded cycle. Then the virtual bundle $\nu$ can be
identified (in KO-theory) with the normal bundle to $ W$ with respect
to $f$, which is the quotient bundle $\pi:f^*(T_X)/T_W\to W$. Upon
choosing a riemannian metric on $X$, we can identify $\nu$ with a
tubular neighbourhood $U$ of $f(W)$ via a diffeomorphism from the open
embedding $\jmath:U\hookrightarrow X$ onto a neighbourhood of the zero
section embedding $W\hookrightarrow\nu$. Let $[\pi^*S(\nu)^+,
\pi^*S(\nu)^-; c(v)]$ be the Atiyah-Bott-Shapiro representative
of the Thom class $\Thom(\nu)$, in the $\K$-theory with compact
vertical support $$\K^r_\cpt(\nu):=\K^r(\nu,\nu\setminus W) \ , $$ which
restricts to the Bott class $u^{-n}\in\K^{-r}(\pt)$ on each fibre of
$\nu$. Here $$S(\nu)^\pm~\longrightarrow~W$$ are the half-spinor
bundles associated to $\nu$ and the morphism $c(v) :\pi^*S(\nu)^+ \to
\pi^*S(\nu)^-$ is given by Clifford multiplication by the tautological
section $v$ of the bundle $\pi^*\nu\to\nu$ which assigns to a vector
in $\nu$ the same vector in $\pi^*\nu$.

Then one can define the Gysin homomorphism in ordinary K-theory %%
\begin{displaymath}
f^{\K}_{!}\,:\,\K^\bullet( W)~\longrightarrow~{\K^\bullet(X)} \ .
\end{displaymath}
It is defined as the composition of the Thom isomorphism
\bea
\K^\bullet(W)&\xrightarrow{\approx}&{\K^\bullet_\cpt(\nu)}
\nonumber \\ \xi&\longmapsto&\pi^*(\xi)\otimes\Thom(\nu) \nonumber
\eea
with the natural ``extension by zero'' homomorphism 
$\jmath:\K^\bullet_\cpt(\nu)\to\K^\bullet(X)$ given by composing
$\K^\bullet(U,U\setminus W)\to\K^\bullet(X,X\setminus
W)\to\K^\bullet(X)$, where the first map is the excision isomorphism
and the second map is induced by the inclusion
$(X,\pt)\hookrightarrow(X,W)$. For a general smooth proper map $f:W\to
X$, we use the fact that every smooth compact manifold $W$ can be smoothly
embedded in $\real^{2q}$ for $q$ sufficiently large to define a parametrized
version that yields an embedding $$\kappa\,:\,W~\longrightarrow~
X\times\real^{2q} \ , $$ whose normal bundle is spin$^c$. The
corresponding Gysin map is a homomorphism
$$\kappa_!^\K\,:\,\K^\bullet(W)~\longrightarrow~
\K^\bullet_\cpt\big(X\times\real^{2q}\big) \  . $$ The
Gysin homomorphism $f_!^\K:\K^\bullet(W)\to\K^\bullet(X)$ is then
defined as the composition of $\kappa_!^\K$ with the inverse Thom
isomorphism $\K_{\cpt}^\bullet(X\times\real^{2q})\cong\K^\bullet(X)$
for the trivial \spinc bundle
$$\pi_1\,:\,X\times\real^{2q}~\longrightarrow~X \ . $$ By
homotopy invariance of K-theory and functoriality for pushforward
maps, the map $f_!^\K$ is independent of the choice of identification
of the normal bundle with a tubular neighbourhood and of Whitney
embedding $W\hookrightarrow\real^{2q}$.

Let us now consider the $G$-actions on $ W$ and on $X$. In
a similar way as in ordinary K-theory, if $\nu$ is $\K_{G}$-oriented
then one has the equivariant Thom isomorphism %%
\bea
\K^\bullet_{G}( W)&\xrightarrow{\approx}&{\K^\bullet_{G,\cpt}(\nu)}
\nonumber\\ \xi&\longmapsto& \pi^*(\xi)\otimes\Thom_G(\nu) \ ,
\nonumber
\eea
where the equivariant Thom class $\Thom_G(\nu)\in\K_{G,\cpt}^r(\nu)$
is defined in the same way as above using the $G$-\spinc structure on
$\nu$ and the equivariant version of the Atiyah-Bott-Shapiro
construction~\cite{Landweber2005}. The associated Gysin
homomorphism, constructed as above via a choice of $G$-invariant
riemannian metric on $X$ and of $G$-invariant Whitney embedding
$W\hookrightarrow\real^{2q}$ with $G$ acting trivially on
$\real^{2q}$, is the pushforward map
$f_!^{\K_G}:\K^\bullet_{G}(W)\to\K^\bullet_{G}(X)$. This establishes
that the charge of a fractional D-brane in the Type~II spacetime
orbifold $[X/G]$, associated to a generic $G$-equivariant K-cycle
$(W,E,f)\in\cat{D}^G(X)$ on the covering space $X$, takes values
$f_!^{\K_G}\big([E]\big)\in\K^\bullet_{G}(X)$ in the equivariant
K-theory of $X$.

\end{document}